\documentclass[reqno]{amsart}
\usepackage{amssymb}
\usepackage[dvips]{epsfig}
\usepackage{graphicx}
\usepackage{color}
\usepackage{amsmath}
\usepackage{amssymb}
\usepackage{amsfonts}
\usepackage{amsthm}
\usepackage{tikz}
\usepackage{caption}
\newcommand{\R}{\mathbb R}

\newcommand{\Z}{\mathbb Z}

\renewcommand{\l}{\lambda}

\newcommand{\N}{\mathbb{N}}
\newcommand{\T}{\mathbb{T}}

\newtheorem{thm}{Theorem}[section]
\newtheorem{lem}[thm]{Lemma}
\newtheorem{prop}[thm]{Proposition}
\newtheorem{cor}{\bf Corollary}[section]
\newtheorem{rem}{\bf Remark}[section]
\theoremstyle{definition}

\theoremstyle{statement}

\numberwithin{equation}{section}
\usepackage[colorlinks=true,pdfstartview=FitV,linkcolor=magenta,citecolor=cyan]{hyperref}
\usepackage{bm}

\begin{document}
	\title[$C^2$-Arithmetic Anderson localization]{Localization and  Regularity of the Integrated Density of States  for Schr\"odinger Operators on $\Z^d$ with $C^2$-cosine Like Quasi-periodic  Potential}
	\author[Cao]{Hongyi Cao}
\address[H. Cao] {School of Mathematical Sciences,
Peking University,
Beijing 100871,
China}
\email{chyyy@stu.pku.edu.cn}
\author[Shi]{Yunfeng Shi}
\address[Y. Shi] {School of Mathematics,
Sichuan University,
Chengdu 610064,
China}
\email{yunfengshi@scu.edu.cn}

\author[Zhang]{Zhifei Zhang}
\address[Z. Zhang] {School of Mathematical Sciences,
Peking University,
Beijing 100871,
China}
\email{zfzhang@math.pku.edu.cn}

\date{\today}


\begin{abstract}
In this paper,  we study the multidimensional  lattice   Schr\"odinger operators with $C^2$-cosine like quasi-periodic (QP) potential.   We  establish quantitative Green's function estimates,  the  arithmetic version of Anderson (and dynamical) localization,  and the  finite volume version of $(\frac 12-)$-H\"older continuity of the integrated density of states (IDS) for  such QP Schr\"odinger operators. Our proof is based on an extension of the fundamental  multi-scale analysis (MSA) type method of  Fr\"ohlich-Spencer-Wittwer [\textit{Comm. Math. Phys.} 132 (1990):  5--25]  to the higher lattice dimensions.  We resolve the  level crossing issue on eigenvalues parameterizations in the case of both higher lattice dimension and $C^2$ regular potential. 
 \end{abstract}

\maketitle

\maketitle
	
	\section{Introduction and main results}

In this paper,  
we  are concerned with  the QP Schr\"odinger operator 
\begin{align}\label{model}
	H(\theta)=\varepsilon \Delta+v(\theta+ x\cdot{\omega})\delta_{x, y},\ x\in\Z^d,
\end{align}
where $\varepsilon\geq0$ and  the discrete Laplacian $\Delta$ is defined as
\begin{align*}
	\Delta(x, y)=\delta_{{\|x- y\|_{1}, 1}},\ \| x\|_{1}:=\sum_{i=1}^{d}\left|x_{i}\right|.
\end{align*}
For the diagonal part of \eqref{model}, we let $\theta\in \mathbb{T}=\mathbb{R}/\mathbb{Z},   \omega\in 	{\rm DC}_{\tau, \gamma}$ and $x\cdot\omega=\sum\limits_{i=1}^dx_i \omega_i$, 	with	$$	{\rm DC}_{\tau, \gamma}=\left\{\omega\in[0,1]^d:\ \|x\cdot\omega\|=\inf_{l\in\mathbb{Z}}|l- x\cdot\omega|\geq \frac{\gamma}{\| x\|_1^{\tau}} \ {\rm for}\ \forall\  x\in\mathbb{Z}^d\setminus\{0\}\right\},$$ 
where $\tau>d,\gamma>0.$ We call $\theta$ the phase  and $\omega$ the frequency. 
We  further assume  that the  potential $v\in C^2(\mathbb{T}; \R)$ is an \textit{even} function with \textit{exactly two non-degenerate critical points} \footnote{Without loss of generality, we  assume that $\theta=0$ is the maxima point and $\theta=1/2$ is the minima one for  $v$. Since we are considering  small $\varepsilon$, we  further assume that there exists $0<a<1/10$, such that $|v''(\theta)|>3$ for $\theta\in\{\theta\in \mathbb{T}:\ \|\theta\|<a\}\cup\{\theta\in \mathbb{T}:\ \|\theta-1/2\|<a\}$,  and $|v'(\theta)|>3$ for  $\theta\in\{\theta\in \mathbb{T}:\ \|\theta\|\geq a\}\cap\{\theta\in \mathbb{T}:\ \|\theta-1/2\|\geq a\}$. Under these assumptions, we denote $$M_1=\sup_{\theta\in \T}\max(|v(\theta)|,|v'(\theta)|,|v''(\theta)|)>0.$$}\label{ftnote1}.  
 The special case of $d=1$ and $v=\cos2\pi\theta$ corresponds to the famous almost Mathieu operator  (AMO). The main goals of the present work are as follows. 
\begin{itemize}
\item We first extend the celebrated  multi-scale analysis (MSA) type method of Fr\"ohlich-Spencer-Wittwer \cite{FSW90} to the higher lattice dimensions. In particular, we establish the quantitative Green's function estimates for \eqref{model}. 
\item Based on the quantitative Green's function estimates, we prove the arithmetic version of Anderson (and dynamical) localization in the perturbative regime. 
\item We  prove the finite volume version of the $(\frac12-)$-H\"older continuity of the IDS.
\end{itemize}

Our main motivations come  from extending some fine properties obtained for AMO  to the general QP Schr\"odinger operators.  
In particular,  we are interested in   the   Anderson localization (i.e., pure point spectrum with exponentially decaying eigenfunctions).   
Actually, since the fundamental works of Sinai \cite{Sin87} and Fr\"ohlich-Spencer-Wittwer \cite{FSW90}, the Anderson localization  has been  obtained  for the $1D$ QP Schr\"odinger operators with $C^2$-cosine like potentials or even more general Gevrey  potentials  \cite{Eli97} assuming Diophantine frequencies.  
 However, all these $1D$  results are perturbative in the sense that the required perturbation strength depends  on the Diophantine frequency  (i.e., localization holds for $|\varepsilon|\leq \varepsilon_0(v,\omega)$).   Then Jitomirskaya made a  breakthrough in \cite{Jit94, Jit99}, where the non-perturbative method  for control of  Green's functions (cf. \cite{Jit02})  was developed first  for  AMO.  This will allow effective (even optimal in many cases) and independent of $\omega$ estimate on $\varepsilon_0$. In addition, applying this  method can prove the \textit{arithmetic version of Anderson localization} for AMO which means the removed sets on both  $\omega$ and $\theta$ when establishing localization have an explicit arithmetic description (cf. \cite{Jit99,JL18} for details).  The non-perturbative method of Jitomirskaya \cite{Jit99} was later  extended by Bourgain-Goldstein \cite{BG00} to the case of general analytic potentials.  However,  the localization results of \cite{BG00}  hold for arbitrary $\theta\in\T$ and a.e. Diophantine frequencies (the permitted set of frequencies depends  on $\theta$).  So, there seems  no  arithmetic  version of Anderson localization  result  for general analytic QP Schr\"odinger operators even in the $1D$ case. Recently,  the evenness condition of \cite{FSW90} on the potential was removed in \cite{FV21} in the $1D$ case.  We also mention the work \cite{GYZ} in which the arithmetic version of the Anderson localization was proved for $1D$ quasi-periodic Schr\"odinger operators with a $C^2$-cosine like potential via the reducibility method. 
  
 It is well-known that the non-perturbative localization is not expected  for QP operators on $\Z^d$ for $d\geq 2$ (cf. \cite{Bou02}). In the multidimensional case, Chulaevsky-Dinaburg \cite{CD93} and Dinaburg \cite{Din97}  first extended  results of Sinai \cite{Sin87} to the exponential long-range QP operators with  $C^2$ regular potentials on $\Z^d$ for arbitrary $d\geq1.$   However,  while the localization results of \cite{CD93,Din97} allow any Diophantine frequencies,  there is simply no explicit arithmetic description on the $\theta$.
Later, the remarkable work of Bourgain-Goldstein-Schlag \cite{BGS02}  established the Anderson localization for  general analytic QP Schr\"odinger operators on $\Z^2$ via Green's function estimates. In 2007,   Bourgain \cite{Bou07} successfully extended the  results of  \cite{BGS02} to arbitrary dimensions. The results of \cite{Bou07} have been largely generalized by Jitomirskaya-Liu-Shi  \cite{JLS20} to  the  case of both arbitrarily  dimensional multi-frequencies and exponential long-range hopping.
 We want to remark that the localization results of \cite{BGS02,Bou07,JLS20} are  non-arithmetic. 
 Very recently, Ge-You \cite{GY20}  applied a reducibility argument (based on ideas of \cite{JK16,AYZ17}) to  the multidimensional  long-range QP operators with the cosine potential, and proved the arithmetic version of Anderson localization. The authors \cite{CSZ22} also provided an alternative proof (based on Green's function estimates) of the arithmetic Anderson localization.
 
 To the best of our knowledge,  there is simply no arithmetic version of  Anderson localization result  for QP Schr\"odinger operators on $\Z^d$ ($d\geq 2$) with the potential beyond the cosine function.  This is one of our main motivations of the present work.  For this, we first establish the quantitative Green's function estimates, which is based on the MSA type method of \cite{FSW90}.  Occasionally, by combining the Green's function estimates  with an argument of Bourgain \cite{Bou00}, we can also obtain the finite volume version of the  $(\frac12-)$-H\"older continuity of the IDS.  However, to extend the method of \cite{FSW90}  to work in the higher lattice dimensions, we have to deal with the essential difficulty of the {\it level crossing}  on eigenvalues parameterizations. This motivates us to take full advantage of  the deep results of  Rellich \cite{Rel69} and Kato \cite{Kat95} concerning the $C^1$  eigenvalues variations. In addition,  to handle the resonances using MSA, it requires to overcome the difficulty of the {\it non-interval} structure of  the resonant blocks, which is  accomplished via the method developed previously by the authors in  \cite{CSZ22}.

\subsection{Main results}
In this section, we will introduce our main results.
\subsubsection{Quantitative Green's function estimates}
We begin with the quantitative Green's function estimates.

Let $\Lambda\subset\Z^d$, $E\in\R$  and $\theta\in\T$. The Green's function $G_\Lambda(\theta;E)$ is defined by 
$$G_\Lambda(\theta;E)=(H_\Lambda(\theta)-E)^{-1},$$
where $H_\Lambda(\theta)=R_\Lambda H(\theta)R_\Lambda$ with $R_\Lambda$ being the restriction operator.  We also write
$$G_\Lambda(\theta;E)(x,y)=\left\langle\delta_x, G_\Lambda(\theta;E)\delta_y \right\rangle,$$
where $\langle \cdot,\cdot\rangle$ denotes the standard inner product on $\ell^2(\Lambda).$

Let $0< \varepsilon\leq\varepsilon_0$, where $\varepsilon_0$ is sufficiently small depending on $ v,d, \tau,\gamma$. Fix $E^*\in\R, \theta^*\in\T$ and $ \delta_0=\varepsilon_0^{1/20}$. Define the $0$-th  generation of singular points  set
$$Q_0=\big\{c_0^i\in \Z^d:\ |v(\theta^*+ c_0^i\cdot{\omega})-E^*|<\delta_0\big\}.$$
For $n\geq1$, we  inductively define the  family of $l_n$-size (i.e., diameter) blocks $\{B_n^i\}_{c_n^i\in P_n}$, where $l_1= |\log\varepsilon_0|^2$ or $|\log\varepsilon_0|^4$, $l_{n+1}= l_{n}^2$ or $l_{n}^4$  (each $B_n^i$ is centered at $c_n^i$).  These blocks  are  used to  cover  the $(n-1)$-th generation of singular points set $Q_{n-1}$. 
 We also define  the $n$-th generation of  singular points set  (resp. singular blocks)
 $$Q_n=\big\{c_n^i\in P_n:\ \operatorname{dist}(\sigma (H_{B_n^i}(\theta^*)),E^*)<\delta_n:=e^{-l_n^{2/3}}\big\}\ ({\rm resp. }\ \{B_n^i\}_{c_n^i\in Q_n}), $$
where $\sigma(\cdot)$ denotes the spectrum of some operator.  
The non-singular blocks $\{B_n^i\}_{c_n^i\in P_n\setminus Q_n}$ are  $n$-regular.  An arbitrary finite  set $\Lambda\subset\Z^d$ is  $n$-good  if every point  of $\Lambda\cap Q_0$ is contained 
in an $m$-regular block $B_m^i\subset\Lambda$ for some $m\leq n$.
\begin{thm}\label{mainthm}
	Let $\omega\in{\rm DC}_{\tau,\gamma}$. Then there exists some  $\varepsilon_0=\varepsilon_0(v,d,\tau,\gamma)>0$, such that for all $0< \varepsilon\leq \varepsilon_0$, the following two statements hold true.
	\begin{itemize}
		\item {\rm ({\bf Green's function estimates})} If $\Lambda$ is $n$-good, then the estimates 
		\begin{align*}
		\|G_\Lambda(\theta;E)\|&\leq10\delta_n^{-1},\\
		|G_\Lambda(\theta;E)(x,y)|&\leq e^{-\gamma_n\|x-y\|_1}\ {\rm for}\ \|x-y\|_1\geq l_n^\frac{5}{6}\ (l_0=1)
		\end{align*}
		hold for all $|\theta-\theta^*|<\delta_n/(10M_1)$ and $ |E-E^*|<\delta_n/5$. Moreover, we have $$\gamma_n\searrow \gamma_\infty\geq\gamma_0/2=|\log\varepsilon|/4>0.$$	
		\item{\rm ({\bf Center Theorem})} If $c_n^i,c_n^j\in Q_n $, then $$m(c_n^i,c_n^j)\leq2\delta_n^{1/2},$$
		where $$m(c_n^i,c_n^j):=\min(\|(c_n^i-c_n^j)\cdot \omega\|,\|2\theta^*+(c_n^i+c_n^j)\cdot \omega\|).$$
	\end{itemize} 
\end{thm}
\begin{rem}
For a more complete description on the Green's function estimates, we refer to \S \ref{sec2}.  In contrast,  we can not identify the conditions  of being a center of the single resonant block  as in \cite{CSZ22}, but only provide conditions on centers being a pair of  resonant blocks.  This is reasonable since we have low regular $C^2$ potentials. 
\end{rem}

\subsubsection{Arithmetic version of localization}
In this part, we will state our arithmetic version of localization results.  

We first introduce our Anderson localization result. 
\begin{thm}\label{AL}
	Let $H(\theta)$ be given by \eqref{model} and let $\omega\in \operatorname{DC}_{\tau, \gamma}$. Then there exists some $\varepsilon_0=\varepsilon_0(v, d,\tau, \gamma)>0$ such that, for all $0<\varepsilon\leq\varepsilon_0$ and $\theta \in\mathbb{T}\setminus\Theta$,  $H(\theta)$ satisfies the Anderson localization, where 
	$$\Theta=\{\theta\in \T :\  {\rm the \  relation}\  \|2\theta+x \cdot\omega\|\leq \|x\|_1^{-d-2}\ {\rm holds\  for\  infinitely\  many }\   x\in\Z^d\}.$$
\end{thm}
\begin{rem}
We prove the first arithmetic version of Anderson localization for QP Schr\"odinger operators on $\Z^d$ with $C^2$ regular potentials.  The reducibility type method seems invalid  in our  case of both higher lattice dimensions and $C^2$ regular potential. Our result can be easily extended to the exponential long-range  QP operators. 
\end{rem}

We then state  our dynamical localization result.

\begin{thm}\label{t1}	
Let $H(\theta)$ be given by \eqref{model} and let $\omega\in \operatorname{DC}_{\tau, \gamma}$. Then there exists some $\varepsilon_0=\varepsilon_0(v, d,\tau, \gamma)>0$ such that, for all $0<\varepsilon\leq\varepsilon_0$, the following statement holds true.  Denote for $A>0,$  
	\begin{align}\label{D1}
		\Theta_A=\left\{\theta \in \mathbb{T}:\  \left\| 2 \theta+x\cdot \omega \right\|>\frac{A}{\|x\|_1^{d+1}} \text{  {\rm for} $ x\in \mathbb{Z}^d \setminus\{0\}$} \right\}.
	\end{align}
	Then for any $A>0$, $\theta\in \Theta_A$ and  $q>0$, we have 
	\begin{align}
		\nonumber&\ \ \ \sup_{t\in \mathbb{R}}\sum_{x\in \mathbb{Z}^d}(1+\|x\|_1)^q|\langle e^{itH(\theta)}{\bm e}_0, {\bm e}_x\rangle|\\
		\label{jielun} &\leq C({q,d})\max\left(|\log \min(A,1)|^{12(q+2d)},|\log\varepsilon_0|^{12(q+2d)}\right),
	\end{align}
	where $\{\bm e_x\}_{x\in\Z^d}$ denotes the standard basis of $\ell^2(\Z^d)$ and $C({q,d})>0$ depends only on $q, d$.
	Moreover,  we have
	\begin{align*}
		\int_\mathbb{T}\sup_{t\in \mathbb{R}}\sum_{x\in \mathbb{Z}^d}(1+\|x\|_1)^q|\langle e^{itH(\theta)}{\bm e}_0, {\bm e}_x\rangle|d\theta<+\infty.
		\end{align*}
	\end{thm}
\begin{rem}
We note that $\Theta_A\subset \T\setminus\Theta$ for all $A>0$, where $\Theta$ is defined in Theorem \ref{AL}. Our result gives the arithmetic description on $\theta$ at which the dynamical localization holds true. For recent progress on dynamical localization for  the multidimensional QP operators assuming Diophantine frequencies, we refer to \cite{GYZ19}. 
\end{rem}

\subsubsection{H\"older continuity of the IDS}
In this part, we introduce our result concerning regularity of  the IDS. 

For a finite set $\Lambda$, denote by $\#\Lambda$ the cardinality of $\Lambda.$ Let
$$\mathcal{N}_{\Lambda}(E;\theta)=\frac{1}{\#\Lambda}\#\{\lambda\in\sigma({H_\Lambda(\theta)}):\  \lambda\leq E\}$$
and denote by
\begin{align}\label{ids}
	\mathcal{N}(E)=\lim_{N\to\infty}\mathcal{N}_{\Lambda_N}(E;\theta)
\end{align}
the IDS, where $\Lambda_N=\{x\in\Z^d:\ \|x\|_1\leq N\}$ for $N>0$. It is well-known that the limit in  \eqref{ids} exists and is independent of $\theta$  for a.e. $\theta$.

\begin{thm}\label{thm2}
	Let $H(\theta)$ be given by \eqref{model} and let $\omega\in \operatorname{DC}_{\tau, \gamma}$. Then there exists some $\varepsilon_0=\varepsilon_0(v, d,\tau,\gamma)>0$ such that,  for all  $\eta>0$ and for sufficiently large $N$ (depending on $\eta$), we have
	\begin{align}
		\nonumber&\ \ \ \sup_{\theta^*\in\T, E^*\in\R}\left(\mathcal{N}_{\Lambda_N}(E^*+\eta;\theta^*)-\mathcal{N}_{\Lambda_N}(E^*-\eta;\theta^*)\right)\\
		\label{hold}&\leq C(d)\eta^{\frac{1}{2}}\max(1,|\log\eta|^{8d}),
	\end{align}	
	where $C(d)>0$ depends only on $d$. In particular, the IDS is  $(\frac12-)$-H\"older continuous, i.e., for all $\eta>0,$
	$$\mathcal{N}(E+\eta)-\mathcal{N}(E-\eta)\leq C(d)\eta^{\frac{1}{2}}\max(1,|\log\eta|^{8d}).$$
\end{thm}
\begin{rem}
Indeed, we obtain the  quantitative estimate on the regularity of the IDS beyond the $(\frac12-)$-one. Our result also improves the upper bound  on the number of eigenvalues  of Schlag (cf. Proposition 2.2 of \cite{Sch01})  in the special case that the potential is given by the $C^2$-cosine like function.  In our case, since the Aubry duality method might not work, it is unclear wether or not the optimal $\frac12$-H\"older continuity  of the IDS for our model remains true. 
\end{rem}
 
The study of the regularity of the IDS for QP operators  has  attracted great attention over the years.  In \cite{GS01}, Goldstein-Schlag first proved the H\"older continuity of the IDS  for general $1D$ and one-frequency analytic QP Schr\"odinger operators in the regime of positive Lyapunov exponent, but provided no explicit information on the H\"older exponent. 
 In \cite{Bou00}, Bourgain developed a method  based on Green's function estimates to obtain the first  finite volume version of  $(\frac12-)$-H\"older continuity of the IDS for AMO in the perturbative regime. In 2009,  by using KAM reducibility method of Eliasson \cite{Eli92}, Amor \cite{Amo09} obtained the first  $\frac12$-H\"older continuity result of the IDS for $1D$ and multi-frequency QP Schr\"odinger operators with small analytic potentials and Diophantine frequencies.  Later,  the one-frequency result of Amor was essentially  generalized by Avila-Jitomirskaya \cite{AJ10} to the non-perturbative case via the quantitative almost reducibility and localization method.  In \cite{GS08} and in the regime of positive Laypunov exponent,   Goldstein-Schlag  proved the $(\frac{1}{2m}-)$-H\"older continuity of the IDS for $1D$ and one-frequency QP Schr\"odinger operators with  potentials given by analytic perturbations of certain trigonometric polynomials of degree $m\ge1$. This work provides in fact the  finite volume version of  estimates on the IDS.  We  remark that the  H\"older continuity of the IDS for  $1D$ and multi-frequency QP Schr\"odinger operators with large general  potentials is hard to prove. In \cite{GS01}, Goldstein-Schlag obtained the weak H\"older continuity \footnote{i.e, the  estimate  \begin{align}\label{weakids}|\mathcal{N}(E)-\mathcal{N}(E')|\leq e^{-\left(\log\frac{1}{|E-E'|}\right)^\zeta},\ \zeta\in(0,1).\end{align}}
 of  the IDS for $1D$ and multi-frequency QP Schr\"odinger operators assuming the positivity of the Lyapunov exponent and strong Diophantine frequencies. The weak H\"older continuity of the IDS for the multidimensional  QP Schr\"odinger operators  has been  established  by  Schlag \cite{Sch01}, Bourgain \cite{Bou07} and Liu \cite{Liu20}. Ge-You-Zhao \cite{GYZ22} proved the  $(\frac{1}{2m}-)$-H\"older continuity of  the  IDS for the  multidimensional QP  Schr\"odinger operators with  small exponential long-range hopping and trigonometric polynomial (of degree $m$) potentials via the reducibility argument.  By Aubry duality, they can obtain the  $(\frac{1}{2m}-)$-H\"older continuity of the IDS for $1D$ and multi-frequency QP operators with a  finite range hopping.   Recently, the work \cite{XGW} established the $\frac12$-H\"older continuity of the IDS for some $1D$ quasi-periodic Schr\"odinger operator with cosine like potential. Very recently,   the authors \cite{CSZ22} proved the finite volume version of $(\frac12-)$-H\"older continuity of the IDS for QP Schr\"odinger operators on $\Z^d$ with the cosine potential.  In the present, we extend the work \cite{CSZ22} to the case of $C^2$-regular potentials.

\subsection{The strategy of the proof and comparison with previous works}
The key ingredient of our proof is the quantitative Green's function estimates. Once such estimates were obtained, the proof of both the arithmetic version of localization and the finite volume version of the $(\frac12-)$-H\"older continuity of the IDS just follows in a standard way.  To deal with Green's function estimates, we will apply  the MSA type method of Fr\"ohlich-Spencer-Wittwer \cite{FSW90}.  However,  in higher lattice dimensions case, there comes essential difficulties not appeared in \cite{FSW90}. This definitely requires  a   proof with new ideas,  which will be explained below. 
\subsubsection{The level crossing issue}
 The first issue is about the level crossing of eigenvalues parameterizations  in the present case.  More precisely, by the definition of the singular site of the $n$-th step, for $c_n^i\in Q_n$,  there is some $E_n^i(\theta^*)$ so that 
\begin{align}\label{rescd}
{\rm dist}(\sigma(H_{B_n^i}(\theta^*)), E_n^i(\theta^*))\leq \delta_n,
\end{align}
where $B_n^i$ is a resonant block centered at $c_n^i.$   We assume further 
\begin{align}\label{sncd}
s_n=\inf_{c_n^i\neq c_n^j\in Q_n}\|c_n^i-c_n^j\|_1\geq 10 l_n^2.
\end{align} 
Our main goal here is to establish {\bf Center Theorem} at the $(n+1)$-th step.  From \eqref{sncd}, we can define the $(n+1)$-th generation of resonant blocks $\{B_{n+1}^i\}_{c_{n+1}^i\in Q_{n+1}}$ with  ${\rm diam}(B_{n+1}^i)= l_{n+1}\sim l_n^2$ and $c_{n+1}^i=c_n^i.$  By \eqref{rescd}, we can distinguish two cases.\smallskip\\
{\bf Case 1. }  ${\rm dist}(\sigma(H_{B_n^i}(\theta^*))\setminus\{E_n^i(\theta^*)\}, E_n^i(\theta^*))> \delta_n.$ This case is similar to that in \cite{FSW90} without level crossing.  Precisely, in this case, we can show that for every $\theta\in(\theta^*-\delta_n/(10M_1), \theta^*+\delta_n/(10M_1))$,  $H_{B_{n+1}^i}(\theta)$ has a unique eigenvalue $E_{n+1}^i(\theta)$ so that $|E_{n+1}^i(\theta)-E^*|<\delta_n/9,$  where the function $E_{n+1}(\theta)$ is called an eigenvalue parameterization.  Moreover, we can prove  the lower bound $|\frac{d^2E_{n+1}(\theta) }{d\theta^2}|\geq 2$ when $|\frac{d E_{n+1}(\theta) }{d\theta}|$ is small. This combined with the uniqueness of $E_{n+1}^i(\theta)$,  the evenness of $v$ and  
the symmetrical  property of $B_{n+1}^i$ leads to a proof of the {\bf Center Theorem}, i.e., $m(c_{n+1}^i, c_{n+1}^j)\leq 2\delta_{n+1}^{\frac12}.$ In this case,  our proof is similar to that in \cite{FSW90} and contains no essential new ideas. \\
{\bf Case 2. }  ${\rm dist}(\sigma(H_{B_n^i}(\theta^*))\setminus\{E_n^i(\theta^*)\}, E_n^i(\theta^*))\leq  \delta_n.$  This case is not appeared in \cite{FSW90}, since there is no  priori lower bound on differences    of eigenvalues  (cf. Lemma 4.1 in \cite{FSW90}) of $H_{B_{n+1}^i}(\theta^*)$ for $d\geq2$. This situation has also been encountered by Surace \cite{Sur90} in the study of the localization for 
$$\tilde H(K)=\varepsilon \Delta+(K+x_1+x_2\alpha)^2\delta_{x,y},\ K\in\R, \ x=(x_1,x_2)\in\Z^2.$$
Relying on some ideas of Surace \cite{Sur90}, we can show in this case the following: For $\theta\in(\theta^*-10\delta_n^{\frac12}, \theta^*+10\delta_n^{\frac12}),$ there are exactly  two eigenvalues $E_{n+1}^i(\theta)$ and $\mathcal{E}_{n+1}^i(\theta)$ in the energy interval $(E^*-50M_1\delta_n^{\frac12}, E^*+50M_1\delta_n^{\frac12}).$ Then it is inevitable that there may be some $\theta_1 \in(\theta^*-10\delta_n^{\frac12}, \theta^*+10\delta_n^{\frac12})$ with ${E}_{n+1}^i(\theta_1)=\mathcal{E}_{n+1}^i(\theta_1)$, namely,  the level crossing appears.  Fortunately,  we can show the number of  level crossing points   in $(\theta^*-10\delta_n^{\frac12}, \theta^*+10\delta_n^{\frac12})$ is at most $1$ and $\theta_1=\theta_{n+1}^i:=-c_n^i\cdot\omega+\mu_n\mod 1$ ($\mu_n=0$ or $\mu_n=1/2$) whenever $\theta_1$ is a level crossing point.  In addition, if  ${E}_{n+1}^i(\theta_{n+1}^i)\neq \mathcal{E}_{n+1}^i(\theta_{n+1}^i)$, then ${E}_{n+1}^i(\theta)\neq \mathcal{E}_{n+1}^i(\theta)$ for all $\theta\in (\theta^*-10\delta_n^{\frac12}, \theta^*+10\delta_n^{\frac12})$,  and this case reduces to that in \cite{FSW90}.   So, the remaining case is  ${E}_{n+1}^i(\theta_{n+1}^i)=\mathcal{E}_{n+1}^i(\theta_{n+1}^i).$  
For this similar case  in Surace \cite{Sur90}, since $\tilde H(K)$ is analytic in $K$, the analytic version of the Rellich's  theorem  (cf. \cite{Kat95}) can ensure that both $E_{n+1}^i(K)$ and $\mathcal{E}_{n+1}^i(K)$ are analytic in $K$ even though the level crossing occurs. More importantly, the corresponding normalized eigenfunctions associated with $E_{n+1}^i(K)$ and $\mathcal{E}_{n+1}^i(K)$ can also be analytic in $K$.  Based on these analyticity properties, Surace \cite{Sur90} showed by taking derivatives on eigenvalues and eigenfunctions  that both $|\frac{d E_{n+1}^i(K) }{dK}|$ and $|\frac{d \mathcal{E}_{n+1}^i(K) }{dK}|$ have good lower bounds. Then the {\bf Center Theorem} follows.  \textit{Obviously, the method of Surace \cite{Sur90} relies essentially on the smoothness of both eigenvalues and eigenfunctions parameterizations in dealing with the level crossing issue}. Returning to our case, since we have only the $C^2$ regularity of $H(\theta)$ in $\theta$, the level crossing  in this case will destroy the smoothness  of eigenfunctions parameterizations.  To overcome this difficulty, we first employ a more deeper theorem (cf. \cite{Rel69} and also Theorem 6.8 of  \cite{Kat95}) of Rellich, i.e., the $C^1$ version of eigenvalues parameterizations. This remarkable theorem  suggests that one can always  ensure the $C^1$ smoothness (in $\theta$) of ${E}_{n+1}^i(\theta)$ and $\mathcal{E}_{n+1}^i(\theta)$ assuming $H(\theta)$ being $C^1$ (in $\theta$) in some interval.  
Then we  introduce  a  theorem of Kato (cf. Theorem 5.4 in \cite{Kat95}) that can provide  the first order derivatives representations of the $C^1$ eigenvalues parameterizations at some fixed point involving $\frac{ dH(\theta)}{d\theta}$,  but without knowing any smoothness information on the eigenfunctions.   After introducing these two celebrated theorems, we can handle the level crossing issue in the present case. 

\subsubsection{The geometric descriptions of the resonant blocks}
The geometric properties of the resonant blocks $B_{n+1}^i$ play an essential role  in both  the eigenvalues parameterizations analysis and Green's function estimates applying the resolvent identity.  Particularly, we will require $B_{n+1}^i$ to satisfy the following conditions: (i) For any $m\leq n$, if $B_m^j\cap B_{n+1}^i\neq \emptyset$, then $B_m^j\subset B_{n+1}^i$; (ii) Each $B_{n+1}^i$ is translation invariant, i.e., $B_{n+1}^i-c_{n+1}^i$ is independent of $i$; (iii) Each $B_{n+1}^i$ is symmetric about its center $c_{n+1}^i$, i.e., $x\in B_{n+1}^i$ iff $2c_{n+1}^i-x\in B_{n+1}^i.$ In the $1D$ case, the geometric shape of $B_{n+1}^i$ is simple and is given by the  interval. However, in the higher dimensions, the geometric shape of $B_{n+1}^i$ becomes significantly complicated and the interval structure is missing. In fact, it is highly nontrivial to construct $B_{n+1}^i$ satisfying all the properties (i)--(iii) in higher lattice dimensions. While such issue was also appeared in \cite{Sur90}, the author just outlined a possible way  of  achieving the desired constructions, which definitely  restricts to  the  $\Z^2$ lattice.  In the present, we completely resolve this issue by using ideas originated from \cite{CSZ22}. 



\subsection{Organization of the paper}
The paper is organized as follows. Some basic properties on the potentials are introduced in \S \ref{sec1}. The center part of this paper, namely, the Green's function estimates are  presented in \S \ref{sec2}. In \S \ref{sec3}--\S \ref{sec5}, we finish the proof of Theorem \ref{AL}, \ref{t1}. \ref{thm2}, respectively. Some important facts  are collected in the Appendixes. 

\section{Preliminaries}\label{sec1}
In this section, we will introduce some useful lemmas concerning the properties of the potential $v(\theta)$ and $C^1$ eigenvalue variations.  

\begin{lem}[$C^2$-smoothness without the level crossing, \cite{Kat95}]\label{neqs}
	Let $\Lambda$ be a finite set. Assume that $\tilde{E}$ is a simple eigenvalue of $ H_\Lambda(\theta^*)$.  Then there exist a small interval $I$ including $\theta^*$ and a $C^2$ function $E(\theta)$ satisfying {\rm (1)} $E(\theta^*)=\tilde{E}$;  {\rm (2)} For $\theta\in I$, $E(\theta)$ is the unique eigenvalue of $ H_\Lambda(\theta)$  near  $\tilde{E}$. Moreover, the corresponding normalized eigenfunction   $\psi(\theta)$ is also  $C^2$ regular.  
\end{lem}
\begin{proof}
	Note that $f(E, \theta)=\det(E-H_\Lambda(\theta))$ is  a polynomial of $E$ whose  coefficients are  $C^2$ regular in $\theta$. Moreover, $\frac{\partial f}{\partial E}( \tilde{E},\theta^*)\neq0$ since $\tilde{E}$ is simple. The $C^2$ smoothness of $E(\theta)$ follows from the implicit function theorem immediately.  The smoothness of eigenfunction follows from 
	$$\psi(\theta)=\frac{P(\theta)\psi(\theta^*)}{\|P(\theta)\psi(\theta^*)\|},$$
	where $P(\theta)=\int_{\Gamma}(\xi-H_\Lambda(\theta))^{-1}d\xi$ is the $C^2$ projection onto the eigenspace (here $\Gamma$ is a circle enclosing $\tilde{E}$ such that any  other eigenvalues  are outside of $\Gamma$).    
\end{proof}
\begin{rem}
	Since we are working on higher dimensions, the level crossing of eigenvalues parameterizations  may happen. In general, we can not confirm the smoothness of eigenvalues and eigenfunctions parameterizations when $\tilde{E}$ is not a simple eigenvalue.  
\end{rem}

\noindent{\bf Note}. For convenience, we assume that all the eigenfunctions in this paper  are normalized.\smallskip

We then investigate properties of $v(\theta)$ which are important to the proof of {\bf Center Theorem} in the  initial steps. 
\begin{lem}\label{a}For every $\theta_1,\theta_2\in \R$, we have 
\begin{equation}\label{wh}
		|v(\theta_1)-v(\theta_2)|\geq\min(\|\theta_1-\theta_2\|,\|\theta_1+\theta_2\|)^2.
\end{equation}
\end{lem}
\begin{proof}
	Since $v$ is even and $1$-period, it suffices to consider the case  $\theta_1,\theta_2\in [0,\frac{1}{2}]$. Without loss of generality, we assume $\theta_1<\theta_2$. By our assumption (cf. Footnote \ref{ftnote1}), $v$ is strictly decreasing on $[0,\frac{1}{2}]$ satisfying  $v'(\theta)<-2$ for $\theta\in[a,\frac{1}{2}-a]$ and  $v''(\theta)<-2$ (resp. $>2$) for $\theta\in [0,a]$ (resp. $[\frac{1}{2}-a,\frac{1}{2}]$).\smallskip\\
	{\it Case }1. $0\leq\theta_1<\theta_2\leq a$. We have  in this case $$v(\theta_2)-v(\theta_1)=v'(\theta_1)(\theta_2-\theta_1)+\frac{1}{2}v''(\xi)(\theta_2-\theta_1)^2\leq-(\theta_2-\theta_1)^2.$$\\
	{\it Case }2. $0\leq\theta_1\leq a\leq\theta_2\leq\frac{1}{2}-a$. We have in this case
	\begin{align*}
		v(\theta_2)-v(\theta_1)&=(v(\theta_2)-v(a))+(v(a)-v(\theta_1))\\
		&\leq-2(\theta_2-a)-(a-\theta_1)^2\\&\leq-(\theta_2-\theta_1)^2.
	\end{align*}\\
	{\it Case} 3. $a\leq\theta_1<\theta_2\leq\frac{1}{2}-a$. We have 
	$$	v(\theta_2)-v(\theta_1)=v'(\xi)(\theta_2-\theta_1)\leq-(\theta_2-\theta_1)^2.$$\\
	{\it Case} 4. $\theta_1\leq a<\frac{1}{2}-a\leq\theta_2$. We have 
	$$	v(\theta_2)-v(\theta_1)\leq v(\frac{1}{2}-a)-v(a)\leq-(\frac{1}{2}-2a)<-\frac{1}{4}\leq-(\theta_2-\theta_1)^2.$$\\
	{\it Case} 5. $\frac{1}{2}-a\leq\theta_1<\theta_2\leq\frac{1}{2}$. This case  is similar to {\it Case} 1.\\
	{\it Case} 6. $a\leq\theta_1\leq \frac{1}{2}-a\leq\theta_2\leq\frac{1}{2}$. This case  is similar to {\it Case} 2.
\end{proof}

\begin{lem}\label{'0}
	For any $\theta\in \R$, we have 	$|v'(\theta)|\geq2\min(\|\theta\|,\|\theta-\frac{1}{2}\|)$.\end{lem}
\begin{proof}
	It again suffices to consider $\theta\in [0,\frac{1}{2}].$ If $\theta\in [a,\frac{1}{2}-a]$, we have $|v'(\theta)|>2$. If $\theta\in [0,a]$, we have $|v'(\theta)|=|v'(\theta)-v'(0)|\geq |v''(\xi)(\theta-0)|\geq2|\theta|$. Similarly, if $\theta\in [\frac{1}{2}-a,\frac{1}{2}]$, we have $|v'(\theta)|\geq2|\theta-\frac{1}{2}|$.
\end{proof}

\section{ quantitative Green's function estimates}\label{sec2}
In this section, we prove Theorem \ref{mainthm}, i.e., the quantitative Green's function estimates. The proof is based on a MSA type  iteration method of Fr\"ohlich-Spencer-Wittwer \cite{FSW90}.  The $0$-th step of the iteration uses the Neumann series argument and properties of $v$. In the first iteration step, the level crossing issue  has already arisen, and we apply the eigenvalue variations methods of Rellich and Kato to resolve the issue. We want to remark that in the first step,  the   resonant blocks  are simply given by  cubes.  The central part of the proof is definitely the general  iteration steps,  and we design a delicate inductive scheme to handle the level crossing issue. In the general iteration steps, the structure of resonant blocks becomes significantly complicated, since we have to take account of all previous resonant blocks of different sizes.

The following  subsections are devoted to dealing with the $0$-th, $1$-th and general induction steps, respectively.

\subsection{Definition and properties of $Q_0$}
We begin with defining the 0-step singular point 
$$Q_0=\{c_0^i\in \Z^d:\ |v(\theta^*+c_0^i\cdot{\omega})-E^*|<\delta_0\}.$$
Since $\{c_0^i\}$ is a single point, $H_{\{c_0^i\}}(\theta^*)$ has a unique eigenvalue $v(\theta^*+ c_0^i\cdot{\omega})$. We denote it by $E_0^i(\theta^*)$. Then the {\bf Center Theorem} at the $0$-th step is 
\begin{thm}[]\label{0c}
	If $c_0^i,c_0^j\in Q_0 $, then \begin{equation}\label{st}
		m(c_0^i,c_0^j)\leq2|E_0^i(\theta^*)-E_0^j(\theta^*)|^{1/2}\leq2\delta_0^{1/2},
	\end{equation}
	where $m(c_0^i,c_0^j):=\min(\|(c_0^i-c_0^j)\cdot \omega\|,\|2\theta^*+(c_0^i+c_0^j)\cdot \omega\|).$
\end{thm}
\begin{proof}
	Let $c_0^i,c_0^j\in Q_0$.  We have $|v(\theta^*+c_0^i\cdot \omega)-v(\theta^*+c_0^j\cdot \omega)|<2\delta_0$. From Lemma \ref{a},  we obtain
	\begin{align*}
		m(c_0^i,c_0^j)^2&=\min(\|(c_0^i-c_0^j)\cdot \omega\|,\|2\theta^*+(c_0^i+c_0^j)\cdot \omega\|)^2\\&
		\leq|v(\theta^*+c_0^i\cdot \omega)-v(\theta^*+c_0^j\cdot \omega)|\\&=|E_0^i(\theta^*)-E_0^j(\theta^*)|\leq2\delta_0,
	\end{align*}
which proves the theorem. 
\end{proof}
Next, we give  Green's function estimates  for the  $0$-good set. 
\begin{thm}[]\label{0g}
	Let $\Lambda\cap Q_0=\emptyset$, $|\theta-\theta^*|<\delta_0/(10M_1)$ and $ |E-E^*|<\delta_0/5$. Then for $\varepsilon\leq \varepsilon_0=\delta_0^{20}\ll1,$
	\begin{align}
		\label{542}\|G_\Lambda(\theta;E)\|&\leq10\delta_0^{-1},\\
	\label{542.}
		|G_\Lambda(\theta;E)(x,y)|&<e^{-\gamma_0\|x-y\|_1}\ (x\neq y).
	\end{align}
\end{thm}
\begin{proof}
	Denote by $V_\Lambda(\theta)$ the operator $R_\Lambda v(\theta+x\cdot\omega)\delta_{x,y}R_\Lambda$. Since $\Lambda\cap Q_0=\emptyset$, 
	we have $\|V_\Lambda(\theta^*)-E^*\|\geq\delta_0$.  So   $\|V_\Lambda(\theta)-E\|\geq\delta_0/2$ for $|\theta-\theta^*|<\delta_0/(10M_1)$ and $ |E-E^*|<\delta_0/5$. Since $\|\Delta\|\leq2d$, we have by  the Neumann series  argument
	\begin{align*}
		G_\Lambda(\theta;E)& =\left(\varepsilon \Delta+V_\Lambda(\theta)-E\right)^{-1} \\
		& =\sum_{n=0}^{\infty}(-1)^n \varepsilon^n\left[\left(V_\Lambda(\theta)-E\right)^{-1} \Delta\right]^n\left(V_\Lambda(\theta)-E\right)^{-1}.
	\end{align*}
	Thus for $\varepsilon\leq \varepsilon_0,$
	$$\|G_\Lambda(\theta;E)\|\leq2\|\left(V_\Lambda(\theta)-E\right)^{-1}\|<4\delta_0^{-1},$$
	and 
	\begin{align*}
		\left|G_\Lambda(\theta;E)(x, y)\right| & \leq \frac{4}{\delta_0}\left(\frac{4 d\varepsilon}{\delta_0}\right)^{\|x-y\|_1}\leq\sqrt{\varepsilon}^{\|x-y\|_1}=e^{-\gamma_0\|x-y\|_1}\ (x\neq y).
	\end{align*}
\end{proof}

In the following, we will deal with the first and the general inductive steps in Section 3.2 and Section  3.4, respectively.  For convenience, we include a diagram to clarify the inductive structure.  
						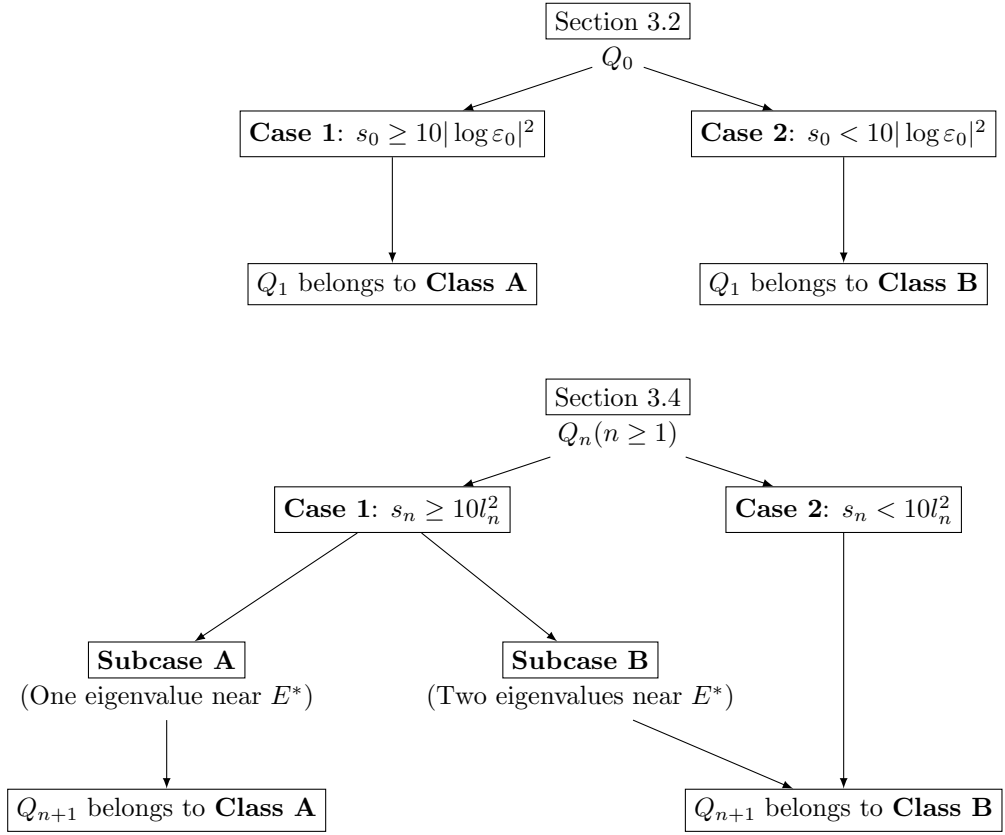
\begin{figure}[htp]
		\begin{tikzpicture}[>=latex, scale=1]
			\draw (0,0.5)node[rectangle, draw=black]{Section 3.2};
			\draw (0,0)node(s1){$Q_0$};
			\draw (-3,-1)node(s2)[rectangle, draw=black]{{\bf Case 1}:  $s_0\geq 10|\log\varepsilon_0|^2$};		
			\draw (3,-1)node(s3)[rectangle, draw=black]{{\bf Case 2}: $s_0< 10|\log\varepsilon_0|^2$};		
			\draw [->](s1)--(s2);
			\draw [->](s1)--(s3);
				\draw (-3,-3)node(s4)[rectangle, draw=black]{$Q_1$ belongs to {\bf Class A}};
				\draw (3,-3)node(s5)[rectangle, draw=black]{$Q_1$ belongs to {\bf Class B}};
					\draw [->](s2)--(s4);
				\draw [->](s3)--(s5);
					\draw (0,-4.5)node[rectangle, draw=black]{Section 3.4};
					\draw (0,-5)node(s6){$Q_n(n\geq 1)$};
					\draw (-3,-6)node(s7)[rectangle, draw=black]{{\bf Case 1}:  $s_n\geq 10l_n^2$};		
					\draw (3,-6)node(s8)[rectangle, draw=black]{{\bf Case 2}:  $s_n<10l_n^2$};	
					\draw [->](s6)--(s7);
				\draw [->](s6)--(s8);
					\draw (-6,-8)node(s9)[rectangle, draw=black]{{\bf Subcase A}};		
						\draw  (-6,-8.5)node(s10){(One eigenvalue near $E^*$)};	
						\draw (-0.5,-8)node(s11)[rectangle, draw=black]{{\bf Subcase B}};
								
						\draw  (-0.5,-8.5)node(s100){(Two  eigenvalues near $E^*$)};	
						\draw [->](s7)--(s9);
						\draw [->](s7)--(s11);
					\draw (-6,-10)node(s12)[rectangle, draw=black]{$Q_{n+1}$ belongs to {\bf Class A}};
				\draw (3,-10)node(s13)[rectangle, draw=black]{$Q_{n+1}$ belongs to {\bf Class B}};
						\draw [->](s10)--(s12);
				\draw [->](s100)--(s13);
					\draw [->](s8)--(s13);
		\end{tikzpicture}
		\caption*{A diagram of the inductive structure}
	\end{figure}

\subsection{Definition and properties of $Q_1$}
In this section,  we define $Q_1$ and establish Theorem \ref{mainthm} for $n=1$. Let $$s_0=\min_{c_1^i\neq c_1^j\in Q_0}\|c_1^i-c_1^j\|_1.$$ 
We shall distinguish two cases. \smallskip\\
{\textbf{Case 1}}.  $s_0>10|\log\varepsilon_0|^2.$ We define $P_1=Q_0$ and associate every $c_1^i\in P_1$ an $l_1:=|\log\varepsilon_0|^2$-size block $B_1^i=\Lambda_{l_1}(c_1^i)$. 
Define
$$Q_1=\big\{c_1^i\in P_1:\ \operatorname{dist}(\sigma(H_{B_1^i}(\theta^*)),E^*)<\delta_1:=e^{-l_1^{2/3}}\big\}. $$
\begin{rem}
	Since $|\log\delta_1|=l_1^{2/3}\sim|\log\delta_0|^{4/3}$, we have $\delta_1<\delta_0^{100}$.
\end{rem}
We show  that in this case, for $c_1^i\in Q_1, |\theta-\theta^*|<\delta_0/(10M_1)$,  the eigenvalue parametrization of $H_{B_1^i}(\theta)$ in the interval   $|E-E^*|<\delta_0/5$  is unique and hence a well-defined $C^2$ function  of  those $\theta$ by Lemma \ref{neqs}. 
\begin{prop}\label{k1}
	For every $c_1^i\in Q_1$ and $|\theta-\theta^*|<\delta_0/(10M_1)$,
	\begin{itemize}
		\item[\textbf{(a)}]  $H_{B_1^i}(\theta)$ has a unique eigenvalue $E_1^i(\theta)$ such that $|E_1^i(\theta)-E^*|<\delta_0/9$. Moreover, any other $\hat{E}\in\sigma(H_{B_1^i}(\theta)) $ must obey $|\hat{E}-E^*|>\delta_0/5$.
		\item[\textbf{(b)}] The corresponding  eigenfunction $\psi_1$  satisfies  
		$$|\psi_1(x)|\leq e^{-\gamma_0\|x-c_1^i\|_1} .$$
		\item[\textbf{(c)}] $\|G_{B_1}^\perp(\theta;E_1^i)\|\leq20\delta_0^{-1}$, where $G_{B_1}^\perp$ denotes the Green's function for $B_1^i$ restricted on the orthogonal complement of $\psi_1$.
	\end{itemize} 
\end{prop}

\begin{proof}
	Since $ B_1^i$ is singular, by definition,  $H_{B_1^i}(\theta^*)$ has an eigenvalue $E_1^i(\theta^*)$ such that  $|E_1^i(\theta^*)-E^*|<\delta_1\ll\delta_0^3$. By $|V'|\leq M_1$,   $\sigma(H_{B_1^i}(\theta))$ and  $\sigma(H_{B_1^i}(\theta^*))$   differ at most $M_1|\theta-\theta^*|<\delta_0/10$, which shows the existence of $E_1^i(\theta)$ in $|E-E^*|<\delta_0/9$. Denote $\Lambda=B_1^i\setminus\{c_1^i\}$. Let $E\in \sigma(H_{B_1^i}(\theta))$ be such that $|E-E^*|<\delta_0/5$. We determine the value of $\psi_1(x)$ by 
	$$\psi_1(x)= \sum_{\|y-c_1^i\|_1=1} G_{\Lambda}(\theta;E)(x, y) \Gamma_{y,c_1^i} \psi_1\left(c_1^i\right).$$
	Since $\Lambda$ is 0-good, we have $$|G_\Lambda(\theta;E)(x,y)|\leq\delta_0^{-1}e^{-\gamma_0\|x-y\|_1}.$$ 
	Thus, 
	\begin{equation}\label{g}
		|\psi_1(x)|\leq C\frac{\varepsilon}{\delta_0}e^{-\gamma_0\|x-c_1^i\|_1}\leq e^{-\gamma_0\|x-c_1^i\|_1}.
	\end{equation} 
	This proves \textbf{(b)}. If there is another $\hat{E}\in \sigma(H_{B_1^i}(\theta))$ satisfying $|\hat{E}-E^*|\leq\delta_0/5$, by the above argument,  its   eigenfunction $\hat{\psi}$ must also almost localize on the single point $\{c_1^i\}$, which violates the orthogonality of $\psi_1 $ and $\hat{\psi}$. Thus, we prove the uniqueness part of \textbf{(a)}. Finally, \textbf{(c)} follows from the fact that any other $\hat{E}\in \sigma(H_{B_1^i}(\theta))$ must obey $|\hat{E}-E_1^i(\theta)|\geq |\hat{E}-E^*|-|E^*-E_1^i(\theta)|\geq\delta_0/5-\delta_0/9\geq\delta_0/20$ and 
 $\|G_{B_1}^\perp(\theta;E_1^i)\|= \operatorname{dist}(\sigma(H_{B_1^i}(\theta)),E_1^i(\theta))^{-1}.$
\end{proof}
We then give upper bounds on the derivatives of $E_1^i(\theta)$. 
\begin{prop}\label{ap}
	For $|\theta-\theta^*|<\delta_0/(10M_1)$, we have 
	$$\big|\frac{d^s}{d\theta^s}(E_1^i(\theta)-E_0^i(\theta))\big|\leq\delta_0^7\ \ {\rm for\ }s=0,1,2.$$
\end{prop}

\begin{proof}
	Denote by $\psi_r$ the corresponding  eigenfunction of $E_r^i$ for $r=0,1$.
	Recalling \eqref{g}, we have $\|\psi_1-\psi_0\|\leq 2\frac{\varepsilon}{\delta_0}\leq2\delta_0^{10}.$ Thus, 
	$$|E_1^i(\theta)-E_0^i(\theta)|=|\left\langle \psi_1,H_{B_1^i}(\theta)\psi_1\right\rangle -\left\langle \psi_0,H_{B_0^i}(\theta)\psi_0\right\rangle|\leq\delta_0^{9}.$$
For $s=1,2$, we use the eigenvalue perturbation formulas from Appendix \ref{AB}. Thus  $$|\frac{d}{d \theta} E_1^i(\theta)-\frac{d}{d \theta} E_0^i(\theta)|=|\left\langle\psi_1,V' \psi_1\right\rangle-\left\langle\psi_0,V' \psi_0\right\rangle|\leq\delta_0^9,$$
and
 $$\frac{d^2}{d \theta^2} E_1^i(\theta)=\left\langle\psi_1, V^{\prime \prime} \psi_1\right\rangle-2\left\langle\psi_1, V' G_{B_1}^\perp(\theta;E_1^i) V' \psi_1\right\rangle, \ \frac{d^2}{d \theta^2} E_0^i(\theta)= \left\langle\psi_0, V^{\prime \prime} \psi_0\right\rangle. $$
	Since   $|\psi_1(x)|\leq e^{-\gamma_0\|x-c_1^i\|_1}$,  $|V'\psi_1(x)|\leq M_1e^{-\gamma_0\|x-c_1^i\|_1}$ are two functions  almost localized on $\{c_1^i\}$, we deduce  $\|P_1^\perp (V'\psi_1)\|\leq\delta_0^{9}$, where $P_1^\perp$ denotes projection  onto  the orthogonal complement  of  $\psi_1$. Thus,
	\begin{align*}
		\big|\frac{d^2}{d \theta^2} E_1^i(\theta)-\frac{d^2}{d \theta^2} E_0^i(\theta)\big|&=|\left\langle\psi_1, V^{\prime \prime} \psi_1\right\rangle-\left\langle\psi_0, V^{\prime \prime} \psi_0\right\rangle-2\left\langle\psi_1, V' G_{B_1}^\perp(\theta;E_1^i) V' \psi_1\right\rangle|\\
		&\leq \delta_0^9+2\|G_{B_1}^\perp(\theta;E_1^i)\|\cdot \|P_1^\perp (V'\psi_1)\|\\
		&<\delta_0^7,
	\end{align*}
where we have used Proposition \ref{k1} to bound the term $\|G_{B_1}^\perp(\theta;E_1^i)\|$.
\end{proof}

We can also have the lower bound on the derivatives of $E_1^i(\theta).$

\begin{prop}\label{133}
	For $|\theta-\theta^*|<\delta_0/(20M_1)$, there exists $\mu=0\text{ or }1/2$ such that 
	$$\big|\frac{d}{d \theta} E_1^i(\theta)\big|\geq\min(\delta_0^2,\|\theta+c_1^i\cdot\omega- \mu \|).$$
\end{prop}
\begin{proof}
	Assuming  $|\frac{d}{d \theta} E_1^i(\theta)|\leq\delta_0^2$, then by Proposition \ref{ap}, we have $|v'(c_1^i\cdot \omega+\theta)|=|\frac{d}{d \theta} E_0^i(\theta)|<2\delta_0^2.$ Using Lemma \ref{'0}, we get  
$$\min(\|\theta+c_1^i\cdot \omega\|,\|\theta+c_1^i\cdot \omega-\frac{1}{2}\|)\leq \delta_0^2<a.$$ 
Without loss of generality, we can assume $\|\theta+c_1^i\cdot \omega\|\leq\delta_0^2$. Set $\mu=0$.  Thus, $\|\theta^*+c_1^i\cdot \omega\|\leq\delta_0/(20M_1)+\delta_0^2<\delta_0/(10M_1)$.  It follows that the interval of  $\T$ with endpoints  $\theta$ and $-c_1^i\cdot \omega$  is contained in $\{\theta:\ |\theta-\theta^*|<\delta_0/(10M_1)\}$. By the assumption of $v$, $|\frac{d^2}{d \theta^2} E_0^i(\xi)|=v''(\xi+c_1^i\cdot \omega)>2$ for $\xi$ belonging to the interval of $\T$ with endpoints $\theta$ and $-c_1^i \cdot \omega$.    So, by Proposition \ref{ap},  $|\frac{d^2}{d \theta^2} E_1^i(\xi)|>1$. Notice that  $E_1^i$ is  symmetric about $-c_1^i\cdot \omega$ since $H_{B_1^i}(\theta)=H_{B_1^i}(-2c_1^i\cdot \omega-\theta)$ and the uniqueness of eigenvalue in $|E-E^*|<\delta_0/5$. We have  $ \frac{d}{d \theta} E_1^i(-c_1^i\cdot \omega)=0$. Thus, 
$$\big|\frac{d}{d \theta} E_1^i(\theta)\big|=|\frac{d}{d \theta} E_1^i(\theta)-\frac{d}{d \theta} E_1^i(-c_1^i\cdot \omega)|\geq|\frac{d^2}{d \theta^2} E_1^i(\xi)|\cdot \|\theta+c_1^i\cdot \omega\|\geq \|\theta+c_1^i\cdot \omega\|.$$
\end{proof}
\begin{rem}
	We will see from the Theorem \ref{1c} that $\mu=0 \text{ or }1/2$ can be chosen independently of $c_1^i\in Q_1$. 
\end{rem}
Combining the above two propositions shows 
\begin{prop}\label{dx}
	If $|\frac{d}{d\theta}E_1^i|<\delta_0^2$ for some $|\theta-\theta^*|<\delta_0/(10M_1)$, then $|\frac{d^2}{d\theta^2}E_1^i|\geq3-\delta_0^3>2$  for all $|\theta-\theta^*|<\delta_0/(10M_1)$.
\end{prop}

\begin{proof}
	From the proof of  Proposition \ref{133},  $\min(\|\theta+c_1^i\cdot \omega\|,\|\theta+c_1^i\cdot \omega-\frac{1}{2}\|)\leq\delta_0^2\ll a$, which gives  $|\frac{d^2}{d\theta^2}E_0^i(\theta)|=|v''(\theta+c_1^i\cdot \omega )|>3$ for all $|\theta-\theta^*|<\delta_0/(10M_1)$.   Thus, $|\frac{d^2}{d\theta^2}E_1^i|>3-\delta_0^3$ by Proposition \ref{ap}.
\end{proof}

\begin{prop}\label{233}
	For  $c_1^i,c_1^j\in Q_1$, we have  $m(c_1^i,c_1^j)\leq\delta_0^{3}$. Thus, $\theta^*\pm m(c_1^i,c_1^j)$ belongs to the interval  with $|\theta-\theta^*|<\delta_0/(20M_1)$. Moreover, we have $E_1^j(\theta^*)=E_1^i(\theta^*+h)$, where $h=(c_1^j-c_1^i)\cdot \omega$ or $-((c_1^i+c_1^j)\cdot \omega+2\theta^*)$ {\rm (mod $1$)} satisfying $|h|=m(c_1^i,c_1^j)$.
\end{prop}

\begin{proof}
	From Proposition \ref{ap} for $s=0$,  we have $|E_0^i(\theta^*)-E_0^j(\theta^*)|\leq|E_0^i(\theta^*)-E_1^i(\theta^*)|+|E_1^i(\theta^*)-E_1^j(\theta^*)|+|E_1^j(\theta^*)-E_0^j(\theta^*)|\leq2\delta_1+2\delta_0^7<\delta_0^6/2$.   Then we get $m(c_1^i,c_1^j)=m(c_0^i,c_0^j)\leq\delta_0^{3}$ by Theorem \ref{0c}. The second statement follows from $H_{B_1^i}(\theta^*+h)=H_{B_1^j}(\theta^*)$ and the uniqueness of the eigenvalue in the interval  with $|E-E^*|<\delta_0/5$.
\end{proof}

Then {\bf Center Theorem} of the $1$-th step in {\bf Case 1} is as follows.

\begin{thm}[]\label{1c}
If $c_1^i,c_1^j\in Q_1 $, then \begin{equation*}
	m(c_1^i,c_1^j)<\sqrt{2}|E_1^i(\theta^*)-E_1^j(\theta^*)|^{1/2}<2\delta_1^{1/2}.
\end{equation*}
\end{thm}
\begin{proof}
By Proposition \ref{233}, we have $E_1^j(\theta^*)-E_1^i(\theta^*)=E_1^i(\theta^*+h)-E_1^i(\theta^*)$. If $|\frac{d}{d\theta}E_1^i|\geq\delta_0^3$ for all $|\theta-\theta^*|<|h|$, where $h$ was  defined in   Proposition \ref{233}, we get  
$$|E_1^i(\theta^*+h)-E_1^i(\theta^*)|\geq\delta_0^3 |h|\geq h^2.$$
Otherwise,  $|\frac{d}{d\theta}E_1^i|<\delta_0^3$ for some  $|\theta-\theta^*|<|h|$. By 
Proposition \ref{133}, we have $\delta_0^3>|\frac{d}{d \theta} E_1^i(\theta)|\geq\min(\|\theta+c_1^i\cdot \omega\|,\|\theta+c_1^i\cdot \omega-\frac{1}{2}\|)$.  So, the symmetry point $\theta_s=-c_1^i\cdot \omega$ or $-c_1^i\cdot \omega -1/2$ (mod $1$) belongs to the interval with  $|\theta-\theta^*|<\delta_0^2$. Recalling Proposition \ref{dx}, $E_1^i$ satisfies the conditions of Lemma \ref{C2} in Appendix \ref{AA} with $\theta_2=\theta^*+h,\theta_1=\theta^*,\delta=\delta_0^2$ and $|h|\leq\delta$. Thus, we have  
$$|E_1^i(\theta^*+h)-E_1^i(\theta^*)|\geq \frac{1}{2}\min(h^2,|2\theta^*+h-2\theta_s|^2)=\frac{1}{2}h^2.$$
\end{proof}

In the following, we deal with {\bf Case 2}, in which the level crossing may  take  place. \smallskip\\
{\textbf{Case 2}}.  $s_0\leq10|\log\varepsilon_0|^2.$
First, we have
\begin{lem}\label{mirrr}
Let $c_0^I,c_0^J\in Q_0$ satisfy $\|c_0^I-c_0^J\|_1=s_0$. Then every point $c_0^i\in Q_0$ has a mirror image $\tilde{c}_0^i=c_0^i\pm(c_0^J-c_0^I)$, whose sign is uniquely determined by 
\begin{equation}\label{mirr}
	\|2\theta^*+(c_0^i+\tilde{c}_0^i)\cdot \omega\|\leq6\delta_0^{1/2}.
\end{equation}
\end{lem}

\begin{proof}
 Since $s_0 \leq10|\log\varepsilon_0|^2$, by the Diophantine  condition of $\omega$ and Theorem \ref{1c},  we must have $\|c_0^I\cdot \omega+c_0^J\cdot \omega+2 \theta^*\| \leq 2 \delta_0^{1/2}$. If $\|(c_0^i-c_0^I)\cdot \omega\| \leq 2 \delta_0^{1 / 2}$, we define $\tilde{c}_0^i=c_0^i+(c_0^J-c_0^I)$ and it is easy to check that  \eqref{mirr} holds true. If $\|(c_0^i+c_0^I)\cdot \omega+2 \theta^*\| \leq 2 \delta_0^{1 / 2}$, then $\tilde{c}_0^i=c_0^i-(c_0^J-c_0^I)$ is the required mirror image.
\end{proof}
\begin{rem}
	We call $\tilde{c}_0^i$ the mirror image of $c_0^i$ because for all $x\in \Z^d$, $v(\theta^*+(c_0^i+x)\cdot \omega)=v(\theta^*+(\tilde{c}_0^i-x)\cdot \omega)+O(\delta_0^{1/2})$. The mirror image is almost singular (in the sense of $\delta_0^{1/2}$-resonance) but might not belong to $Q_0$. This lemma together with Theorem \ref{0c} shows that each set $\Lambda$ with  $\operatorname{diam}\Lambda\sim |\log\varepsilon_0|^4$ can contain no more than two points of $Q_0$ and its mirror images. A third point  is excluded by $|\log\varepsilon_0|^4\ll \gamma \delta_0^{-1/(2\tau)}$ and the Diophantine condition of $\omega$.
\end{rem}
  In this case, we define $P_1=\{c_1^i:\ c_1^i=(c_0^i+\tilde{c}_0^i)/2,\ c_0^i\in  Q_0\}$ and associate every $ c_1^i\in P_1$ an $l_1:=|\log\varepsilon_0|^4$-size block $B_1^i=\Lambda_{l_1}(c_1^i)$. Again $Q_1$ is defined as 
$$Q_1=\big\{c_1^i\in P_1:\ \operatorname{dist}(\sigma(H_{B_1^i}(\theta^*)),E^*)<\delta_1:=e^{-l_1^{2/3}}\big\}. $$

\begin{lem}\label{514}
	There exists  $\mu=0\text{ or }1/2$ such that for every $c_1^i \in Q_1$,  $\|\theta^*+ c_1^i\cdot \omega+\mu\|\leq 3\delta_0^{1/2}.$
	\end{lem}
\begin{proof}
	Let $c_1^i, c_1^j \in Q_1$. Recall  the definition of mirror image in Lemma \ref{mirrr}. If we denote $(c_1^i)^{\pm}=c_1^i\pm(c_0^J-c_0^I)/2, (c_1^j)^{\pm}=c_1^j\pm(c_0^J-c_0^I)/2$, then $(c_1^i)^{+}$ (resp. $(c_1^j)^{+}$) is the mirror image of   $(c_1^i)^{-}$ (resp. $(c_1^j)^{-}$). Using  \eqref{mirr} and the  simple fact $m(k_1,k_3)\leq m(k_1,k_2)+m(k_2,k_3)$, we deduce $m((c_1^i)^{+},(c_1^j)^{+})\leq20\delta_0^{1/2}$. So, we must exclude the case  \begin{equation}\label{434}
		\|((c_1^i)^{+}+(c_1^j)^{+})\cdot \omega+2\theta^*\|\leq20\delta_0^{1/2}.
	\end{equation} 
Otherwise, assume that \eqref{434} holds true. From \eqref{mirr}, we obtain  $\|((c_1^j)^{+}-(c_1^i)^{-})\cdot \omega\|\leq26\delta_0^{1/2}$ and  $\|((c_1^j)^{-}-(c_1^i)^{+})\cdot \omega\|\leq26\delta_0^{1/2}$, which gives  us $$\|((c_1^j)^{+}-(c_1^j)^{-}+(c_1^i)^{+}-(c_1^i)^{-})\cdot \omega\|\leq52\delta_0^{1/2}.$$
However, by the Diophantine condition of $\omega$,  the left hand side of the above inequality has a lower bound 
$$\|((c_1^j)^{+}-(c_1^j)^{-}+(c_1^i)^{+}-(c_1^i)^{-})\cdot \omega\|=\|2(c_1^J-c_1^I)\cdot \omega\|\geq\frac{\gamma}{(2s_0)^\tau}\gg\delta_0^{1/2}.$$
Thus, we must have $\|((c_1^i)^{+}-(c_1^j)^{+})\cdot \omega\|\leq20\delta_0^{1/2}$ and hence,  
\begin{equation}\label{508}
	\|(c_1^i-c_1^j)\cdot \omega\|\leq20\delta_0^{1/2}. 
\end{equation} By \eqref{mirr}, we have $\|\theta^*+c\cdot \omega\|\leq3\delta_0^{1/2}$ or  $\|\theta^*+c\cdot \omega+\frac{1}{2}\|\leq3\delta_0^{1/2}$ for every $c\in Q_1 $ and exactly one of the inequalities holds true since $6\delta_0^{1/2}<\frac{1}{2}$. Assume that there exist $c_1^i, c_1^j \in Q_1$ such that $\|\theta^*+c_1^i\cdot \omega\|\leq3\delta_0^{1/2}$ and $\|\theta^*+c_1^j\cdot \omega-\frac{1}{2}\|\leq3\delta_0^{1/2}$. Then $\|(c_1^j-c_1^i)\cdot \omega+\frac{1}{2}\|\leq6\delta_0^{1/2}$, which contradicts  \eqref{508}. 
\end{proof}
From  Lemma \ref{514}, there is $\mu=0\text{ or }1/2$ such that for every $c_1^i\in Q_1$, there exists a symmetric  point $\theta_s$ satisfying  \begin{equation}\label{duichen}
	\theta_s:=-c_1^i\cdot \omega+\mu\  { \rm (mod\ 1)},\  |\theta_s-\theta^*|\leq3\delta_0^{1/2}.
\end{equation}
We  call $\theta_s$  the symmetric  point of  $H_{B_1^i}(\theta)$ since $H_{B_1^i}(\theta)=H_{B_1^i}(2\theta_s-\theta)$. 
For $c_0^i \in Q_0$, we have $|E^i_0(\theta^*)-E^*| \leq \delta_0 $. For convenience, we define $\tilde{E}_0^i(\theta)=v(\theta+\tilde{c}_0^i\cdot \omega)$. Moreover, $|\tilde{E}_0^i(\theta^*)-E^*| \leq 10M_1\delta_0^{1/2}$ since $m(c_0^i,\tilde{c}_0^i)\leq6\delta_0^{1/2}$. Thus, in each block $B_1^i$, we have two values of the potential near $E^*$ which will be used to generate two eigenvalues in $\sigma(H_{B_1^i}(\theta))$ near $E^*$. More precisely, we have 
\begin{prop}\label{55}If $c_1^i\in Q_1$, then for $|\theta-\theta^*|<10\delta_0^{1/2}$,
	\begin{itemize}
		\item [\textbf{(a)}] $H_{B_1^i}(\theta)$ has  exact two eigenvalues $E_1^i(\theta)$ and $\mathcal{E}_1^i(\theta)$ in the interval $|E-E^*|<50M_1\delta_0^{1/2}$. Moreover, any other $\hat{E}\in\sigma(H_{B_1^i}(\theta)) $ must obey $|\hat{E}-E^*|>2\delta_0^{1/8}$.
		\item [\textbf{(b)}] The corresponding   eigenfunction of $E_1^i$ {\rm (resp. $\mathcal{E}_1^i$)}, $\psi_1$ {\rm(resp. $\Psi_1$)} decays exponentially fast away from $c_0^i$ and $\tilde{c}_0^i$, i.e.,
		$$|\psi_1(x)|\leq e^{-\gamma_0\|x-c_1^i\|_1}+e^{-\gamma_0\|x-\tilde{c}_1^i\|_1},$$
		$$|\Psi_1(x)|\leq e^{-\gamma_0\|x-c_1^i\|_1}+e^{-\gamma_0\|x-\tilde{c}_1^i\|_1}.$$
	Thus, the two eigenfunctions can be expressed as 
	\begin{align}\label{y}
	\begin{split}
	\psi_1(x)=A\delta(x-c_0^i)+B\delta(x-\tilde{c}_0^i)+O(\delta_0^{10}),\\
	\Psi_1(x)=B\delta(x-c_0^i)-A\delta(x-\tilde{c}_0^i)+O(\delta_0^{10}),
	\end{split}	
	\end{align}
	where $A^2+B^2=1$.
	\item [\textbf{(c)}] $\|G_{B_1^i}^{\perp\perp}(\theta;E_1^i)\|\leq\delta_0^{-1/8}$, where $G_{B_1^i}^{\perp\perp}$ denotes the Green's function for $B_1^i$ on the orthogonal complement of the space spanned by $\psi_1$ and $\Psi_1$.
	\end{itemize}
\end{prop}
\begin{proof}
From  $|v(\theta^*+c_0^i\cdot \omega)-E^*|<\delta_0$, we get $\|(H_{B_1^i}(\theta^*)-E^*)\delta(x-c_0^i)\|<\delta_0+2d\varepsilon<2\delta_0$. Since $m(c_0^i.\tilde{c}_0^i)\leq6\delta_0^{1/2}$, we have $|v(\theta^*+\tilde{c}_0^i\cdot \omega)-E^*|<6\delta_0^{1/2}+\delta_0$ and hence $\|(H_{B_1^i}(\theta^*)-E^*)\delta(x-\tilde{c}_0^i)\|<7\delta_0^{1/2}$.
	Thus, we find two orthogonal  trial wave functions of  $H_{B_1^i}(\theta^*)-E^*$, which proves the existence of $E_1^i(\theta^*),\mathcal{E}_1^i(\theta^*)$ in $|E-E^*|\leq7\sqrt{2}\delta_0^{1/2}$ by Corollary \ref{trialcor} in Appendix \ref{AD}. Using $|V'|\leq M_1$, we can extend the existence of  $E_1^i(\theta),\mathcal{E}_1^i(\theta)$ in $|E-E^*|<50M_1\delta_0^{1/2}$   for $|\theta-\theta^*|<10\delta_0^{1/2},$ which proves the existence part of \textbf{(a)}. To establish the decay of eigenfunctions, we notice that \begin{align*}
		|v(\theta+x\cdot \omega)-E^*|&\geq |v(\theta^*+x\cdot \omega)-v(\theta^*+c_1^i\cdot \omega)|\\
		&\ \ \ -|v(\theta+x\cdot \omega)-v(\theta^*+x\cdot \omega)|-|v(\theta^*+c_1^i\cdot \omega)-E^*|\\
		&\geq m(x,c_1^i)^2- 10M_1\delta_0^{1/2}-\delta_1 \\
		&\geq (\frac{\gamma}{(2l_1)^\tau}-6\delta_0^{1/2})^2-11M_1\delta_0^{1/2}\\&>10\delta_0^{1/8}
	\end{align*} for $x\in B_1^i\setminus\{c_0^i,\tilde{c}_0^i\}$ by Lemma \ref{a} and the Diophantine condition. Thus, the  Green's function of $\Lambda=B_1^i\setminus\{c_0^i,\tilde{c}_0^i\}$ satisfies 
	$$\|G_\Lambda(\theta;E_1^i)\|\leq \delta_0^{-1/8},\ |G_\Lambda(\theta;E_1^i)(x,y)|\leq\delta_0^{-1/8}e^{-\gamma_0\|x-y\|_1},$$ which along with the Poisson's identity yields the exponential decay of eigenfunctions in {\bf (b)}.  The expression \eqref{y} follows from the fact that $\psi_1$ and $\Psi_1$ are normalized and orthogonal to each other. Finally, if there exists a third  eigenvalue $|\hat{E}-E^*|<2\delta_0^{1/8}$, the same argument shows that its eigenfunction decays exponentially fast  away from $c_0^i$ and $\tilde{c}_0^i$ and hence almost localized in $\{c_0^i,\tilde{c}_0^i\}$, which violates the orthogonality. Obviously, \textbf{(c)} immediately follows from \textbf{(a)}.
\end{proof}
\begin{rem}
	In \textbf{(b)}, we express $\psi_{1}$ and $\Psi_{1}$ in terms of $\psi_{0}=\delta(x-c_0^i)$ and $\tilde{\psi}_0=\delta(x-\tilde{c}_0^i)$. This  will allow us to relate the derivatives of $E_{1}^i$ and $\mathcal{E}_{1}^i$ to those of $E_0$ and $\tilde{E}_0$. To this end, we need to prove two technical lemmas about $E_0$ and $\tilde{E}_0$.
\end{rem}
\begin{lem}\label{cha}
	For $|\theta-\theta^*|<10\delta_0^{1/2}$, we have 
\begin{equation}\label{ab}
		\big|\frac{d}{d\theta}\left(E_0^i+\tilde{E}_0^i\right)(\theta)\big|\leq 30M_1\delta_0^{1/2}.
\end{equation}
\end{lem}
\begin{proof}
Recalling the definition of $\theta_s$ (cf. \eqref{duichen}) and  from  $\tilde{E}_0^i(\theta)=E_0^i(-\theta+2\theta_s)$, we deduce 
 $$|\frac{d}{d\theta}\left(E_0^i+\tilde{E}_0^i\right)(\theta_s)|=0.$$
 Thus,
 \begin{align*}
 	\big|\frac{d}{d\theta}\left(E_0^i+\tilde{E}_0^i\right)(\theta)\big|&=|\frac{d}{d\theta}\left(E_0^i+\tilde{E}_0^i\right)(\theta)-\frac{d}{d\theta}\left(E_0^i+\tilde{E}_0^i\right)(\theta_s)|\\
 	&\leq\sup \left(|\frac{d^2}{d\theta^2}E_0^i|+ |\frac{d^2}{d\theta^2}\tilde{E}_0^i|\right)\cdot |\theta-\theta_s|\\
 	&\leq30M_1\delta_0^{1/2}.
 \end{align*}
\end{proof}
\begin{lem}\label{xiajie}
	For $|\theta-\theta^*|<10\delta_0^{1/2}$, we have $|\frac{d}{d\theta}E_0^i(\theta)|\geq\delta_0^{1/9}$.
\end{lem}
\begin{proof}
	Since $\|2\theta^*+(c_0^i+\tilde{c}_0^i)\cdot \omega\|\leq6\delta_0^{1/2}$, we deduce from  the Diophantine  condition that $$\|2\theta+2c_0^i\cdot \omega\|\geq \|(\tilde{c}_0^i-c_0^i)\cdot \omega\|-\|2\theta^*+(c_1^i+\tilde{c}_1^i)\cdot \omega\|-2|\theta-\theta^*|>2\delta_0^{1/9}.$$
	Thus, $\min(\|\theta+c_0^i\cdot \omega\|,\|\theta+c_0^i\cdot \omega+\frac{1}{2}\|)\geq\delta_0^{1/9}.$
	The proof now follows from  Lemma \ref{'0}.
\end{proof}

Now we can prove the following proposition,  which relates the derivatives of $E_{1}^i$ and $\mathcal{E}_{1}^i$ to those of $E_0$ and $\tilde{E}_0$.

\begin{prop}\label{guji}
	Let $|\theta-\theta^*|<10\delta_0^{1/2}$. Then 
	\begin{itemize}
		\item [\textbf{(a)}] $E_1^i$ and $\mathcal{E}_1^i$ are $C^1$ functions and if $ E_1^i(\theta)\neq\mathcal{E}_1^i(\theta)$, then
		\begin{align}
			\frac{d}{d \theta} E_1^i & =(A^2-B^2) \frac{d}{d\theta} E_0^i+O(\delta_0^{1 / 2} ), \label{de}\\
			\frac{d}{d \theta} \mathcal{E}_1^i & =(B^2-A^2) \frac{d}{d\theta} E_0^i+O(\delta_0^{1 / 2} ).\nonumber
		\end{align}
	\item [\textbf{(b)}] If $ E_1^i(\theta)\neq\mathcal{E}_1^i(\theta)$, then $\frac{d^2}{d \theta^2} E_1^i(\theta)$ and 	$\frac{d^2}{d \theta^2} \mathcal{E}_1^i(\theta) $ exist. Moreover,
		\begin{align}
			\frac{d^2}{d \theta^2} E_1^i & =\frac{2\left\langle\psi_1^i, V' \Psi_1^i\right\rangle^2}{E_1^i-\mathcal{E}_1^i}+O(\delta_0^{-1 / 8} ),\label{df} \\
			\frac{d^2}{d \theta^2} \mathcal{E}_1^i & =\frac{2\left\langle\psi_1^i, V' \Psi_1^i\right\rangle^2}{\mathcal{E}_1^i-E_1^i}+O(\delta_0^{-1 / 8} ). \label{ddf}
		\end{align}
		\item [\textbf{(c)}] At the point $E_1^i(\theta)\neq \mathcal{E}_1^i(\theta)$, if $|\frac{d}{d\theta}E_1^i(\theta)|<\delta_0^{1/4}$, then $|\frac{d^2}{d\theta^2}E_1^i(\theta)|>\delta_0^{-1/4}>2$. Moreover, the  sign of $\frac{d^2}{d\theta^2}E_1^i(\theta)$ is the same as that of $ E_1^i(\theta)- \mathcal{E}_1^i(\theta)$. The analogous conclusion holds by exchanging  $E_1^i(\theta)$ and $\mathcal{E}_1^i(\theta)$.
	\end{itemize}
\end{prop}
\begin{proof}We only  give the proof concerning  $E_1^i$. The $C^1$ smoothness of the eigenvalues is a remarkable result of perturbation theory for self-adjoint operators (cf. \cite{Rel69} and \cite{Kat95}). By \eqref{y} and Lemma \ref{cha}, we refer to  Appendix \ref{AB} to obtain 
	\begin{align*}
		\frac{d}{d \theta} E_1^i & =\left\langle\psi_1^i, V' \psi_1^i\right\rangle=A^2 \frac{d}{d \theta} E_0^i+B^2 \frac{d}{d \theta} \tilde{E}_0^i+O(\delta_0^2 ) \\
		& =(A^2-B^2) \frac{d}{d \theta} E_0^i+B^2(\frac{d}{d \theta} E_0^i+\frac{d}{d \theta} \tilde{E}_0^i)+O(\delta_0^2 )\\
		& =(A^2-B^2) \frac{d}{d\theta} E_0^i+O(\delta_0^{1 / 2} ),
	\end{align*}
	where we  have used \eqref{ab} in the last identity. This  completes  the proof of  \textbf{(a)}.
	To prove  \textbf{(b)},  we use the formula 
	$$
	\frac{d^2}{d \theta^2} E_1^i=\left\langle\psi_1^i, V'' \psi_1^i\right\rangle+2 \frac{\left\langle\psi_1^i, V' \psi_1^i\right\rangle^2}{E_1^i-\mathcal{E}_1^i}-2\left\langle V' \psi_1^i,G^{\perp \perp}_{B_1^i}(\theta;E_1^i) V' \psi_1^i\right\rangle .
	$$
	The last term is bounded by $2\|G^{\perp \perp}_{B_1^i}(\theta;E_1^i)\|\cdot\|V' \psi_1^i\|^2$, where we can use the estimate $\|G^{\perp \perp}_{B_1^i}(\theta;E_1^i)\| \leq\delta_0^{-1 / 8}$in  \textbf{(c)} of Proposition \ref{55}.  Now we turn to the proof of \textbf{(c)}. If $|\frac{d}{d\theta}E_1^i(\theta)|<\delta_0^{1/4}$, then by \eqref{de},  we have 
	$$|A^2-B^2|\cdot  |\frac{d}{d\theta}E_0^i(\theta)|<2\delta_0^{1/4},$$
	which implies $A^2\approx B^2\approx \frac{1}{2}$ by Lemma \ref{xiajie}. Thus,
	\begin{align}\label{2110}
\begin{split}
		|\langle\psi_1^i, V' \Psi_1^i\rangle|&=|AB \frac{d}{d \theta} E_0^i-AB \frac{d}{d \theta} \tilde{E}_0^i+O(\delta_0^2 )| \\
	& \geq2AB |\frac{d}{d\theta} E_0^i|-O(\delta_0^{1 / 2} )\\
	&\geq\frac{1}{2}\delta_0^{1/9}.
\end{split}
	\end{align}
By Proposition \ref{55} \textbf{(a)}, we have $|E_1^i-\mathcal{E}_1^i|\leq100M_1\delta_0^{1/2}$. Combining  \eqref{df}, we obtain  $|\frac{d}{d\theta}E_1^i(\theta)|\geq\frac{1}{4}\delta_0^{2/9}(100M_1\delta_0^{1/2})^{-1}-O(\delta_0^{-1/8})>\delta_0^{-1/4}$, whose sign is determined by that of $ E_1^i(\theta)- \mathcal{E}_1^i(\theta)$.
\end{proof}
\begin{rem}
	We will see in the proof of Theorem \ref{C1} that under the hypothesis of $|\frac{d}{d\theta}E_1^i(\theta)|<\delta_0^{1/4}$ for some $|\theta-\theta^*|<10\delta_0^{1/2}$, then $E_1^i(\theta)\neq \mathcal{E}_1^i(\theta)$ for all $|\theta-\theta^*|<10\delta_0^{1/2}$.
\end{rem}
\begin{rem}\label{diyici}
	From $H_{B_1^i}(\theta)=H_{B_1^i}(2\theta_s-\theta)$, we deduce that the union of two eigenvalue curves is symmetric about $\theta_s$ for  $|\theta-\theta^*|<10\delta_0^{1/2}$. Moreover, if there is no eigenvalue level crossing, then each curve itself is symmetric.
\end{rem}

We are ready to prove the {\bf Center Theorem} for $n=1$ in {\bf Case 2}.

\begin{thm}[]\label{C1}
If $c_1^i,c_1^j\in Q_1$, then
 \begin{align}\label{chang}
	m(c_1^i,c_1^j)&\leq\sqrt{2}\min(|E_1^i(\theta^*)-E_1^j(\theta^*)|^{1/2},|\mathcal{E}_1^i(\theta^*)-\mathcal{E}_1^j(\theta^*)|^{1/2},\nonumber\\
	&|E_1^i(\theta^*)-\mathcal{E}_1^j(\theta^*)|^{1/2},|\mathcal{E}_1^i(\theta^*)-E_1^j(\theta^*)|^{1/2})\\\nonumber &\leq2\delta_1^\frac{1}{2}.
\end{align}
\end{thm}
\begin{proof}
Applying Lemma \ref{514} gives us a preliminary bound 
 \begin{equation}\label{preb}
	m(c_1^i,c_1^j)\leq6\delta_0^{1/2},
\end{equation}
which implies that $\theta^*\pm m(c_1^i,c_1^j)$ belongs to  the interval of $|\theta-\theta^*|<10\delta_0^{1/2}$, where $E_1^i$ and $\mathcal{E}_1^i$ are well defined. Recall the definition of $\theta_s$ (cf. \eqref{duichen}),
to establish {\bf Center Theorem}, we consider two cases.\\

	\textbf{Case I}.  $E_1^i(\theta_s) \neq \mathcal{E}_1^i(\theta_s)$ (cf.  FIGURE 1). 
	\begin{center}
\begin{tikzpicture}[relative]
\draw [thick] plot [smooth] coordinates { (0,0) (1,1) (3,2) (5,1) (6,0)};
\draw [thick] plot [smooth] coordinates { (0,5) (1,4) (3,3) (5,4) (6,5)};
\draw (3,-2) node[]   {FIGURE 1.};
\draw [thick](0,-1) -- (6,-1);
\draw (3,-1) node[below]   {$\theta_s$};
\draw (3.8,-1) node[below]   {$\theta^*$};
\draw [dashed] (3.8,-1) -- (3.8,1.76);
\draw (4.8,-1) node[below]   {$\theta^*+h$};
\draw [dashed] (4.8,-1) -- (4.8,1.15);
\draw (3,-1) node[circle]{};
\draw [dashed] (3,-1) -- (3,3);
\draw (5,3.5) node[below]   {$E_1^i(\theta)$};
\draw (5,2) node[below]   {$\mathcal{E}_1^i(\theta)$};
\end{tikzpicture}
\end{center}
	Without loss of generality,  we may assume $E_1^i(\theta_s)>\mathcal{E}_1^i(\theta_s)$. We must have by Remark \ref{diyici},
	$$E_1^i(\theta_s+\Delta \theta)=E_1^i(\theta_s-\Delta \theta),\ \mathcal{E}_1^i(\theta_s+\Delta \theta)=\mathcal{E}_1^i(\theta_s-\Delta \theta)
	$$
	for $\Delta \theta$ small. Therefore,
	$$
	\frac{d}{d \theta} E_1^i(\theta_s)=\frac{d}{d \theta} \mathcal{E}_1^i(\theta_s)=0 .
	$$ 
	By Proposition \ref{guji} (cf. \textbf{(b)} and \textbf{(c)}), we see that $\theta_s$ is a local minimum point of $ E_1^i$ (resp. a local maximum point of $ \mathcal{E}_1^i$). It follows that  $\frac{d}{d \theta} E_1^i$ is increasing and $\frac{d}{d \theta} \mathcal{E}_1^i$ is decreasing whenever $|\frac{d}{d \theta} E_1^i|\leq\delta_0^{1/4}$.  Thus,  $E_1^i>\mathcal{E}_1^i$ continues to hold for all $|\theta-\theta^*|<10\delta_0^{1/2}$, which implies that $\frac{d^2}{d\theta^2}E_1^i(\theta)>2$ whenever $|\frac{d}{d\theta}E_1^i(\theta)|<\delta_0^{1/4}$. Moreover,  $\frac{d}{d\theta}E_1^i$ (resp. $\frac{d}{d\theta}\mathcal{E}_1^i$) cannot reenter the band $|\frac{d}{d\theta}E|<\delta_0^{1/4}$ since it is increasing (resp. decreasing) there.  From  the preliminary bound  \eqref{preb}, we deduce that    both ${E}_1^i(\theta)$ and $\mathcal{E}_1^i(\theta)$ satisfy the condition of  Lemma \ref{C2} with $\theta_2=\theta^*+h,\theta_1=\theta^*,\delta=\delta_0^{1/4},|h|\leq\delta.$ Thus, we get
	$$|E_1^i(\theta^*+h)-E_1^i(\theta^*)|\geq \frac{1}{2}\min(h^2,|2\theta^*+h-2\theta_s|^2)=\frac{1}{2}h^2$$
	and the same estimate holds  true for $\mathcal{E}_1^i$,
	where $h=(c_1^j-c_1^i)\cdot \omega$ or $-((c_1^i+c_1^j)\cdot \omega+2\theta^*)$  {\rm (mod $1$)} satisfying $|h|=m(c_1^i,c_1^j)$.
	An easy inspection gives us 
	\begin{align*}
	|\mathcal{E}_1^i(\theta^*+h)-E_1^i(\theta^*)|&\geq\min(|E_1^i(\theta^*+h)-E_1^i(\theta^*)|,|\mathcal{E}_1^i(\theta^*+h)-\mathcal{E}_1^i(\theta^*)|)\\
	&\geq\frac{1}{2}h^2,\\
	|E_1^i(\theta^*+h)-\mathcal{E}_1^i(\theta^*)|&\geq\min(|E_1^i(\theta^*+h)-E_1^i(\theta^*)|,|\mathcal{E}_1^i(\theta^*+h)-\mathcal{E}_1^i(\theta^*)|)\\
	&\geq\frac{1}{2}h^2.
	\end{align*}
	Now \eqref{chang} follows from  $\{E_1^j(\theta^*),\mathcal{E}_1^j(\theta^*)\}=\{E_1^i(\theta^*+h),\mathcal{E}_1^i(\theta^*+h)\}$, since $H_{B_1^i}(\theta^*+h)=H_{B_1^j}(\theta^*)$, and one of the eigenvalue differences must be bounded above by $2\delta_1$ from the definition of $Q_1$. This proves the theorem. \\
	
	\textbf{Case II}. $E_1^i(\theta_s)= \mathcal{E}_1^i(\theta_s)$ (cf. FIGURE 2). 	
	\begin{center}
\begin{tikzpicture}[relative]
\draw (3,-2) node[]   {FIGURE 2.};
\draw[thick] (0,-1) -- (6,-1);
\draw (3,-1) node[below]   {$\theta_s$};
\draw (3.8,-1) node[below]   {$\theta^*$};
\draw [dashed] (3.8,-1) -- (3.8,1.78);
\draw (4.8,-1) node[below]   {$\theta^*+h$};
\draw [dashed] (4.8,-1) -- (4.8,1.2);
\draw (3,-1) node[circle]{};
\draw [dashed] (3,-1) -- (3,3);
\draw (5,4.5) node[below]   {$E_1^i(\theta)$};
\draw (5,2) node[below]   {$\mathcal{E}_1^i(\theta)$};
\draw[thick] (0,5) to [out=30,in=150] (3,3);
\draw[thick] (3,3) to [out=-30,in=-150] (6,1);
\draw[thick] (0,1) to [out=-30,in=-150] (3,3);
\draw[thick] (3,3) to [out=30,in=150] (6,5);

\end{tikzpicture}
\end{center}
	This means the level crossing occurs. In this case,  we claim that $|\frac{d}{d \theta} E_1^i|\geq\delta_0^{1/4}$ and $|\frac{d}{d \theta} \mathcal{E}_1^i|\geq\delta_0^{1/4}$ hold for $|\theta-\theta^*|<10\delta_0^{1/2}.$ Moreover, they have opposite signs. First,  we show that it is true for $\theta=\theta_s$. Since $E_1^i(\theta_s)$ is not simple, the first order eigenvalue perturbation formula in Theorem \ref{daoshu} can not be used. However, we still can compute $\frac{d}{d \theta} E_1^i(\theta_s),\frac{d}{d \theta} \mathcal{E}_1^i(\theta_s)$ by the remarkable result originated from Kato \cite{Kat95}. 
	
\begin{lem}\label{daogroup}
	The derivative group  $\{\frac{d}{d \theta} E_1^i(\theta_s),\frac{d}{d \theta} \mathcal{E}_1^i(\theta_s)\}$ of the non-simple eigenvalue $E_1^i(\theta_s)$ is equal to the eigenvalues of $PH'_{B_1^i}(\theta_s)P$, where $H'$ is the derivative of the self-adjoint operator $H$ and  $P$ is the total projection onto the two dimensional  eigenspace of $E_1^i(\theta_s)$. Namely,
	$$\big\{\frac{d}{d \theta} E_1^i(\theta_s),\frac{d}{d \theta} \mathcal{E}_1^i(\theta_s)\big\}=\{\text{Eigenvalues of the }2\times2\text{ matrix } PH'_{B_1^i}(\theta_s)P\}.$$
\end{lem} 

\begin{proof}The ideas of the proof come from  Theorem 5.4 in \cite{Kat95}.
		It suffices to show 
	$$ E_1^i(\theta)=E_1^i(\theta_s)+\lambda_1(\theta-\theta_s)+o(\theta-\theta_s),$$
	$$ \mathcal{E}_1^i(\theta)=\mathcal{E}_1^i(\theta_s)+\lambda_2(\theta-\theta_s)+o(\theta-\theta_s),$$
	where $\lambda_1,\lambda_2$ are the eigenvalues of $PH'_{B_1^i}(\theta_s)P$.
	Denote $P(\theta)=\int_{\Gamma}(\zeta-H_{B_1^i}(\theta))^{-1}d\zeta$ the $(C^2)$ total projection on the eigenvalue group $\{ E_1^i(\theta),\mathcal{E}_1^i(\theta)\}$, where $\Gamma$ is a small circle centered at $E_1^i(\theta_s)$ such that $E_1^i(\theta_s)$ is the unique eigenvalue of $H_{B_1^i}(\theta_s)$ inside $\Gamma$ and $\Gamma\cap\sigma(H_{B_1^i}(\theta))=\emptyset$ for all $\theta$ in a small neighborhood of $\theta_s$. Thus, the eigenvalue group  is just the eigenvalue of $P(\theta)H_{B_1^i}(\theta)P(\theta)$ restricting to the small neighborhood of $\theta_s$, namely, 
		$$\{ E_1^i(\theta),\mathcal{E}_1^i(\theta)\}=\{\text{Eigenvalues of the }2\times2\text{ matrix } P(\theta)H_{B_1^i}(\theta)P(\theta)\}\text{ for $\theta$ near $\theta_s$}.$$ 
	 Denote $E=E_1^i(\theta_s)=\mathcal{E}_1^i(\theta_s)$.  Then
	$$\{ E_1^i(\theta)-E,\mathcal{E}_1^i(\theta)-E\}=\{\text{Eigenvalues of the }2\times2\text{ matrix } P(\theta)(H_{B_1^i}(\theta)-E)P(\theta)\}.$$ 
	To finish the proof, it remains to show $(P(\theta)(H_{B_1^i}(\theta)-E)P(\theta))/(\theta-\theta_s)\to PH'_{B_1^i}(\theta_s)P$ as $\theta\to\theta_s$. Direct  computation gives  
	\begin{align*}
		 &\ \ \ \lim_{\theta\to \theta_s}\frac{P(\theta)(H_{B_1^i}(\theta)-E)P(\theta)}{\theta-\theta_s}\\
		&=P'(\theta_s)(H_{B_1^i}(\theta_s)-E)P(\theta_s)\\
		&\ \ \ +P(\theta_s)H'_{B_1^i}(\theta_s)P(\theta_s)+P(\theta_s)(H_{B_1^i}(\theta_s)-E)P'(\theta_s)\\
		&=P(\theta_s)H'_{B_1^i}(\theta_s)P(\theta_s),
	\end{align*} where we have used $(H_{B_1^i}(\theta_s)-E)P(\theta_s)=P(\theta_s)(H_{B_1^i}(\theta_s)-E)=0$.
\end{proof}

  To calculate these eigenvalues, we represent $PV'P:=PH'P$ in a special basis. Notice that $H_{B_1^i}(\theta_s)$ commutes with the reflect operator $(R\psi)(x):=\psi(2c_1^i-x)$. It follows that $\operatorname*{Range}P$ is a two dimensional  invariant subspace of $R$, which can be spanned by two eigenfunctions of $R$ since $R$ is diagonalizable. All the eigenfunctions of $R$ are symmetric functions $\{\psi_s\}$ and antisymmetric functions $\{\psi_a\}$. We note that   $\operatorname*{Range}P$ cannot be spanned by only symmetric functions (resp. antisymmetric functions), otherwise $\psi_{1}$ and $\Psi_{1}$ are  symmetric (resp. antisymmetric), contradicting  the expression \eqref{y}. This allows us to express $PV'P$ in the basis $\{\psi_s,\psi_a\}$, which consists of one symmetric function and one  antisymmetric function: 
  $$P V' P=\left(\begin{array}{ccc}
  	\left\langle\psi_s, V' \psi_s\right\rangle & \left\langle\psi_s, V' \psi_a\right\rangle \\
  	\left\langle\psi_s, V' \psi_a\right\rangle & \left\langle\psi_a, V' \psi_a\right\rangle
  \end{array}\right)\ (\text {at } \theta=\theta_s) .
  $$
  Since $v$ is even and $1$-periodic, we deduce that $(V'(\theta_s))(2c_1^i-x)=v'(\theta_s+(2c_1^i-x)\cdot \omega)=-v'(\theta_s+x\cdot \omega)=-(V'(\theta_s))(x)$, yielding $V'(\theta_s)$ is antisymmetric.  By the symmetry and anti-symmetry  properties of $\psi_s, \psi_a$ and $V'(\theta_s)$, we have $\langle\psi_s, V' \psi_s\rangle=$ $\langle\psi_a, V' \psi_a\rangle=0$, which gives us
  $$
  P V' P=\left(\begin{array}{cc}
  	0 & \left\langle\psi_s, V' \psi_a\right\rangle \\
  	\left\langle\psi_s, V' \psi_a\right\rangle & 0
  \end{array}\right)
  $$
  and therefore,
  $$
  \frac{d}{d\theta} E_1^i(\theta_s)=-\frac{d}{d\theta} \mathcal{E}_1^i(\theta_s)=\langle\psi_s, V' \psi_a\rangle.
  $$
  We choose $E_1^i$ to satisfy $\frac{d}{d\theta} E_1^i(\theta_s)\geq0$ and  will show that it is not too small and then extend this for $|\theta-\theta^*| \leq 10 \delta_0^{1 / 2}$. Using the symmetry properties and the decay of the eigenfunctions, we have
 $\frac{d}{d\theta} E_1^i(\theta_s)=2 \psi_s(c_0^i) \psi_a(c_0^i)\frac{d}{d\theta} E_0^i(\theta_s)+O(\delta_0^2)$, 
  where $|\psi_s(c_0^i)| \approx 1 / \sqrt{2}$ and $|\psi_a(c_0^i)| \approx 1 / \sqrt{2}$. By Lemma \ref{xiajie}, we get
  $$
  \frac{d}{d\theta}E_1^i(\theta_s) \geq \delta_0^{1 / 4}.
  $$
  We now  show that this continues to hold for all $\theta$ in the interval $|\theta-\theta^*| \leq 10 \delta_0^{1 / 2}$. Since $E_1^i$ is increasing and $\mathcal{E}_1^i$ is decreasing, we deduce  $E_1^i>\mathcal{E}_1^i$ for $\theta>\theta_s$.  If $\frac{d}{d\theta}E_1^i(\theta)\leq \delta_0^{1 / 4} $ for some smallest  $\theta>\theta_s$, by Proposition \ref{guji} (cf. \textbf{(c)}), we have  $\frac{d^2}{d\theta^2}E_1^i(\theta)>2$. This is impossible. The same argument shows that there is no $\theta<\theta_s$ such that $\frac{d}{d\theta}E_1^i(\theta)\leq \delta_0^{1 / 4} $, which proves our claim. In this case, we have  $E_1^i(\theta)=\mathcal{E}_1^i(2\theta_s-\theta)$ by the symmetry property of the eigenvalue curve. Thus, by the preliminary bound  \eqref{preb}, we obtain 
  \begin{align*}
  |E_1^i(\theta^*+h)-E_1^i(\theta^*)|&\geq \delta_0^{1/4}|h|\geq h^2,\\
   |\mathcal{E}_1^i(\theta^*+h)-\mathcal{E}_1^i(\theta^*)&|\geq \delta_0^{1/4}|h|\geq h^2,\\
  |E_1^i(\theta^*+h)-\mathcal{E}_1^i(\theta^*)|&=|E_1^i(\theta^*+h)-E_1^i(2\theta_s-\theta^*)|\\
  &\geq \delta_0^{1/4}|2\theta^*+h-2\theta_s|\geq h^2,\\
  |\mathcal{E}_1^i(\theta^*+h)-E_1^i(\theta^*)|&=|\mathcal{E}_1^i(\theta^*+h)-\mathcal{E}_1^i(2\theta_s-\theta^*)|\\
  &\geq \delta_0^{1/4}|2\theta^*+h-2\theta_s|\geq h^2,
  \end{align*}
  where $h=(c_1^j-c_1^i)\cdot \omega$ or $-((c_1^i+c_1^j)\cdot \omega+2\theta^*)$ {\rm (mod $1$)} satisfying $|h|=m(c_1^i,c_1^j)$. Now \eqref{chang}  follows from $E_1^j(\theta^*),\mathcal{E}_1^j(\theta^*)=E_1^i(\theta^*+h)\text{ or  }\mathcal{E}_1^i(\theta^*+h)$ and one of the eigenvalue differences must be bounded above by $2\delta_1$ by the definition of $Q_1$. This finishes the proof of Theorem \ref{C1}.
\end{proof}
 
We end the discussions of \textbf{Case 2} with  two  theorems, which are  significant in the follow-up  inductive process. 

\begin{thm}\label{dao2}
		For $|\theta-\theta^*|<10\delta_0^{1/2}$, we have 
		$$|\frac{d}{d \theta} E_1^i(\theta)|\geq\min(\delta_0^2,|\theta-\theta_s|).$$
\end{thm}
\begin{proof}We consider two cases.\smallskip\\
	\textbf{Case I}. $E_1^i(\theta_s)> \mathcal{E}_1^i(\theta_s)$. It immediately follows from Lemma \ref{C2} and  \textbf{(c)} in Proposition \ref{guji}.\\
		\textbf{Case II}. $E_1^i(\theta_s)= \mathcal{E}_1^i(\theta_s)$. In this case,  we have $|\frac{d}{d \theta} E_1^i(\theta)|\geq\delta_0^{1/4}\geq \delta_0^2.$
\end{proof}
From the proof of  Theorem \ref{C1}, we see that the eigenvalues $E_1^i(\theta)$ and $\mathcal{E}_1^i(\theta)$ may cross only at the symmetry point $\theta_s$  (\textbf{Case II}), and their separation distance grows as $\theta$ moves away from $\theta_s$. These observations were originated from \cite{Sur90}.  The following theorem  gives us a lower bound of the separation distance. 
 \begin{thm}
If $c_1^i \in Q_1$, then
    $$
 	|E_1^i(\theta)-\mathcal{E}_1^i(\theta)| \geq \delta_0^2 |\theta-\theta_s|
 	$$
 	for all $\theta$ in the interval of  $|\theta-\theta^*| \leq 10 \delta_0^{1 / 2}$.
 \end{thm}
\begin{proof}
 We  consider two cases.\smallskip\\
\textbf{Case I.} $E_1^i(\theta_s)>\mathcal{E}_1^i(\theta_s)$. Then 
$$\frac{d}{d\theta}E_1^i(\theta_s)=\frac{d}{d\theta}\mathcal{E}_1^i(\theta_s)=0$$
and by \eqref{2110}, 
$$
|\langle\psi_1^i, V' \Psi_1^i\rangle(\theta_s)| > \delta_0^{1 / 8}.
$$
Therefore, there must be a largest  interval $\theta_s \leq \theta \leq \theta_d$  on which  $|\langle\psi_1^i, V' \Psi_1^i\rangle(\theta)| \geq \delta_0^{1 / 8} $. If $\theta$ is in this interval, then
\begin{align*}
	(E_1^i-\mathcal{E}_1^i)(\theta)&= (E_1^i-\mathcal{E}_1^i)(\theta_s)+\frac{d}{d \theta}(E_1^i-\mathcal{E}_1^i)(\theta_s)\cdot (\theta-\theta_s) \\
	&\ \ \ +\frac{1}{2} \frac{d^2}{d \theta^2}(E_1^i-\mathcal{E}_1^i)(\xi)\cdot (\theta-\theta_s)^2\\
	&\geq\frac{1}{2} \frac{d^2}{d \theta^2}(E_1^i-\mathcal{E}_1^i)(\xi)\cdot (\theta-\theta_s)^2.
\end{align*}
By \eqref{df} and \eqref{ddf}, we have since $(E_1^i-\mathcal{E}_1^i)(\theta)= O(\delta_0^{1/2})$
\begin{align*}
	\frac{d^2}{d \theta^2}(E_1^i-\mathcal{E}_1^i)(\xi) & =\frac{4\langle\psi_1^i, V' \Psi_1^i\rangle^2(\xi)}{(E_1^i-\mathcal{E}_1^i)(\xi)}+O(\delta_0^{-1/8}) \\
	& \geq\frac{2(\delta_0^{1 / 8} )^2}{(E_1^i-\mathcal{E}_1^i)(\theta)},
\end{align*}
 which implies
$$
(E_1^i-\mathcal{E}_1^i)(\theta) \geq  \frac{\delta_0^{1 /4 } }{(E_1^i-\mathcal{E}_1^i)(\theta)}(\theta-\theta_s)^2
$$
and proves the theorem.

We now consider the case when $\theta \geq \theta_d$. By the argument in the proof of  Theorem \ref{C1} (cf. \textbf{Case I}), we have 
$$
\frac{d}{d \theta} E_1^i \geq \delta_0^{1 / 4} \text{ and }  \frac{d}{d \theta} \mathcal{E}_1^i \leq-\delta_0^{1 / 4}
$$
assuming  $\theta\geq\theta_d$, which gives us 
\begin{align*}
	(E_1^i-\mathcal{E}_1^i)(\theta) & =(E_1^i-\mathcal{E}_1^i)(\theta_d)+\frac{d}{d \theta}(E_1^i-\mathcal{E}_1^i)(\xi)\cdot (\theta-\theta_d) \\
	& \geq(E_1^i-\mathcal{E}_1^i)(\theta_d)+2 \delta_0^{1 / 4} (\theta-\theta_d)\\
	&\geq\delta_0^{1 / 8} (\theta_d-\theta_s)+2 \delta_0^{1 / 4} (\theta-\theta_d)\\
	&\geq\delta_0^2 (\theta-\theta_s).
\end{align*}

\noindent\textbf{Case II}. $E_1^i(\theta_s)= \mathcal{E}_1^i(\theta_s)$.
In this case,   we have $\frac{d}{d \theta} E_1^i\geq\delta_0^{1/4}$ and $\frac{d}{d \theta} \mathcal{E}_1^i\leq-\delta_0^{1/4}$, thus,
\begin{align*}
	|(E_1^i-\mathcal{E}_1^i)(\theta)| & =|(E_1^i-\mathcal{E}_1^i)(\theta_s)+\frac{d}{d \theta}(E_1^i-\mathcal{E}_1^i)(\xi)\cdot (\theta-\theta_s)| \\
	& \geq 2 \delta_0^{1 / 4} |\theta-\theta_s|.
\end{align*}
\end{proof}

Finally, we give  estimates  on  the Green's function restricted to $1$-good sets  by using the resolvent identity. 
\begin{thm}[] \label{1g}If $\Lambda$ is $1$-good, then for all $|\theta-\theta^*|<\delta_1/(10M_1)$ and $ |E-E^*|<\delta_1/5$, 
\begin{align*}
	\|G_\Lambda(\theta;E)\|&\leq10\delta_1^{-1},\\
	|G_\Lambda(\theta;E)(x,y)&<e^{-\gamma_1\|x-y\|_1}\ {\rm for}\  \|x-y\|_1\geq l_1^\frac{5}{6},
\end{align*}
 where $\gamma_1=(1-O(l_1^{-\frac{1}{30}}))\gamma_0.$
\end{thm}
\begin{proof} The proof is based on the application of resolvent identity,  which can be divided into three steps.

 First,  we prove the case when $\Lambda=B_1^i$ is a $1$-regular block. By the definition of $1$-regular block, we have 
\begin{align*}
		\|G_{B_1^i}(\theta^*;E^*)\|\leq\delta_1^{-1}.
\end{align*}
Hence,  by the Neumann series argument, for $|\theta-\theta^*|<\delta_1/(10M_1)$ and  $|E-E^*|<\frac{2}{5}\delta_1$, we have 
	\begin{equation*}
		\|G_{B_1^i}(\theta;E)\|\leq2\delta_1^{-1}.
	\end{equation*}
For convenience, we omit the dependence of Green's  functions on  $\theta$ and  $E$. Let $x,y\in B_1^i$ satisfy $\|x-y\|_1\geq l_1^\frac{4}{5}$. Since $G_{B_1^i}$ is self-adjoint,  we may  assume $\|x-c_1^i\|_1\geq l_1^\frac{3}{4}$. Let $I_1^i$ be the $l_1^\frac{2}{3}$-size cube centered at $c_1^i$.  Then $B_1^i\setminus I_1^i$ is $0$-good and hence we have estimates \eqref{542} and \eqref{542.} for  its Green's function. Using the resolvent identity shows
 \begin{align*}
 	|G_{B_1^i}(x,y)|&=|G_{B_1^i\setminus I_1^i}(x,y)\chi(y)+\sum_{z,z'}G_{B_1^i\setminus I_1^i}(x,z)\Gamma_{z,z'}G_{B_1^i}(z',y)|\\
 	&\leq e^{-\gamma_0\|x-y\|_1}+C(d)\sup_{z,z'}e^{-\gamma_0\|x-z\|_1}|G_{B_1^i}(z',y)|\\
 	&\leq e^{-\gamma_0\|x-y\|_1}+C(d)\sup_{z,z'}e^{-\gamma_0\|x-z\|_1}e^{-\gamma_0(\|z'-y\|_1-l_1^\frac{3}{4})}\delta_1^{-1}\\
 	&\leq e^{-\gamma'_0\|x-y\|_1}
 \end{align*}
with  $\gamma'_0=(1-O(l_1^{-\frac{1}{30}}))\gamma_0$, where we have used  the fact that for $\|z'-y\|_1\leq l_{1}^\frac{3}{4}$, 
\begin{equation*}
	|G_{B_{1}^i}(z',y)|\leq 	\|G_{B_{1}^i}\|\leq2\delta_{1}^{-1}\leq 2e^{-\gamma_0(\|z'-y\|_1-l_{1}^\frac{3}{4})}\delta_{1}^{-1}, 
\end{equation*}
and  for  $\|z'-y\|_1\geq l_{1}^\frac{3}{4}$,  
\begin{align*}
	|G_{B_{1}^i}(z',y)|&=|G_{B_{1}^i}(y,z')|\\
	&\leq \sum_{w,w'}|G_{B_1^i\setminus I_1^i}(y,w)\Gamma_{w,w'}G_{B_{1}^i}(w',z')|\\
	&\leq C(d)e^{-\gamma_0\|y-w\|_1}\|G_{B_{1}^i}\| \\
	&\leq C(d)	e^{-\gamma_0(\|z'-y\|_1-l_{1}^\frac{3}{4})}\delta_{1}^{-1}
\end{align*} and  eventually $\delta_1^{-1}=e^{l_1^\frac{2}{3}}\ll e^{\gamma_0\|x-y\|_1}$.

Second,  we establish the upper bound on norms of  Green's functions on general $1$-good sets. Now  assume that $\Lambda$ is an arbitrary $1$-good set. So, all the  blocks $B_1^i$ inside $\Lambda$ are $1$-regular by the definition of $1$-good sets. We must show that $G_\Lambda$ exists. 
By  the Schur's test, it suffices to show
\begin{equation}\label{Schur}
	\sup_x\sum_{y}|G_\Lambda(\theta; E+i0)(x,y)|<C<\infty.
\end{equation}
Denote $P'_1=\{c_1^i\in P_1:\ B_1^i\subset\Lambda\} \text{ and } \Lambda'=\Lambda\setminus\cup_{c_1^i\in P'_1} I_1^i$. Then $\Lambda'$ is $0$-good since $Q_0$ is contained in  the square root-size kernel in $B_1^i\ (c_1^i\in P_1)$ by our construction.  For $x\in \Lambda\setminus\cup_{c_1^i\in P'_1} 2I_1^i$ ($2I_1^i$ denotes a $2l_1^\frac{2}{3}$-size cube centered at $c_1^i$), we have by using the resolvent  identity
\begin{align*}
	\sum_y|G_\Lambda(x,y)|&\leq \sum_y|G_{\Lambda'}(x,y)|+\sum_{z,z',y}|G_{\Lambda'}(x,z)\Gamma_{z,z'}G_{\Lambda}(z',y)|\\
	&\leq	C(d)\delta_0^{-1}+	C(d)e^{-l_1^\frac{2}{3}}\sup_{z'}\sum_y|G_{\Lambda}(z',y)|.
\end{align*}
For $x\in 2I_1^i$, we have also by using the resolvent  identity
\begin{align*}
	\sum_y|G_\Lambda(x,y)|&\leq \sum_y|G_{B_1^i}(x,y)|+\sum_{z,z',y}|G_{B_1^i}(x,z)\Gamma_{z,z'}G_{\Lambda}(z',y)|\\
	&\leq \delta_1^{-2}+C(d)e^{-\frac{1}{2}l_1}\sup_{z'}\sum_y|G_{\Lambda}(z',y)|.
\end{align*}
By taking supremum in  $x$, we get $$\sup_x\sum_y|G_\Lambda(x,y)|\leq \delta_1^{-2}+\frac{1}{2}\sup_x\sum_y|G_\Lambda(x,y)|,$$ 
and then
$$\sup_x\sum_y|G_\Lambda(x,y)|\leq 2\delta_1^{-2},$$
which gives \eqref{Schur}.
So, it follows that for $|\theta-\theta^*|<\delta_1/(10M_1)$ and $|E-E^*|<\frac{2}{5}\delta_1$, $G_\Lambda(\theta;E)$ exists, from which we get  $\operatorname{dist}(\sigma(H_\Lambda(\theta)),E^*)\geq \frac{2}{5}\delta_1$ and hence $\operatorname{dist}(\sigma(H_\Lambda(\theta)),E)\geq \frac{1}{5}\delta_1$ for $|E-E^*|<\frac{1}{5}\delta_1$, giving  the desired  bound 
	$$\|G_\Lambda(\theta;E)\|=\frac{1}{\operatorname{dist}(\sigma(H_\Lambda(\theta)),E)}\leq10\delta_1^{-1}.$$

Finally,  we use the  bound above and the resolvent identity to prove the  exponential off-diagonal decay of Green's functions  via  the  standard iteration argument.
Let $x,y\in \Lambda$ such that $\|x-y\|_1\geq l_1^\frac{5}{6}$. We define 
	\[B_x=\left\{\begin{aligned}
	&\Lambda_{\l_1^\frac{1}{2}}(x)\cap\Lambda  \quad \text{if }   x\in \Lambda\setminus\cup_{c_1^i\in P'_1} 2I_1^i\text{ \ (Choice 1)},  \\
	&B_1^i\   \quad \quad\quad \quad\quad\quad \text{if }  x \in2I_1^i \text{\ (Choice 2)}. \ 
\end{aligned}\right. \]
The set $B_x$ has the following two properties: \textbf{(1)} $B_x$ is either   a   $0$-good  set  or a $1$-regular block; \textbf{(2)}  The $x$ is close to the center of $B_x$ and away from its relative boundary with $\Lambda$. So, we can iterate the resolvent identity  to obtain 
\begin{align}
	|G_\Lambda(x,y)|&\leq\prod_{s=0}^{L-1} (C(d) l_{1}^d e^{-\gamma_0'\|x_{s}-x_{s+1}\|_1})|G_\Lambda(x_L,y)|\nonumber \\
	&\leq e^{-\gamma_0''\|x-x_L\|_1}|G_\Lambda(x_L,y)|, \label{728}
\end{align}
where $x_0:=x$ and  $x_{s+1}\in \partial B_{x_{s}}$ ( $\|x_{s+1}-x_s\|_1\geq l_1^\frac{1}{2}$ in Choice 1 and  $\|x_{s+1}-x_s\|_1\geq \frac{1}{2}l_1$ in Choice 2). We can stop the iteration until  $y\in B_{x_L}$. 
Using the resolvent identity again, we get 
 \begin{align}
|G_\Lambda(x_L,y)|&\leq|G_{B_{x_L}}(x_L,y)|+\sum_{z,z'}|G_{B_{x_L}}(x_L,z)\Gamma_{z,z'}G_{\Lambda}(z',y)|\nonumber\\
&\leq C(d) e^{-\gamma_0'(\|x_L-y\|_1-l_1^\frac{4}{5})}\delta_1^{-1},\label{728.}
\end{align}
where we have used the  exponential off-diagonal decay of $G_{B_{x_L}}$ and the bound   $\|G_{\Lambda}\|\leq10\delta_1^{-1}$. Then \eqref{728} together with  \eqref{728.} gives the desired off-diagonal estimate
$$|G_\Lambda(x,y)|\leq e^{-\gamma_1\|x-y\|_1}$$
with $\gamma_1=(1-O(l_1^{-\frac{1}{30}}))\gamma_0$. We complete the proof.
\end{proof}

\subsection{Induction hypothesis}  Now, we can lay down the induction  hypothesis. 
We first list the most important properties of $Q_n$ in our induction hypothesis. Assume that $Q_{n-1}$ has been constructed,  and then we define $$s_{n-1}=\inf\big\{\|c_{n-1}^i-c_{n-1}^j\|_1:\ {c_{n-1}^i\neq c_{n-1}^j\in Q_{n-1} }\big\}.$$
Then we have two cases.\smallskip\\
\textbf{Case 1.}  $s_{n-1}\geq10l_{n-1}^2$. Then  $P_n$ consists of  the centers  of $n$-th stage resonant  blocks and  is defined to be $Q_{n-1}$. We associate every $c_n^i\in P_n$ a block $B_{n}^i$ satisfying
\begin{itemize}
 \item[{\bf (1)}] $\Lambda_{l_{n-1}^2}(c_n^i)	\subset B_{n}^i\subset	\Lambda_{l_{n-1}^2+50l_{n-1}}(c_n^i)$.
 \item[{\bf (2)}]  If $B_m^{j}\cap B_{n}^i\neq \emptyset\ (1\leq m<n)$, then $B_m^{j}\subset B_{n}^i$.
 \item[{\bf (3)}] $B_n^i$ is symmetric about $c_n^i$ (i.e., $x\in B_n^i\Rightarrow 2c_n^i-x\in B_n^i$).
 \item[{\bf (4)}]  The set $B_n^i-c_n^i$ is independent of $i$, i.e.   $B_n^j=B_n^i+(c_n^j-c_n^i)$.
\end{itemize}
 \textbf{Case 2.} $s_{n-1}<10l_{n-1}^2$. Then $P_n$ is defined as $$\{c_n^i=(c_{n-1}^i+\tilde{c}_{n-1}^i)/2:\ c_{n-1}\in Q_{n-1}\}, $$ where $\tilde{c}_{n-1}^i$ is the mirror image of $c_{n-1}^i$ satisfying $\|c_{n-1}^i-\tilde{c}_{n-1}^i\|_1=s_{n-1}$ and 	$\|2\theta^*+(c_{n-1}^i+\tilde{c}_{n-1}^i)\cdot \omega\|\leq6\delta_{n-1}^{1/2}$  (cf.  Lemma \ref{mirrr} for an analog). 
The  block $B_n^i$ is required to satisfy the same properties  as in \textbf{Case 1} except \textbf{(1)} replaced by $\Lambda_{l_{n-1}^4}(c_n^i)	\subset B_{n}^i\subset	\Lambda_{l_{n-1}^4+50l_{n-1}}(c_n^i)$. 

From the above constructions,  we have $B_n^i\cap B_n^j=\emptyset$ for $i\neq j$ in both cases and every singular block of stage $n-1$ is  contained in the  square root-size kernel of a unique block from stage $n$.  This is not a trivial  issue, which will be handled  in the Appendix \ref{AC}.  

Finally, the $n$-th stage singular points set  $Q_n$ is defined as 
$$Q_n=\big\{c_n^i\in P_n:\ \operatorname{dist}(\sigma(H_{B_n^i}(\theta^*)),E^*)<\delta_n:=e^{-l_n^{2/3}}\big\}.$$

Now, we assume that every $c_n^i\in Q_n$ belongs to the following either {\bf Class A} or {\bf B}: \\

\fbox{\bf Class A}: For every  $|\theta-\theta^*|<\delta_{n-1}/(10M_1)$, we have 
\begin{itemize}
	\item [\textbf{(H1)}] There is a unique eigenvalue $E_{n}^i(\theta)\in \sigma(H_{B_n^i}(\theta))$ satisfying   $|E_{n}^i(\theta)-E^*|<\delta_{n-1}/9$. Moreover, any other $\hat{E}\in\sigma(H_{B_n^i}(\theta)) $ must obey $|\hat{E}-E^*|\geq\delta_{n-1}/5$.
	\item [\textbf{(H2)}] The corresponding eigenfunction $\psi_n^i$ satisfies $|\psi_n^i(x)|\leq e^{-(\gamma_0/4)\|x-c_n^i\|_1}$ for $\|x-c_n^i\|_1\geq l_n^{6/7}.$
	\item [\textbf{(H3)}] If $|\frac{d}{d \theta} E_n^i(\theta)|\leq\delta_{n-1}^2$, then $|\frac{d^2}{d \theta^2} E_n^i(\theta)|\geq3-\sum_{l=0}^{n-1}\delta_l^3\geq2$ and $\frac{d^2}{d \theta^2} E_n^i(\theta)$ has a unique sign. 
	 \item [\textbf{(H4)}] There exists $\mu_n=0\text{ or }1/2$, such that for all $c_n^i$ belonging to {\bf Class A} and $|\theta-\theta^*|<\delta_{n-1}/(20M_1)$, $$|\frac{d}{d \theta} E_n^i(\theta)|\geq\min(\delta_{n-1}^2,\|\theta+c_n^i\cdot \omega-\mu_n\|).$$
	\item [\textbf{(H5)}]    If $c_n^j\in Q_n$, then  $$m(c_n^i,c_n^j)\leq\sqrt{2}|E_{n}^i(\theta^*)-E_{n}^i(\theta^*+h)|^{1/2}=\sqrt{2}|E_{n}^i(\theta^*)-E_{n}^j(\theta^*)|^{1/2}\leq2\delta_{n}^{1/2},$$
	where $h=(c_n^j-c_n^i)\cdot \omega$ or $-((c_n^i+c_n^j)\cdot \omega+2\theta^*)$  (mod $1$) satisfying $|h|=m(c_n^i,c_n^j)$.
\end{itemize}
\medskip

\fbox{\bf Class B}: For every  $|\theta-\theta^*|<10\delta_{n-1}^{1/2}$, we have \begin{itemize}
	\item [\textbf{(H6)}] There is $\mu_n=0\text{ or }1/2$, such that for all $c_n^i$ in {\bf Class B}, the symmetric point $\theta_n^i:=-c_n^i\cdot \omega+\mu_n$ (mod $1$) belongs to the interval  of $|\theta-\theta^*|<3\delta_{n-1}^{1/2}$.   
	\item [\textbf{(H7)}] There are exact two  eigenvalues $E_{n}^i(\theta),\mathcal{E}_{n}^i(\theta)\in \sigma(H_{B_n^i}(\theta))$ satisfying   $|E_{n}^i(\theta)-E^*|<50M_1\delta_{n-1}^{1/2}$ and $|\mathcal{E}_{n}^i(\theta)-E^*|<50M_1\delta_{n-1}^{1/2}$. Moreover, any other $\hat{E}\in\sigma(H_{B_n^i}(\theta)) $ must obey $|\hat{E}-E^*|\geq\delta_{n-2}/6$. ({\bf Note}: $\delta_{-1}=\delta_0^{1/8}$).
	\item [\textbf{(H8)}] The corresponding eigenfunction $\psi_n^i$ (resp. $\Psi_n$) for $E_n^i$  (resp. $\mathcal{E}_n^i$) satisfies  $|\psi_n^i(x)|\leq e^{-(\gamma_0/4)\|x-c_n^i\|_1}$ (resp.  $|\Psi_n^i(x)|\leq e^{-(\gamma_0/4)\|x-c_n^i\|_1}$)  for $\|x-c_n^i\|_1\geq l_n^{6/7}$.
	\item [\textbf{(H9)}]  If $|\frac{d}{d \theta} E_n^i(\theta)|\leq10\delta_{n-1}^{1/2}$, then $|\frac{d^2}{d \theta^2} E_n^i(\theta)|\geq3-\sum_{l=0}^{n-1}\delta_l^3\geq2$ and $\frac{d^2}{d \theta^2} E_n^i(\theta)$ has a unique sign.
	\item [\textbf{(H10)}] $|\frac{d}{d \theta} E_n^i(\theta)|\geq\min(\delta_{n-1}^2,|\theta-\theta_n^i|).$
	\item [\textbf{(H11)}] $|E_n^i(\theta)-\mathcal{E}_n^i(\theta)| \geq \delta_{n-1}^2 |\theta-\theta_n^i|$.
	\item [\textbf{(H12)}] If $c_n^j\in Q_n$, then $$\{E_{n}^i(\theta^*+h),\mathcal{E}_{n}^i(\theta^*+h)\} =\{E_{n}^j(\theta^*),\mathcal{E}_{n}^j(\theta^*)\},$$ 
		where $h=(c_n^j-c_n^i)\cdot \omega$ or $-((c_n^i+c_n^j)\cdot \omega+2\theta^*)$  (mod $1$) satisfying $|h|=m(c_n^i,c_n^j)$. Moreover, we have  $$m(c_n^i,c_n^j)\leq\sqrt{2}|E_{n}^i(\theta^*)-E_{n}^j(\theta^*)|^{1/2}.$$ The same estimate holds for $|E_{n}^i(\theta^*)-\mathcal{E}_{n}^j(\theta^*)|$, $|\mathcal{E}_{n}^i(\theta^*)-E_{n}^j(\theta^*)|$ and  $|\mathcal{E}_{n}^i(\theta^*)-\mathcal{E}_{n}^j(\theta^*)|$. 
\end{itemize}
\begin{rem}
\textbf{(H5)} and \textbf{(H12)} are stronger versions of {\bf Center Theorem}. The Hypotheses are still true if we enlarge  the $\theta$'s interval to  a $\delta_n^{1/2}$ size. Thus if one $c_n^i \in Q_n$ belongs to certain Class, then all the points in $Q_n$ belong to this Class.  However, the two Classes need  not  be  incompatible.
\end{rem}
We also assume that we have established Green's function estimates at stage $n$.  It remains to verify the induction hypothesis of the stage $n+1$, which will be finished in the following subsection. 

\subsection{Definition and properties of $Q_{n+1}$}
In this section, we will assume that the induction hypothesis is true at stage $l$ for $0\leq l\leq n$,  and then prove that it holds at stage $n+1$. We distinguish two cases.

\subsubsection{}{\bf Case 1}.  $s_n\geq10l_n^2$.  
In this case, $l_{n+1}=l_n^2$ and we define
$$Q_{n+1}=\big\{c_{n+1}^i\in P_{n+1}=Q_n:\ \operatorname{dist}(\sigma(H_{B_{n+1}^i}(\theta^*)),E^*)<\delta_{n+1}:=e^{-l_{n+1}^{2/3}}\big\}.$$ This case will be further distinguished  into two subcases, according to the number of eigenvalues of $H_{B_n^i}(\theta^*)$ that are near $E^*$. We list  all eigenvalues counting  multiplicities. The following notation ``$-$'' means deleting an  element from  the set. \\

\fbox{{\bf Subcase A}}.  We have  $c_{n+1}^i=c_n^i\in Q_{n+1}$ satisfying  
\begin{equation}\label{SA}
	\operatorname{dist}(\sigma(H_{B_n^i}(\theta^*))- E_n^i(\theta^*),E^*)>\delta_{n}.
\end{equation}
We will show  how to get back to  {\bf Class A}  of the  induction hypothesis from \textbf{Subcase A}.

\begin{prop}\label{kn+1}
	Assume that \eqref{SA} holds true.  Then  for $|\theta-\theta^*|<\delta_n/(10M_1)$,
	\begin{itemize}
		\item[\textbf{(a)}]  $H_{B_{n+1}^i}(\theta)$ has a unique eigenvalue $E_{n+1}^i(\theta)$ such that $|E_{n+1}^i(\theta)-E^*|<\delta_n/9$. Moreover, any other $\hat{E}\in\sigma(H_{B_{n+1}^i}(\theta)) $ must obey $|\hat{E}-E^*|>\delta_{n}/5$.
		\item[\textbf{(b)}] The corresponding eigenfunction of $E_{n+1}^i(\theta)$,  $\psi_{n+1}$  satisfies  
	\begin{equation}\label{dec}
			|\psi_{n+1}(x)|\leq e^{-(\gamma_0/4)\|x-c_{n+1}^i\|_1}\ {\rm for}\   \|x-c_{n+1}^i\|_1\geq l_n^{6/7}.
	\end{equation}
		\item[\textbf{(c)}] Let $\psi_{n}$ be  the eigenfunction of $E_n(\theta)$ for $H_{B_{n}^i}(\theta)$. Then 
		\begin{equation}\label{closed}
			\|\psi_{n+1}-\psi_{n}\|\leq \delta_{n}^{10}.
		\end{equation}
		\item[\textbf{(d)}] $\|G_{B_{n+1}^i}^\perp(E_{n+1}^i)\|\leq20\delta_n^{-1}$, where $G_{B_{n+1}^i}^\perp$ denotes the Green's function for $B_{n+1}^i$ on the orthogonal complement of $\psi_{n+1}$.
	\end{itemize} 
\end{prop}

\begin{proof}
	Since $ B_{n+1}^i$ is singular, by definition,  $H_{B_{n+1}^i}(\theta^*)$ has an eigenvalue $E_{n+1}^i(\theta^*)$ such that  $|E_{n+1}^i(\theta^*)-E^*|<\delta_{n+1}$. By $|V'|\leq M_1$,   $\sigma(H_{B_{n+1}^i}(\theta))$ and  $\sigma(H_{B_{n+1}^i}(\theta^*))$   differ at most $M_1|\theta-\theta^*|<\delta_n/10$, which shows the existence of $E_{n+1}^i(\theta)$ in $|E-E^*|<\delta_n/9$. Define  $\Lambda=B_{n+1}^i\setminus\hat{B}_n^i$, where  $\hat{B}_n^i$ is a  $O(l_n^{2/3})$-size block with the center $c_{n+1}^i$ so that $\Lambda$ is $n$-good.   Let $E\in \sigma(H_{B_{n+1}^i}(\theta))$ be such that $|E-E^*|<\delta_n/5$. We determine the value of $\psi_{n+1}(x)$ by 
	$$\psi_{n+1}(x)= \sum_{z,z'} G_{\Lambda}(\theta;E)(x, z) \Gamma_{z,z'} \psi_{n+1}(z').$$
	For  $\|x-c_{n+1}^i\|_1\geq l_n^{6/7}$, we have $\operatorname{dist}(x,\partial\hat{B}_n^i)\geq \|x-c_{n+1}^i\|_1-O(l_n^{2/3})\geq\frac{2}{3} \|x-c_{n+1}^i\|_1>l_n^{5/6}. $ Using the exponential off-diagonal  decay of  $G_\Lambda(\theta; E)$,  we get 
	\begin{align*}
	|\psi_{n+1}(x)|&\leq C(d)\sum_{z'\in \partial^+ \hat{B}_n^i}e^{-\frac{1}{3}\gamma_0 \|x-c_{n+1}^i\|_1}   |\psi_{n+1}(z')|\\
	&\leq e^{-\frac{1}{4}\gamma_0 \|x-c_{n+1}^i\|_1}.
	\end{align*}
	 Thus, we finish the proof of  \textbf{(b)}. To establish \textbf{(c)}, we must show that $\psi_{n+1}$ is close to  $\psi_{n}$ inside $B_n^i$. To see this, we restrict $H_{B_{n+1}^i}(\theta)\psi_{n+1}=E_{n+1}^i(\theta)\psi_{n+1}$ to $B_n^i$ to obtain
		 $$\left(H_{B_{n}^i}-E_{n+1}^i\right)\psi_{n+1}=\Gamma_{B_{n}^i}\psi_{n+1}.$$
	Combining \eqref{SA}, \eqref{dec} and the above equation, we get  $$\|P_n^\perp\psi_{n+1}\|=\|G_{B_n^i}^\perp(E_{n+1}^i)P_n^\perp\Gamma_{B_{n}^i}\psi_{n+1}\|= O(\delta_n^{-1}e^{-\frac{1}{4}\gamma_0l_n})<\frac{1}{2}\delta_n^{10},$$
   where $P_n^\perp$ is the projection onto the orthogonal complement of $\psi_n$ and $G_{B_n^i}^\perp(E_{n+1}^i)$ is the Green's function for $B_n^i$ on $\operatorname{Range}P_n^\perp$ with upper bound 
   \begin{align*}
   \|G_{B_n^i}^\perp(E_{n+1}^i)\|&\leq\operatorname{dist}(\sigma(H_{B_n^i}(\theta))- E_n^i(\theta),E_{n+1})^{-1}\\
   &\leq(\operatorname{dist}(\sigma(H_{B_n^i}(\theta^*))- E_n^i(\theta^*),E^*)-\delta_{n}/10-\delta_n/5)^{-1}\\
   &\leq2\delta_n^{-1}
   \end{align*}
   by the assumption \eqref{SA}. Since $\psi_{n+1}$ is normalized, we obtain $\|\psi_{n+1}-\psi_{n}\|\leq \delta_{n}^{10}$.
	  If there is another $\hat{E}\in \sigma(H_{B_{n+1}^i}(\theta))$ satisfying $|\hat{E}-E^*|\leq\delta_n/5$, the same argument shows that the corresponding eigenfunction $\hat{\psi}$  must also almost localize on $B_n^i$ and be close to $\psi_n$ inside $B_n^i$, which violates the  orthogonality. Thus, we prove the uniqueness part of \textbf{(a)}. Finally, \textbf{(d)} follows from the fact that any other $\hat{E}\in \sigma(H_{B_{n+1}^i}(\theta))$ must obey $|\hat{E}-E_{n+1}^i(\theta)|\geq |\hat{E}-E^*|-|E^*-E_{n+1}^i(\theta)|\geq\delta_n/5-\delta_n/9\geq\delta_n/20$.
\end{proof}
Next, we estimate the upper bounds on derivatives of eigenvalues parameterizations. 
	  \begin{prop}\label{apn}
	  	For $|\theta-\theta^*|<\delta_n/(10M_1)$, we have 
	  	$$|\frac{d^s}{d\theta^s}(E_{n+1}^i(\theta)-E_n^i(\theta))|\leq\delta_n^7\ \ {\rm for\ }s=0,1,2.$$
	  \end{prop}
	  \begin{proof}
        Using \eqref{closed}, we get 
	  	$$|E_{n+1}^i(\theta)-E_n^i(\theta)|=|\left\langle \psi_{n+1},H_{B_{n+1}^i}(\theta)\psi_{n+1}\right\rangle -\left\langle \psi_n,H_{B_n^i}(\theta)\psi_n\right\rangle|= O(\delta_n^{10})$$
	  	and $$|\frac{d}{d \theta} E_{n+1}^i(\theta)-\frac{d}{d \theta} E_n^i(\theta)|=|\left\langle\psi_{n+1},V' \psi_{n+1}\right\rangle-\left\langle\psi_n,V' \psi_n\right\rangle|=O(\delta_n^{10}).$$
	  	For $s=2$, we use the formulas from Theorem \ref{daoshu},
	  	\begin{align*}
		\frac{d^2}{d \theta^2} E_{n}^i(\theta)&=\left\langle\psi_{n}, V'' \psi_{n}\right\rangle-2\left\langle\psi_{n}, V' G_{B_{n}^i}^{\perp}(E_{n}^i) V' \psi_{n}\right\rangle,\\
	  	\frac{d^2}{d \theta^2} E_{n+1}^i(\theta)&=\left\langle\psi_{n+1}, V'' \psi_{n+1}\right\rangle-2\left\langle\psi_{n+1}, V' G_{B_{n+1}^i}^{\perp}(E_{n+1}^i) V' \psi_{n+1}\right\rangle. 
		\end{align*}
	  Thus, it suffices to estimate 
	  \begin{align*}
	  	&\Big|\left\langle\psi_{n+1}, V' G_{B_{n+1}^i}^{\perp}(E_{n+1}^i) V' \psi_{n+1}\right\rangle-\left\langle\psi_{n}, V' G_{B_{n}^i}^{\perp}(E_{n}^i) V' \psi_{n}\right\rangle\Big|\\
	  	\leq&\Big|\left\langle\psi_{n+1}, V' G_{B_{n+1}^i}^{\perp}(E_{n+1}^i) V' \psi_{n+1}\right\rangle-\left\langle\psi_{n}, V' G_{B_{n+1}^i}^{\perp}(E_{n+1}^i) V' \psi_{n}\right\rangle\Big|\\
		&\ \  +\Big|\left\langle\psi_{n}, V' G_{B_{n+1}^i}^{\perp}(E_{n+1}^i) V' \psi_{n}\right\rangle-\left\langle\psi_{n}, V' G_{B_{n}^i}^{\perp}(E_{n}^i) V' \psi_{n}\right\rangle\Big|\\
	  	\leq&\delta_n^8+\Big|\left\langle V'\psi_{n},  \left(G_{B_{n+1}^i}^{\perp}(E_{n+1}^i)-G_{B_{n}^i}^{\perp}(E_{n}^i)\right) V' \psi_{n}\right\rangle\Big|.
	  \end{align*}
  We must estimate 
  \begin{align}
  	\nonumber&\ \ \ \ G_{B_{n+1}^i}^{\perp}(E_{n+1}^i)-G_{B_{n}^i}^{\perp}(E_{n}^i)\\
	\nonumber&=G_{B_{n+1}^i}^{\perp}(E_{n+1}^i)P_n^\perp-P_{n+1}^\perp G_{B_{n}^i}^{\perp}(E_{n}^i)\\
	\nonumber&\ \ +G_{B_{n+1}^i}^{\perp}(E_{n+1}^i)P_n-P_{n+1}G_{B_{n}^i}^{\perp}(E_{n}^i)
	\\
  	\nonumber&=G_{B_{n+1}^i}^{\perp}(E_{n+1}^i)\left(-\Gamma_n +(E_{n+1}^i-E_n^i)\right)G_{B_{n}^i}^{\perp}(E_{n}^i)\\
&\ \ +G_{B_{n+1}^i}^{\perp}(E_{n+1}^i)P_{n+1}^\perp P_n-P_{n+1}P_n^\perp G_{B_{n}^i}^{\perp}(E_{n}^i) \label{545}
  \end{align}
restricted to $B_n^i$. This equation follows from the resolvent identity. We have used the orthogonal projections $P_n$ and $P_{n+1}$ onto $\psi_{n}$ and $\psi_{n+1}$ respectively and the relation $P_n+P_n^\perp=\operatorname{Id}_{n}$, $P_{n+1}+P_{n+1}^\perp=\operatorname{Id}_{n+1}$. The last two terms of \eqref{545} are bounded by $\delta_n^8$ using 
\begin{align*}
\|P_{n+1}^\perp P_n\|&=\|P_{n+1}^\perp\psi_{n}\|\leq2\|\psi_{n+1}-\psi_{n}\|\leq2\delta_n^{10},\\
\|P_{n+1}P_n^\perp\|&=\|P_n^\perp P_{n+1}\|=\|P_{n}^\perp\psi_{n+1}\|\leq2\|\psi_{n+1}-\psi_{n}\|\leq2\delta_n^{10},
\end{align*}
and (by the assumption \eqref{SA}) 
$$\|G_{B_{n}^i}^{\perp}(E_{n}^i)\| \ , \ \|G_{B_{n+1}^i}^{\perp}(E_{n+1}^i)\|= O(\delta_n^{-1}).$$
The second term on the right hand side of \eqref{545} is bounded by $O(\delta_n^{10})$ since $|E_{n+1}-E_n|\leq\delta_n^9$. Therefore, the case $s=2$ follows if we prove 
\begin{equation}\label{928}
	\|\Gamma_n G_{B_n^i}^\perp(E_n^i)V'\psi_{n}\|= O(\delta_n^9).
\end{equation}
Let $\chi_n$ be the characteristic function of the block $\Lambda_{\frac{l_n}{4}}(c_n^i)\subset B_n^i$.  By the estimate \eqref{dec}, we have 
$$\|(1-\chi_n)V'\psi_{n}\|= O(e^{-\frac{\gamma_0}{16}l_n})\leq\delta_n^{10}.$$
Thus, in order to prove \eqref{928}, it suffices to show $\|\Gamma_nG_{B_n^i}^\perp \chi_n\|= O(\delta_n^9)$. To do this, we choose a $O(l_n^{2/3})$-size block $\hat{B}$ with the center $c_n^i$ so that $A=B_n^i\setminus \hat{B}$ is $(n-1)$-good. Using the resolvent identity, we get 
\begin{align*}
	\Gamma_nG_{B_n^i}^\perp\chi_n&=\Gamma_nG_A\chi_n+\Gamma_n(\chi_AG_{B_n^i}^\perp-G_AP_n^\perp)\chi_n-\Gamma_nG_AP_n\chi_n\\
	&=\Gamma_nG_A\chi_n+\Gamma_nG_A\Gamma_AG_{B_n^i}^\perp\chi_n-\Gamma_nG_AP_n\chi_n.
\end{align*}
Since $A$ is $(n-1)$-good and $|E_n^i(\theta)-E^*|\leq2\delta_n$, we deduce that $\|G_A(E_n^i)\|\leq10\delta_{n-1}^{-1}$ and   $G_A(E_n^i)(x,y)$ decays exponentially fast for $\|x-y\|_1\geq l_{n-1}^\frac{5}{6}$. Thus, $\|\Gamma_nG_A\chi_n\|= O(l_ne^{-\frac{1}{5}\gamma_0l_n})\leq\delta_n^{10}$ and  $\|\Gamma_nG_A\Gamma_A\|= O(l_ne^{-\frac{1}{5}\gamma_0l_n})\leq\delta_n^{10}$. To estimate the final term, we use $\|P_n\chi_n-\chi_nP_n\|\leq \|P_n\chi_n-P_n\|+\|P_n-\chi_nP_n\|\leq2\|(1-\chi_n)P_n\|= 2\|(1-\chi_n)\psi_{n}\|\leq\delta_n^{10}$ to obtain 
$$\|\Gamma_nG_AP_n\chi_n\|\leq\|\Gamma_nG_A\chi_nP_n\|+\|\Gamma_nG_A\chi_n(P_n\chi_n-\chi_nP_n)\|= O(\delta_n^9).$$
	  \end{proof}
  We also have the transversality type estimates. 
  \begin{prop}\label{ds}
  If $|\frac{d}{d \theta} E_{n+1}^i(\theta)|\leq\delta_{n}^2$ for some $|\theta-\theta^*|<\delta_n/(10M_1)$, then $|\frac{d^2}{d \theta^2} E_{n+1}^i(\theta)|\geq3-\sum_{l=0}^{n}\delta_l^3\geq2$ and $\frac{d^2}{d \theta^2} E_{n+1}^i(\theta)$ has a unique sign.
  \end{prop}
  
  \begin{proof}
  	Assume $|\frac{d}{d \theta} E_{n+1}^i(\theta)|\leq\delta_{n}^2$. By Proposition \ref{apn}, we have \begin{equation}\label{51}
  		|\frac{d}{d \theta} E_{n}^i(\theta)|\leq|\frac{d}{d \theta} E_{n+1}^i(\theta)|+O(\delta_{n}^7)\leq2\delta_{n}^2.
  	\end{equation} 
	So, applying the induction hypothesis \textbf{(H3)} and \textbf{(H9)} gives 
  	$$|\frac{d^2}{d \theta^2} E_n^i(\theta)|\geq3-\sum_{l=0}^{n-1}\delta_l^3$$  with a unique sign for these $\theta$. Using Proposition \ref{apn} with $s=2$ finishes  the proof.
  \end{proof}
 Moreover, we have 
  \begin{prop}\label{shuang}
  If $|\frac{d}{d \theta} E_{n+1}^i(\theta)|\leq\delta_{n}^2$ for some $|\theta-\theta^*|<\delta_n/(20M_1)$,  then we have 
  	$$|\frac{d}{d \theta} E_{n+1}^i(\theta)|\geq\|\theta+c_{n+1}^i\cdot \omega-\mu_{n+1}\|,$$
  	where $\mu_{n+1}:=\mu_n\  (=0 \text{ or }1/2)$ is given by  the induction hypothesis \textbf{(H4)} or \textbf{(H10)}.
  \end{prop}
  \begin{proof}
 Assume  $|\frac{d}{d \theta} E_{n+1}^i(\theta)|\leq\delta_{n}^2$. Then \eqref{51} holds. So, we deduce from \textbf{(H4)} and \textbf{(H10)} that 
 $$\|\theta+c_{n}^i\cdot \omega-\mu_{n}\|\leq2\delta_n^2.$$
 Since $c_{n+1}^i=c_{n}^i$ and $\mu_{n+1}=\mu_{n}$, it follows that the symmetric point $\theta_{n}^i=\theta_{n+1}^i=-c_{n+1}^i\cdot \omega+\mu_{n+1}$ (mod $1$) belongs to the interval of $|\theta-\theta^*|\leq \delta_n/(20M_1)+2\delta_n^2<\delta_n/(10M_1)$. We can now apply Proposition \ref{ds} and Lemma \ref{C2} with $\theta_s=\theta_{n+1}^i,\delta=\delta_n^2$ to complete the proof.
  \end{proof}
  We then prove a preliminary upper bound concerning the {\bf Center Theorem}. 
\begin{lem}\label{qiang}
	For all $c_{n+1}^i,c_{n+1}^j\in Q_{n+1}$, we have 
	\begin{equation}\label{changa}
		m(c_{n+1}^i,c_{n+1}^j):=\min(\|(c_{n+1}^i-c_{n+1}^j)\omega\|,\|2\theta^*+(c_{n+1}^i+c_{n+1}^j)\omega\|)\leq\delta_n^3.
	\end{equation}  
Thus, $\theta^*+h$ belongs to the interval of  $|\theta-\theta^*|<\delta_n/(10M_1)$, where $h=(c_{n+1}^j-c_{n+1}^i)\cdot \omega$ or $-((c_{n+1}^i+c_{n+1}^j)\cdot \omega+2\theta^*)$ {\rm (mod $1$)} satisfying $|h|=m(c_{n+1}^i,c_{n+1}^j)$. 
\end{lem}
\begin{proof}
	Since $c_{n+1}^i=c_{n}^i$, $c_{n+1}^j=c_{n}^j$  and from \textbf{(H5)}, \textbf{(H12)},  it suffices to show that there exist $E_n^i\in \sigma(H_{B_n^i}(\theta^*))$ and $ E_n^j\in \sigma(H_{B_n^j}(\theta^*))$ such that $|E_n^i-E_n^j|\leq \delta_n^6/\sqrt{2}$.
	Note that  \eqref{dec} holds for all $c_{n+1}^i\in Q_{n+1}$ (in the proof of this property,   the assumption \eqref{SA} is not necessary).   So, restricting  the equation $H_{B_{n+1}^r}(\theta^*)\psi_{n+1}^r=E_{n+1}^r(\theta^*)\psi_{n+1}^r$ to $B_n^r\ (r=i,j)$ implies 
\begin{align*}
		 \|\left(H_{B_n^r}(\theta^*)-E_{n+1}^r(\theta^*)\right)\psi_{n+1}^r\|=\|\Gamma_{B_{n}^r}\psi_{n+1}^r\|\leq\delta_n^{10},
\end{align*}
which shows  $|E_n^r-E_{n+1}^r(\theta^*)|\leq2\delta_n^{10}$ for some $ E_n^r\in \sigma (H_{ B_n^r}(\theta^*))$ by Corollary \ref{trialcor} and $\|\psi_{n+1}^r\chi_{B_n^r}\|\approx1$. Since $c_{n+1}^i,c_{n+1}^j\in Q_{n+1}$, we get
\begin{align*}
	|E_n^i-E_n^j|&\leq |E_n^i-E_{n+1}^i(\theta^*)|+|E_n^j-E_{n+1}^j(\theta^*)|+|E_{n+1}^i(\theta^*)-E_{n+1}^j(\theta^*)|\\
	&\leq2\delta_n^{10}+2\delta_{n+1}\leq \delta_n^6/\sqrt{2}.
\end{align*}
\end{proof}

We are in a position to prove the {\bf Center Theorem} of stage $n+1$ in {\bf Subcase A} of {\bf Case 1}. 

  \begin{thm}[] Assume $c_n^i$ satisfies \eqref{SA}. Then for  any $c_{n+1}^j\in Q_{n+1}$, we have 
 	$$m(c_{n+1}^i,c_{n+1}^j)\leq\sqrt{2}|E_{n+1}^i(\theta^*)-E_{n+1}^j(\theta^*)|^{1/2}\leq2\delta_{n+1}^{1/2}.$$
  \end{thm}
\begin{proof}
By \eqref{changa}, $\theta^*+h$ belongs to the interval of  $|\theta-\theta^*|<\delta_n/(10M_1)$ on which $E_{n+1}^i$ is defined.  By {\bf (a)} of Proposition \ref{kn+1}, there is a unique eigenvalue of $H_{B_{n+1}^i}(\theta^*+h)$ with $|E-E^*|<\delta_{n}/5$. Since $H_{B_{n+1}^i}(\theta^*+h)=H_{B_{n+1}^j}(\theta^*)$ and $c_n^j\in Q_n$, we  must have  $E_{n+1}^j(\theta^*)=E_{n+1}^i(\theta^*+h)$.  If $|\frac{d}{d\theta}E_{n+1}^i|\geq\delta_n^3$ for all $|\theta-\theta^*|<|h|$, we get  
$$|E_{n+1}^i(\theta^*+h)-E_{n+1}^i(\theta^*)|\geq\delta_n^3 |h|\geq h^2.$$
Otherwise,  $|\frac{d}{d\theta}E_{n+1}^i|<\delta_n^3$ for some  $|\theta-\theta^*|<|h|$. By 
Proposition \ref{shuang}, we have $\delta_n^3>|\frac{d}{d \theta} E_{n+1}^i(\theta)|\geq\|\theta+c_{n+1}^i\cdot \omega-\mu_{n+1}\|$. Thus,  the symmetry point $\theta_{n+1}^i=-c_{n+1}^i\cdot \omega+\mu_{n+1}$ (mod $1$) belongs to the interval of $|\theta-\theta^*|<\delta_n^2$. Recalling Proposition \ref{ds}, $E_{n+1}^i$ satisfies the condition of Lemma \ref{C2} with $\theta_s=\theta_{n+1}^i,\theta_2=\theta^*+h,\theta_1=\theta^*,\delta=\delta_n^2,|h|\leq\delta$. Thus we have  
$$|E_{n+1}^i(\theta^*+h)-E_{n+1}^i(\theta^*)|\geq \frac{1}{2}\min(h^2,|2\theta^*+h-2\theta_{n+1}^i|^2)=\frac{1}{2}h^2.$$
\end{proof}

\fbox{\textbf{Subcase B}}.  The negation of \eqref{SA}, i.e.,  $c_{n+1}^i=c_n^i \in Q_{n+1}$ satisfies  \begin{equation}\label{SB}
  \operatorname{dist}(\sigma(H_{B_n^i}(\theta^*))- E_n^i(\theta^*),E^*)\leq\delta_{n}.
 \end{equation}
\begin{rem}
	In the one dimension case,  \textbf{Subcase B} is excluded by splitting lemma of \cite{FSW90}. However, this  lemma restricts to the  one  dimension case.  So,  we must  deal with this subcase in higher dimensions. 
\end{rem}

We will show  how to get back to  {\bf Class B} of the  induction hypothesis from {\bf Subcase B}.

First, we notice that \eqref{SB} can not  be in the case in \textbf{(H1)} of {\bf Class A}. Thus, such $c_n^i$ belongs to {\bf Class B} and \textbf{(H6)}--\textbf{(H12)} hold true. Second,  as we have seen, \textbf{Case 1} along with \textbf{Subcase A} at stage $n$ implies {\bf Class A}, and hence \textbf{Subcase A} at stage $n+1$. Thus, if \eqref{SB} holds, then there must be some largest $m\leq n-1$ such that $s_m\leq10l_m^2$. So,  we have $c_m^i\in Q_m$ and its mirror image   $\tilde{c}_m^i$ together with   two blocks $B_m^i,\tilde{B}_m^i$ such that 
$$B_m^i,\tilde{B}_m^i\subset B_{m+1}^i\subset \cdots \subset B_n^i\subset B_{n+1}^i.$$
Note that \begin{equation}\label{mton}
	c_{n+1}^i=c_n^i=\cdots =c_{m+1}^i=(c_m^i+\tilde{c}_m^i)/2.
\end{equation}
Since  \eqref{SB}, there is another eigenvalue $\mathcal{E}_n^i(\theta^*)$ of $H_{B_n^i}(\theta^*)$ in the interval of $|E-E^*|\leq\delta_{n}$. Hence by  \textbf{(H11)}, we have $$\delta_{n-1}^2|\theta_n^i-\theta^*|\leq|E_n^i(\theta^*)-\mathcal{E}_n^i(\theta^*)|\leq2\delta_{n},$$
where $\theta_n^i=-c_n^i\cdot \omega+\mu_n$ (mod 1). Thus, the symmetric point $\theta_{n+1}^i:=\theta_n^i$ satisfies  \begin{equation}\label{symn}
	|\theta_{n+1}^i-\theta^*|\leq2\delta_n/\delta_{n-1}^2<\delta_n^{1/2}.
\end{equation} 
Recalling \eqref{mton}, we obtain
$$\min(\|c_{m+1}^i\cdot \omega+\theta^*\|,\|c_{m+1}^i\cdot \omega+\theta^*-\frac{1}{2}\|)< \delta_n^{1/2}$$ 
and hence, 
\begin{equation}\label{dif}
	\|(c_m^i+\tilde{c}_m^i)\cdot \omega+2\theta^*\|<2\delta_n^{1/2}.
\end{equation} 
Based on the Diophantine condition, we have
\begin{align}\label{1209}
\begin{split}
		\|2c_m^i\cdot \omega +2\theta^*\|&\geq\|2(c_m^i-\tilde{c}_m^i)\cdot \omega\|-\|(c_m^i+\tilde{c}_m^i)\cdot \omega+2\theta^*\|\\ &\geq\frac{\gamma}{(20l_m^2)^\tau}-2\delta_{n}^{1/2}\\
	&>\delta_{m-1}^{1/3}.
\end{split}
\end{align}
The above inequality excludes the possibility of \textbf{(H6)} in {\bf Class B} at stage $m$. Thus,  we deduce that $c_m^i$ belongs to {\bf Class A} and \textbf{(H1)}--\textbf{(H5)} hold for $c_m^i$. We  let  $E_m^i(\theta)$ be the  unique eigenvalue of $H_{B_m^i}(\theta)$ in the interval of  $|E_{m}^i(\theta)-E^*|<\delta_{m-1}/9$ and  let $\psi_m$ be  its eigenfunction. From \eqref{dif}, we obtain for $|\theta-\theta^*|=O(\delta_n^{1/2})$,  \begin{equation}\label{chabie}
	H_{\tilde{B}_m^i}(\theta)=H_{B_m^i}(-\theta-(c_m^i+\tilde{c}_m^i)\cdot \omega)=H_{B_m^i}(\theta+O(\delta_n^{1/2})).
\end{equation}
Since $\delta_{n}^{1/2}\ll\delta_{m-1}$, by \textbf{(H1)} and \textbf{(H2)}, there is also a unique eigenvalue $\tilde{E}_m^i$ of $H_{\tilde{B}_m^i}(\theta)$ satisfying $|\tilde{E}_{m}^i(\theta)-E^*|<\delta_{m-1}/9$ so that its  eigenfunction $\tilde{\psi}_m^i$  decays exponentially fast away from $\tilde{c}_m^i$.

\begin{prop}\label{325} Assume \eqref{SB} holds true.  Then for $|\theta-\theta^*|<10\delta_{n}^{1/2}$, 
	\begin{itemize}
		\item [\textbf{(a)}] $H_{B_{n+1}^i}(\theta)$ has  exactly  two eigenvalues $E_{n+1}^i(\theta)$ and $\mathcal{E}_{n+1}^i(\theta)$ in the interval of $|E-E^*|<50M_1\delta_{n}^{1/2}$. Moreover, any other $\hat{E}\in\sigma(H_{B_{n+1}^i}(\theta)) $ must obey $|\hat{E}-E^*|\geq\delta_{n-1}/6$.
		\item [\textbf{(b)}] The corresponding eigenfunction of $E_{n+1}^i$ {\rm (resp. $\mathcal{E}_{n+1}^i$)}, $\psi_{n+1}$ {\rm(resp. $\Psi_{n+1}$)} decays exponentially fast away from $c_m^i$ and $\tilde{c}_m^i$, 
		\begin{align}\label{1421}
			\begin{split}
	|\psi_{n+1}(x)|\leq e^{-(\gamma_0/4)\|x-c_m^i\|_1}+e^{-(\gamma_0/4)\|x-\tilde{c}_m^i\|_1},\\
	|\Psi_{n+1}(x)|\leq e^{-(\gamma_0/4)\|x-c_m^i\|_1}+e^{-(\gamma_0/4)\|x-\tilde{c}_m^i\|_1},		
			\end{split}
		\end{align}
for $\operatorname{dist}(x,\{c_m^i,\tilde{c}_m^i\})\geq l_m^{6/7}.$	
	\item [\textbf{(c)}] The two eigenfunctions can be expressed as 
		\begin{align}\label{1224}
			\begin{split}
				\psi_{n+1}=A\psi_m+B\tilde{\psi}_m+O(\delta_m^{10}),\\
				\Psi_{n+1}=B\psi_m-A\tilde{\psi}_m+O(\delta_m^{10}),
			\end{split}	
		\end{align}
		where $A^2+B^2=1$.
		\item [\textbf{(d)}] $\|G_{B_{n+1}^i}^{\perp\perp}(E_{n+1}^i)\|\leq10\delta_{n-1}^{-1}$, where $G_{B_{n+1}^i}^{\perp\perp}$ denotes the Green's function for $B_{n+1}^i$ on the orthogonal complement of the space spanned by $\psi_{n+1}$ and $\Psi_{n+1}$. 
	\end{itemize}
\end{prop}
\begin{proof}
By the exponential decay of $\psi_{n}$ and $\Psi_n$, we have
$$\|(H_{B_{n+1}^i}(\theta^*)-E^*)\psi_n\|\leq|E_n^i-E^*|+\|\Gamma_{B_n^i}\psi_{n}\|\leq2\delta_n,$$ 
$$\|(H_{B_{n+1}^i}(\theta^*)-E^*)\Psi_n\|\leq|\mathcal{E}_n^i-E^*|+\|\Gamma_{B_n^i}\Psi_{n}\|\leq2\delta_n.$$
 The two orthogonal trial wave functions  give two eigenvalues of  $H_{B_{n+1}^i}(\theta^*)$ in $|E-E^*|<2\sqrt2\delta_n$ by Corollary \ref{trialcor}. Using $|V'|\leq M_1$,  we deduce that $H_{B_{n+1}^i}(\theta)$ has  at least two eigenvalues  in $|E-E^*|<50M_1\delta_{n}^{1/2}$, which proves the existence part of \textbf{(a)}. The proof of \textbf{(b)} is an application of Green' function estimates  by restricting the equation $H_{B_{n+1}}(\theta)\psi_{n+1}=E_{n+1}(\theta)\psi_{n+1}$ to some good annuals $A$. Thus, the value of $\psi_{n+1}$ inside $A$ can be given by the Green's function $G_A(E_{n+1})$ and the values of $\psi_{n+1}$ on $\partial A$:
$$ \psi_{n+1}(x)=\sum_{z,z'}G_A(x,z)\Gamma_A\psi_{n+1}(z').$$
We use the fact that $B_{m+2}^i\setminus (B_{m-1}^i\cup \tilde{B}_{m-1}^i)$ is $(m-1)$-good to estimate the value at $x$  satisfying $\operatorname{dist}(x,\{c_m^i,\tilde{c}_m^i\})\geq l_m^{6/7}$ and 	 $\|x-c_{n+1}^i\|_1\leq l_{m+2}/2$, the fact that $B_{r+2}^i\setminus B_{r}^i$ is $(r+1)$-good to estimate the value at $x$ satisfying $l_{r+1}^{6/7}\leq\|x-c_{n+1}^i\|_1\leq l_{r+2}/2$ for $m+1\leq r\leq  n-2$,  and the fact that  $B_{n+1}^i\setminus B_{n-1}^i$ is $(n-1)$-good to estimate the value at  $x$ satisfying $l_{n}^{6/7} \leq\|x-c_{n+1}^i\|_1$. We should emphasize the first fact is because a third $(m-1)$-singular block inside  $B_{m+2}^i$ will be  excluded by the {\bf Center Theorem} of  stage $m-1$ and the Diophantine condition,  and the last fact is because \eqref{dif} implies for $  c_{n-1}^i\neq x\in B_{n+1}^i$,
\begin{align*}\label{aha}
	\|(c_{n-1}^i+x)\cdot \omega+2\theta^*\|&\geq\|(c_{n-1}^i-x)\cdot \omega\|-\|2c_{n-1}^i\cdot \omega+2\theta^*\|\\
	&=\|(c_{n-1}^i-x)\cdot \omega\|-\|(c_m^i+\tilde{c}_m^i)\cdot \omega+2\theta^*\|\\
	&\geq \frac{\gamma}{(2l_{n+1})^\tau}-\delta_n^{1/2}\\
	&>\delta_{n-1}^{1/3},
\end{align*}
which excludes a second $(n-1)$-singular block  inside  $B_{n+1}^i$ by {\bf Center Theorem} of stage $n-1$ and the Diophantine condition. Notice that all the  annuals are good sets of   stage no more than $n-1$. Thus, the  Green's function estimates  hold  for 
\begin{equation*}
	|\theta-\theta^*|<10\delta_{n}^\frac{1}{2}<\delta_{n-1}/(10M_1),\  |E_{n+1}^i-E^*|<50M_1\delta_{n}^{1/2}<\delta_{n-1}/5.
\end{equation*}
Thus, we finish the proof of \textbf{(b)}. Now we establish \textbf{(c)}. It suffices to show $\psi_{n+1}$ and $\Psi_{n+1}$ are close to  a linear combination of $\psi_m$ and $\tilde{\psi}_m$ inside $B_m^i\cup \tilde{B}_m^i$.
We restrict the equation $H_{B_{n+1}^i}(\theta)\psi_{n+1}=E_{n+1}^i(\theta)\psi_{n+1}$ to $B_m^i$ to get 
$$\left(H_{B_{m}^i}-E_{n+1}^i\right)\psi_{n+1}=\Gamma_{B_{m}^i}\psi_{n+1}.$$
	Combining \eqref{1421} and the above equation, we get  $$\|P_m^\perp\psi_{n+1}\|=\|G_{B_m^i}^\perp(E_{n+1})P_m^\perp\Gamma_{B_{m}^i}\psi_{n+1}\|= O(\delta_{m-1}^{-1}e^{-\frac{1}{4}\gamma_0l_m})\leq\frac{1}{2}\delta_m^{10},$$
where $P_m^\perp$ is the projection onto the orthogonal complement of $\psi_m$ and $G_{B_m^i}^\perp(E_{n+1}^i)$ is the Green's function of  $B_m^i$ on $\operatorname{Range}P_m^\perp$ with the upper bound
 \begin{align}\label{839}
	\|G_{B_m^i}^\perp(E_{n+1}^i)\|&\leq\operatorname{dist}(\sigma(H_{B_m^i}(\theta))- E_m^i(\theta),E_{n+1}^i)^{-1}\\
	\nonumber &\leq(\frac{\delta_{m-1}}{5}-\frac{\delta_{n-1}}{6})^{-1}\leq\frac{30}{\delta_{m-1}}
\end{align}by \textbf{(H1)} of stage $m$. Therefore, inside $B_m^i$, we have 
$$P_m^\perp\psi_{n+1}=O(\delta_m^{10})$$and hence, 
$$\psi_{n+1}\chi_{B_m^i}=a\psi_m+O(\delta_m^{10}),$$
where $a=\langle \psi_{n+1},\psi_m\rangle$.
By the approximation \eqref{chabie}, we get a similar estimate in $\tilde{B}_m^i$
$$\psi_{n+1}\chi_{\tilde{B}_m^i}=b\tilde{\psi}_m+O(\delta_m^{10})$$
with $b=\langle \psi_{n+1},\tilde{\psi}_m\rangle$. By \eqref{1421}, we have $\|\psi_{n+1}\chi_{\tilde{B}_{n+1}^i\setminus (B_m^i\cup \tilde{B}_m^i)}\|\leq \delta_m^{10}$, and thus
$$\psi_{n+1}=a\psi_m+b\tilde{\psi}_m+O(\delta_m^{10}).$$
Taking the norm gives $k:=a^2+b^2=1-O(\delta_m^{10})$. We set $A=a/k$ and $B=b/k$. Hence, $A^2+B^2=1$ and $|A-a|,|B-b|= O(\delta_m^{10})$, which gives the desired expression of $\psi_{n+1}$.  Similar arguments give $\Psi_{n+1}=C\psi_m+D\tilde{\psi}_m+O(\delta_m^{10})$ with $C^2+D^2=1$. For convenience, we write $A=\cos\alpha, B=\sin\alpha, C=\sin\beta, D=-\cos\beta$. Using $\langle\psi_{n+1} ,\Psi_{n+1}\rangle=0$, we get  $|\sin(\beta-\alpha)|=O(\delta_m^{10})$. We can choose $\beta$ satisfying $|\beta-\alpha|\leq O(\delta_m^{10})$. Thus, $|B-C|=|\sin\alpha-\sin\beta|=O(\delta_m^{10})$ and $|A+D|=|\cos\alpha-\cos\beta|=O(\delta_m^{10})$, giving the desired expression $\Psi_{n+1}=B\psi_m-A\tilde{\psi}_m+O(\delta_m^{10})$. Now assume that $\hat{E}\in\sigma(H_{B_{n+1}^i}(\theta))$ is a third eigenvalue in the interval of $|\hat{E}-E^*|<\delta_{n-1}/6$. The Green's function estimates  and \eqref{839} still hold if we replace $E_{n+1}$ by $\hat{E}$. Thus, by a similar argument, the eigenfunction of $\hat{E}$ can be expressed as 
$$\hat{\psi}=\hat{A}\psi_m+\hat{B}\tilde{\psi}_m+O(\delta_m^{10})$$ 
with $\hat{A}^2+\hat{B}^2=1.$ By orthogonality, we have $A\hat{A}+B\hat{B}=O(\delta_m^{10})$ and $B\hat{A}-A\hat{B}=O(\delta_m^{10})$. This is impossible since  $(A\hat{A}+B\hat{B})^2+(B\hat{A}-A\hat{B})^2=1$.
 Hence a third eigenvalue $\hat{E} $ must obey $|\hat{E}-E^*|\geq\delta_{n-1}/6$. Finally,  \textbf{(d)} follows from \textbf{(a)} immediately.
\end{proof}

We need the upper bound on derivatives of eigenvalues parameterizations of stage $m.$

\begin{lem}\label{cham}
	For $|\theta-\theta^*|<10\delta_n^{1/2}$, we have 
	\begin{equation}\label{m+}
		|\frac{d}{d\theta}\left(E_m^i+\tilde{E}_m^i\right)(\theta)|\leq \delta_n^{1/3}.
	\end{equation}
\end{lem}
\begin{proof}
	From \eqref{chabie}, we obtain  $\tilde{E}_m^i(\theta)=E_m^i(-\theta+2\theta_{n}^i)$.
	Thus,
	\begin{align*}
		|\frac{d}{d\theta}\left(E_m^i+\tilde{E}_m^i\right)(\theta)|&=|\frac{d}{d\theta}E_m^i(\theta)-\frac{d}{d\theta}E_m^i(-\theta+2\theta_{n}^i)|\\
		&= |\frac{d^2}{d\theta^2}E_m^i(\xi)| \cdot |2\theta-2\theta_{n}^i|\\
		&\leq O(\delta_{m-1}^{-1}\delta_n^{1/2})
		\\
		&\leq \delta_n^{1/3},
	\end{align*}
where on the third line we used the estimate 
	\begin{align*}
		|\frac{d^2}{d \theta^2} E_{m}^i(\theta)|&=|\left\langle\psi_{m}, V'' \psi_{m}\right\rangle-2\left\langle\psi_{m}, V' G_{B_m^i}^{\perp}(E_{m}^i) V' \psi_{m}\right\rangle|\\
		&=O(\|G_{B_m^i}^{\perp}(E_{m}^i)\|)\\
		&=O(\delta_{m-1}^{-1})
	\end{align*}
for $|\theta-\theta^*|<10\delta_n^{1/2}<\delta_{m-1}/(10M_1)$ by \textbf{(H1)} of  stage $m$.
\end{proof}

We also have the lower bound on the derivatives. 

\begin{lem}\label{xiajiem}
	For $|\theta-\theta^*|<10\delta_n^{1/2}$, we have $|\frac{d}{d\theta}E_m^i(\theta)|\geq\delta_{m-1}^{2}$.
\end{lem}
\begin{proof}
Assume  that it is not true. Then by \textbf{(H4)} and recalling \eqref{1209}, we have  
\begin{align*}
	|\frac{d}{d\theta}E_m^i(\theta)|&\geq \min(\|\theta+c_m^i\cdot \omega\|,\|\theta+c_m^i\cdot \omega-\frac{1}{2}\|)\\
	&\geq \frac{1}{2}\|2\theta+2c_m^i\cdot \omega\|\\
	&>\delta_{m-1},
\end{align*}
which leads to a contradiction. 
\end{proof} 

We can also establish estimates of derivatives of stage $n+1.$

\begin{prop}\label{gujim}
	Let $|\theta-\theta^*|<10\delta_n^{1/2}$. Then 
	\begin{itemize}
		\item [\textbf{(a)}] $E_{n+1}^i$ and $\mathcal{E}_{n+1}^i$ are $C^1$ functions and if $ E_{n+1}^i(\theta)\neq\mathcal{E}_{n+1}^i(\theta)$, then
		\begin{align}
			\frac{d}{d \theta} E_{n+1}^i & =(A^2-B^2) \frac{d}{d\theta} E_m^i+O(\delta_{m}^{1/3} ), \label{dem}\\
			\frac{d}{d \theta} \mathcal{E}_{n+1}^i & =(B^2-A^2) \frac{d}{d\theta} E_m^i+O(\delta_{m}^{1/3} ).\nonumber
		\end{align}
		\item [\textbf{(b)}] If $ E_{n+1}^i(\theta)\neq\mathcal{E}_{n+1}^i(\theta)$, then both $\frac{d^2}{d \theta^2} E_{n+1}^i(\theta)$ and 	$\frac{d^2}{d \theta^2} \mathcal{E}_{n+1}^i(\theta) $ exist. Moreover,
		\begin{align}
			\frac{d^2}{d \theta^2} E_{n+1}^i & =\frac{2\left\langle\psi_{n+1}^i, V' \Psi_{n+1}^i\right\rangle^2}{E_{n+1}^i-\mathcal{E}_{n+1}^i}+O(\delta_{n-1}^{-1} ),\label{dfm} \\
			\frac{d^2}{d \theta^2} \mathcal{E}_{n+1}^i & =\frac{2\left\langle\psi_{n+1}^i, V' \Psi_{n+1}^i\right\rangle^2}{\mathcal{E}_{n+1}^i-E_{n+1}^i}+O(\delta_{n-1}^{-1} ). \label{ddfm}
		\end{align}
		\item [\textbf{(c)}] At the point $E_{n+1}^i(\theta)\neq \mathcal{E}_{n+1}^i(\theta)$, if $|\frac{d}{d\theta}E_{n+1}^i(\theta)|\leq10\delta_n^{1/2}$, then $|\frac{d^2}{d\theta^2}E_{n+1}^i(\theta)|>\delta_n^{-1/3}>2$. Moreover, the  sign of $\frac{d^2}{d\theta^2}E_{n+1}^i(\theta)$ is the same as that of $ E_{n+1}^i(\theta)- \mathcal{E}_{n+1}^i(\theta)$. The analogous conclusion holds by exchanging  $E_{n+1}^i(\theta)$ and $\mathcal{E}_{n+1}^i(\theta)$.
	\end{itemize}
\end{prop}

\begin{proof}The proof is similar to that of Proposition \ref{guji}.  The $C^1$ smoothness of the eigenvalues parameterizations is  a remarkable result of perturbations theory for self-adjoint operator \cite{Rel69, Kat95}. When $E_{n+1}^i$ is simple, by \eqref{1224} and  Theorem \ref{daoshu}, we have 
	\begin{align*}
		\frac{d}{d \theta} E_{n+1}^i & =\left\langle\psi_{n+1}^i, V' \psi_{n+1}^i\right\rangle\\
		&=A^2 \frac{d}{d \theta} E_m^i+B^2 \frac{d}{d \theta} \tilde{E}_m^i+O(\delta_m^{10} ) \\
		& =(A^2-B^2) \frac{d}{d \theta} E_m^i+B^2(\frac{d}{d \theta} E_m^i+\frac{d}{d \theta} \tilde{E}_m^i)+O(\delta_m^{10} )\\
		& =(A^2-B^2) \frac{d}{d\theta} E_m^i+O(\delta_m^{1/3} ),
	\end{align*}
	where we have used \eqref{m+} in the last equality. This  completes  the proof of  \textbf{(a)}.
	
	To prove  \textbf{(b)},  we use the formula 
	\begin{align*}
	\frac{d^2}{d \theta^2} E_{n+1}^i&=\left\langle\psi_{n+1}^i, V'' \psi_{n+1}^i\right\rangle+2 \frac{\left\langle\psi_{n+1}^i, V' \Psi_{n+1}^i\right\rangle^2}{E_{n+1}^i-\mathcal{E}_{n+1}^i}\\
	&\ \  -2\left\langle V' \psi_{n+1}^i,G^{\perp \perp}_{B_{n+1}^i}(E_{n+1}^i) V' \psi_{n+1}^i\right\rangle 
	\end{align*}
	from Theorem \ref{daoshu}.
	The remainder term is bounded by $2\|G^{\perp \perp}_{B_{n+1}^i}(E_{n+1}^i)\|\cdot\|V' \psi_{n+1}^i\|^2$, where we can use the estimate $\|G^{\perp \perp}_{B_{n+1}^i}(E_{n+1}^i)\| \leq10\delta_{n-1}^{-1}$ {\bf (d)} in Proposition \ref{325}. 
	
	Now we turn to the proof of \textbf{(c)}. If $|\frac{d}{d\theta}E_{n+1}^i(\theta)|\leq10\delta_n^{1/2}$, then by \eqref{dem},  we have 
	$$|A^2-B^2|\cdot |\frac{d}{d\theta}E_m^i(\theta)|\leq10\delta_n^{1/2}+O(\delta_m^{1/3}) \leq\delta_m^{1/2},$$
	which implies $A^2\approx B^2\approx \frac{1}{2}$ by Lemma \ref{xiajiem}. Thus,
	\begin{align}\label{2108}
\begin{split}
			|\langle\psi_{n+1}^i, V' \Psi_{n+1}^i\rangle|&=|AB \frac{d}{d \theta} E_m^i-AB \frac{d}{d \theta} \tilde{E}_m^i+O(\delta_m^{1 / 3} )| \\
		& \geq|2AB \frac{d}{d\theta} E_m^i|-O(\delta_m^{1 / 3} )\\
		&\geq\frac{1}{2}\delta_{m-1}^{2}.
\end{split}
	\end{align}
	By Proposition \ref{325} \textbf{(a)}, we have $|E_{n+1}^i-\mathcal{E}_{n+1}^i|\leq100M_1\delta_{n}^{1/2}$.  By using  \eqref{dfm}, we obtain  $|\frac{d}{d\theta}E_{n+1}^i(\theta)|\geq\frac{1}{4}\delta_{m-1}^{4}(100M_1\delta_n^{1/2})^{-1}-O(\delta_{n-1}^{-1})>\delta_n^{-1/3}$, whose sign is determined by that of $ E_{n+1}^i(\theta)- \mathcal{E}_{n+1}^i(\theta)$.
\end{proof}
 Since $H_{B_{n+1}^i}(\theta^*+h)=H_{B_{n+1}^j}(\theta^*)$, we deduce from  \textbf{(a)}  in Proposition \ref{325} and  Lemma \ref{qiang} that  $H_{B_{n+1}^j}(\theta^*)$ also has exactly  two eigenvalues $E_n^j, \mathcal{E}_n^j$ in the interval of $|E-E^*|\leq 50M_1\delta_n^{1/2}$ satisfying $\{E_{n}^j,\mathcal{E}_{n}^j\}=\{E_{n}^i(\theta^*+h),\mathcal{E}_{n}^i(\theta^*+h)\} .$\smallskip
 
We are ready to prove the {\bf Center Theorem} of stage $n+1$ in {\bf Subcase B} of {\bf Case 1}. 

\begin{thm}[]\label{C1m}  Assume $c_n^i$ satisfies \eqref{SB}.  Then for  any $c_{n+1}^j\in Q_{n+1}$, we have
	\begin{align}\label{changm}
		m(c_{n+1}^i,c_{n+1}^j)&\leq\sqrt{2}\min(|E_{n+1}^i(\theta^*)-E_{n+1}^j(\theta^*)|^{1/2},|\mathcal{E}_{n+1}^i(\theta^*)-\mathcal{E}_{n+1}^j(\theta^*)|^{1/2},\nonumber\\
		&\ \ |E_{n+1}^i(\theta^*)-\mathcal{E}_{n+1}^j(\theta^*)|^{1/2},|\mathcal{E}_{n+1}^i(\theta^*)-E_{n+1}^j(\theta^*)|^{1/2})\\
		\nonumber&\leq2\delta_{n+1}^\frac{1}{2}.
	\end{align}
\end{thm}
\begin{proof}
	The proof is similar to that of Theorem \ref{C1}. The preliminary bound  \eqref{changa} implies that $\theta^*\pm m(c_{n+1}^i,c_{n+1}^j)$ belong to  the interval of $|\theta-\theta^*|<10\delta_n^{1/2}$ on which $E_{n+1}^i$ and $\mathcal{E}_{n+1}^i$ are defined.  We also  recall \eqref{symn} that the symmetric point $\theta_{n+1}^i$ belongs to the interval of $|\theta-\theta^*|<10\delta_n^{1/2}$. So there will be two cases.\smallskip\\
	\textbf{Case I}.
	 $E_{n+1}^i(\theta_{n+1}^i) \neq \mathcal{E}_{n+1}^i(\theta_{n+1}^i)$. 
	Without loss of generality,  we may assume $E_{n+1}^i(\theta_{n+1}^i)>\mathcal{E}_{n+1}^i(\theta_{n+1}^i)$. Notice that the union of two eigenvalue curves is symmetric about $\theta_{n+1}^i$. Thus  we must have
	$$
	\frac{d}{d \theta} E_{n+1}^i(\theta_{n+1}^i)=\frac{d}{d \theta} \mathcal{E}_{n+1}^i(\theta_{n+1}^i)=0 .
	$$ 
	By \textbf{(b)} and \textbf{(c)} of Proposition \ref{gujim},  we see that $\theta_{n+1}^i$ is a local minimum point of $ E_{n+1}^i$ and a local maximum point of $ \mathcal{E}_{n+1}^i$. Moreover, $\frac{d}{d \theta} E_{n+1}^i$ is increasing and $\frac{d}{d \theta} \mathcal{E}_{n+1}^i$ is decreasing whenever $|\frac{d}{d \theta} E_{n+1}^i|\leq10\delta_{n}^{1/2}$. Thus, $E_{n+1}^i>\mathcal{E}_{n+1}^i$ continues to hold for all $|\theta-\theta^*|<10\delta_n^{1/2}$, which implies $\frac{d^2}{d\theta^2}E_{n+1}^i(\theta)>2$ whenever $|\frac{d}{d\theta}E_{n+1}^i(\theta)|<10\delta_n^{1/2}$. Moreover,  $\frac{d}{d\theta}E_{n+1}^i$(resp. $\frac{d}{d\theta}\mathcal{E}_{n+1}^i$) cannot reenter the band $|\frac{d}{d\theta}E|<10\delta_n^{1/2}$ since it is increasing (resp. decreasing) there.   It follows that     $E_{n+1}^i(\theta),\mathcal{E}_{n+1}^i(\theta)$ satisfy the condition of  Lemma \ref{C2} with $\theta_2=\theta^*+h,\theta_1=\theta^*,\delta=10\delta_n^{1/2},|h|\leq\delta.$ Thus, we get
	$$|E_{n+1}^i(\theta^*+h)-E_{n+1}^i(\theta^*)|\geq \frac{1}{2}\min(h^2,|2\theta^*+h-2\theta_{n+1}^i|^2)=\frac{1}{2}h^2$$
	and the same estimate holds  true for $\mathcal{E}_{n+1}^i$,
	where $h=(c_{n+1}^j-c_{n+1}^i)\cdot \omega$ or $-((c_{n+1}^i+c_{n+1}^j)\cdot \omega+2\theta^*)$ {\rm (mod $1$)} satisfying $|h|=m(c_{n+1}^i,c_{n+1}^j)$.
	An easy inspection gives us 
	\begin{align*}
	&\ \ \ |\mathcal{E}_{n+1}^i(\theta^*+h)-E_{n+1}^i(\theta^*)|\\
	&\geq\min(|E_{n+1}^i(\theta^*+h)-E_{n+1}^i(\theta^*)|,|\mathcal{E}_{n+1}^i(\theta^*+h)-\mathcal{E}_{n+1}^i(\theta^*)|)\\
	&\geq\frac{1}{2}h^2,\\
	&\ \ \  |E_{n+1}^i(\theta^*+h)-\mathcal{E}_{n+1}^i(\theta^*)|\\
	&\geq\min(|E_{n+1}^i(\theta^*+h)-E_{n+1}^i(\theta^*)|,|\mathcal{E}_{n+1}^i(\theta^*+h)-\mathcal{E}_{n+1}^i(\theta^*)|)\\
	&\geq\frac{1}{2}h^2.
	\end{align*}
	Now \eqref{changm}  follows from  $\{E_{n+1}^j(\theta^*),\mathcal{E}_{n+1}^j(\theta^*)\}=\{E_{n+1}^i(\theta^*+h),\mathcal{E}_{n+1}^i(\theta^*+h)\}$ since $H_{B_{n+1}^i}(\theta^*+h)=H_{B_{n+1}^j}(\theta^*)$, and one of the eigenvalue differences must be bounded above by $2\delta_{n+1}$ by the definition of $Q_{n+1}$.\smallskip\\
		\textbf{Case II}. $E_{n+1}^i(\theta_{n+1}^i)= \mathcal{E}_{n+1}^i(\theta_{n+1}^i)$.
	In this case,  we claim that $|\frac{d}{d \theta} E_{n+1}^i|\geq10\delta_n^{1/2}$ and $|\frac{d}{d \theta} \mathcal{E}_{n+1}^i|\geq10\delta_n^{1/2}$ hold for $|\theta-\theta^*|<10\delta_n^{1/2}.$ Moreover, they have opposite signs. First, we show it is true for $\theta=\theta_{n+1}^i$. An analog of  Lemma \ref{daogroup} gives us 
	\begin{align*}
	&\ \ \ \big\{\frac{d}{d \theta} E_{n+1}^i(\theta_{n+1}^i),\frac{d}{d \theta} \mathcal{E}_{n+1}^i(\theta_{n+1}^i)\big\}\\
	&=\{\text{Eigenvalues of the }2\times2\text{ matrix } PH'_{B_{n+1}^i}(\theta_{n+1}^i)P\},
	\end{align*}
	 where $P$ is the projection onto  the two dimensional eigenspace of $E_{n+1}^i(\theta_{n+1}^i)$.
	To calculate these eigenvalues,  we represent $PV'P:=PH'P$ in a special basis. Notice that $H_{B_{n+1}^i}(\theta_{n+1}^i)$ commutes with the reflect operator $(R\psi)(x):=\psi(2c_{n+1}^i-x)$. It follows that $\operatorname*{Range}P$ is a two dimensional  invariant subspace of $R$, which can be spanned by two eigenfunctions of $R$ since $R$ is  diagonalizable. All the eigenfunctions of $R$ are symmetric functions $\{\psi_s\}$ and antisymmetric functions $\{\psi_a\}$. We note that   $\operatorname*{Range}P$ cannot be spanned by only symmetric functions (resp. antisymmetric functions). Otherwise,  $\psi_{n+1}$ and $\Psi_{n+1}$ are  symmetric (resp. antisymmetric), contradicting  the expression \eqref{1224}. This allows us to express $PV'P$ in the basis $\{\psi_s,\psi_a\}$, which consists of one symmetric function and  one antisymmetric function
	$$P V' P=\left(\begin{array}{ccc}
		\left\langle\psi_s, V' \psi_s\right\rangle & \left\langle\psi_s, V' \psi_a\right\rangle \\
		\left\langle\psi_s, V' \psi_a\right\rangle & \left\langle\psi_a, V' \psi_a\right\rangle
	\end{array}\right)\ (\text {at } \theta=\theta_{n+1}^i) .
	$$
 Since $v$ is even and $1$-periodic , we deduce $(V'(\theta_{n+1}^i))(2c_1^i-x)=v'(\theta_{n+1}^i+(2c_1^i-x)\cdot \omega)=-v'(\theta_{n+1}^i+x\cdot \omega)=-(V'(\theta_{n+1}^i))(x)$, yielding $V'(\theta_{n+1}^i)$ is antisymmetric. Now by the symmetry and anti-symmetry properties of $\psi_s, \psi_a$, and $V'(\theta_{n+1}^i)$, we have $\langle\psi_s, V' \psi_s\rangle=$ $\langle\psi_a, V' \psi_a\rangle=0$, which gives us
	$$
	P V' P=\left(\begin{array}{cc}
		0 & \left\langle\psi_s, V' \psi_a\right\rangle \\
		\left\langle\psi_s, V' \psi_a\right\rangle & 0
	\end{array}\right)
	$$
	and therefore
	$$
	\frac{d}{d\theta} E_{n+1}^i(\theta_{n+1}^i)=-\frac{d}{d\theta} \mathcal{E}_{n+1}^i(\theta_{n+1}^i)=\langle\psi_s, V' \psi_a\rangle.
	$$
	We choose $E_{n+1}^i$ to satisfy $\frac{d}{d\theta} E_{n+1}^i(\theta_{n+1}^i)\geq0$ and  will show that it is not too small and then extend this for $|\theta-\theta^*| \leq 10 \delta_n^{1 / 2}$. Using the symmetry properties and the decay of the eigenfunctions, we have  $\psi_{m}= \pm R\tilde{\psi}_m+O(\delta_{m}^{10})$, $\psi_s=1/\sqrt{2}(\psi_m+R\psi_m) +O(\delta_{m}^{10})$ and $\psi_a=1/\sqrt{2}(\psi_m-R\psi_m) +O(\delta_{m}^{10})$, and thus 
	$$\frac{d}{d\theta} E_{n+1}^i(\theta_{n+1}^i)=\langle\psi_m,V'\psi_{m}\rangle+O(\delta_m^{10})=\frac{d}{d\theta} E_{m}^i(\theta_{n+1}^i)+O(\delta_m^{10}).$$
 By Lemma \ref{xiajiem}, we get
	$$
	\frac{d}{d\theta}E_{n+1}^i(\theta_{n+1}^i) \geq \frac{1}{2}\delta_{m-1}^2\geq10\delta_n^{1/2}.
	$$
	We now  show that this continues to hold for all $\theta$ in the interval of $|\theta-\theta^*| \leq 10 \delta_n^{1 / 2}$. Since $E_{n+1}^i$ is increasing and $\mathcal{E}_{n+1}^i$ is decreasing, we deduce  $E_{n+1}^i>\mathcal{E}_{n+1}^i$ for $\theta>\theta_{n+1}^i$.  If $\frac{d}{d\theta}E_{n+1}^i(\theta)\leq 10\delta_n^{1/2} $ for some smallest  $\theta>\theta_{n+1}^i$, by  \textbf{(c)} of Proposition \ref{gujim}, we have  $\frac{d^2}{d\theta^2}E_{n+1}^i(\theta)>0$. This is impossible. The same argument shows there is no $\theta<\theta_{n+1}^i$ such that $\frac{d}{d\theta}E_{n+1}^i(\theta)\leq 10\delta_n^{1/2} $, which proves our claim. In this case, we have  $E_{n+1}^i(\theta)=\mathcal{E}_{n+1}^i(2\theta_{n+1}^i-\theta)$ by the symmetry property of the eigenvalue curve. Thus, by the preliminary bound  \eqref{changa}, we obtain 
	\begin{align*}
	|E_{n+1}^i(\theta^*+h)-E_{n+1}^i(\theta^*)|&\geq 10\delta_n^{1/2}|h|\geq h^2,\\
	 |\mathcal{E}_{n+1}^i(\theta^*+h)-\mathcal{E}_{n+1}^i(\theta^*)|&\geq 10\delta_n^{1/2}|h|\geq h^2,\\
	|E_{n+1}^i(\theta^*+h)-\mathcal{E}_{n+1}^i(\theta^*)|&=|E_{n+1}^i(\theta^*+h)-E_{n+1}^i(2\theta_{n+1}^i-\theta^*)|\\
	&\geq 10\delta_n^{1/2}|2\theta^*+h-2\theta_{n+1}^i|\geq h^2,\\
|\mathcal{E}_{n+1}^i(\theta^*+h)-E_{n+1}^i(\theta^*)|&=|\mathcal{E}_{n+1}^i(\theta^*+h)-\mathcal{E}_{n+1}^i(2\theta_{n+1}^i-\theta^*)\\
&\geq 10\delta_n^{1/2}|2\theta^*+h-2\theta_{n+1}^i|\geq h^2,
	\end{align*}
	where $h=(c_{n+1}^j-c_{n+1}^i)\cdot \omega$ or $-((c_{n+1}^i+c_{n+1}^j)\cdot \omega+2\theta^*)$ {\rm (mod $1$)} satisfying $|h|=m(c_{n+1}^i,c_{n+1}^j)$. Now \eqref{changm}  follows from  $\{E_{n+1}^j(\theta^*),\mathcal{E}_{n+1}^j(\theta^*)\}=\{E_{n+1}^i(\theta^*+h),\mathcal{E}_{n+1}^i(\theta^*+h)\}$ since $H_{B_{n+1}^i}(\theta^*+h)=H_{B_{n+1}^j}(\theta^*)$ and one of the eigenvalue differences must be bounded above by $2\delta_{n+1}$ by the definition of $Q_{n+1}$.
\end{proof}

Finally, we also have 

\begin{thm}\label{daomb}
	For $|\theta-\theta^*|<10\delta_n^{1/2}$, we have 
	$$|\frac{d}{d \theta} E_{n+1}^i(\theta)|\geq\min(\delta_n^2,|\theta-\theta_{n+1}^i|).$$
\end{thm}
\begin{proof}We consider two cases:\smallskip \\
	\textbf{Case I}. $E_{n+1}^i(\theta_{n+1}^i)> \mathcal{E}_{n+1}^i(\theta_{n+1}^i)$. It immediately follows from Lemma \ref{C2} and  \textbf{(c)} of Proposition \ref{gujim}.\\
	\textbf{Case II}. $E_{n+1}^i(\theta_{n+1}^i)= \mathcal{E}_{n+1}^i(\theta_{n+1}^i)$. In this case,  we have $|\frac{d}{d \theta} E_{n+1}^i(\theta)|\geq10\delta_n^{1/2}\geq \delta_n^2.$
\end{proof}
\begin{thm}
	If $c_{n+1}^i \in Q_{n+1}$, then
	$$
	|E_{n+1}^i(\theta)-\mathcal{E}_{n+1}^i(\theta)| \geq \delta_n^2 |\theta-\theta_{n+1}^i|
	$$
	for all $\theta$ in the interval of $|\theta-\theta^*| \leq 10 \delta_n^{1 / 2}$.
\end{thm}
\begin{proof}
	We  consider two cases.\smallskip\\
	\textbf{Case I.} $E_{n+1}^i(\theta_{n+1}^i)>\mathcal{E}_{n+1}^i(\theta_{n+1}^i)$. Then 
	$$\frac{d}{d\theta}E_{n+1}^i(\theta_{n+1}^i)=\frac{d}{d\theta}\mathcal{E}_{n+1}^i(\theta_{n+1}^i)=0$$
	and by \eqref{2108},
	$$
	|\langle\psi_{n+1}^i, V' \Psi_{n+1}^i\rangle(\theta_{n+1}^i)| \geq \frac{1}{2}\delta_{m-1}^{2}\geq\delta_{n-1}^{2}.
	$$
	Therefore, there must be a largest  interval $\theta_{n+1}^i \leq \theta \leq \theta_d$,  where $|\langle\psi_{n+1}^i, V' \Psi_{n+1}^i\rangle(\theta)| \geq  \delta_{n-1}^{2} $. If $\theta$ is in this interval, then
	\begin{align*}
		(E_{n+1}^i-\mathcal{E}_{n+1}^i)(\theta)&=  (E_{n+1}^i-\mathcal{E}_{n+1}^i)(\theta_{n+1}^i)+\frac{d}{d \theta}(E_{n+1}^i-\mathcal{E}_{n+1}^i)(\theta_{n+1}^i)\cdot (\theta-\theta_{n+1}^i) \\
		&\ \  +\frac{1}{2} \frac{d^2}{d \theta^2}(E_{n+1}^i-\mathcal{E}_{n+1}^i)(\xi)\cdot (\theta-\theta_{n+1}^i)^2\\
		&\geq\frac{1}{2} \frac{d^2}{d \theta^2}(E_{n+1}^i-\mathcal{E}_{n+1}^i)(\xi)\cdot (\theta-\theta_{n+1}^i)^2.
	\end{align*}
	By \eqref{dfm} and \eqref{ddfm}, we have
	\begin{align*}
		\frac{d^2}{d \theta^2}(E_{n+1}^i-\mathcal{E}_{n+1}^i)(\xi) & =\frac{4\langle\psi_{n+1}^i, V' \Psi_{n+1}^i\rangle^2(\xi)}{(E_{n+1}^i-\mathcal{E}_{n+1}^i)(\xi)}+O(\delta_{n-1}^{-1}) \\
		& \geq\frac{2\delta_{n-1}^4}{(E_{n+1}^i-\mathcal{E}_{n+1}^i)(\theta)},
	\end{align*}
	 which implies
	$$
	(E_{n+1}^i-\mathcal{E}_{n+1}^i)(\theta) \geq  \frac{\delta_{n-1}^4 }{(E_{n+1}^i-\mathcal{E}_{n+1}^i)(\theta)}(\theta-\theta_{n+1}^i)^2
	$$
	and proves the theorem.
	We now consider the case when $\theta \geq \theta_d$. By the argument in the proof of  Theorem \ref{C1m} (\textbf{Case I}), we have 
	$$
	\frac{d}{d \theta} E_{n+1}^i \geq 10\delta_n^{1/2} \text{ and }  \frac{d}{d \theta} \mathcal{E}_{n+1}^i \leq-10\delta_n^{1/2}, 
	$$
	for $\theta\geq\theta_d$, which gives us 
	\begin{align*}
		(E_{n+1}^i-\mathcal{E}_{n+1}^i)(\theta) & =(E_{n+1}^i-\mathcal{E}_{n+1}^i)(\theta_d)+\frac{d}{d \theta}(E_{n+1}^i-\mathcal{E}_{n+1}^i)(\xi)\cdot (\theta-\theta_d) \\
		& \geq(E_{n+1}^i-\mathcal{E}_{n+1}^i)(\theta_d)+20 \delta_n^{1/2} (\theta-\theta_d)\\
		&\geq\delta_{n-1}^{2} (\theta_d-\theta_{n+1}^i)+20 \delta_n^{1/2} (\theta-\theta_d)\\
		&\geq\delta_n^2 (\theta-\theta_{n+1}^i).
	\end{align*}
	\textbf{Case II}. $E_{n+1}^i(\theta_{n+1}^i)= \mathcal{E}_{n+1}^i(\theta_{n+1}^i)$.
	In this case,   we have $\frac{d}{d \theta} E_{n+1}^i\geq10\delta_n^{1/2}$, $\frac{d}{d \theta} \mathcal{E}_{n+1}^i\leq-10\delta_n^{1/2}$ and
	\begin{align*}
		|(E_{n+1}^i-\mathcal{E}_{n+1}^i)(\theta)| & =|(E_{n+1}^i-\mathcal{E}_{n+1}^i)(\theta_{n+1}^i)+\frac{d}{d \theta}(E_{n+1}^i-\mathcal{E}_{n+1}^i)(\xi)\cdot (\theta-\theta_{n+1}^i)| \\
		& \geq 20 \delta_n^{1/2} |\theta-\theta_{n+1}^i|\geq \delta_n^2|\theta-\theta_{n+1}^i|.
	\end{align*}
\end{proof}

\subsubsection{} {\bf Case 2.}  $s_n<10l_n^2$. In this case, $l_{n+1}=l_n^4$.  Every $c_n^i\in Q_n$ has a mirror image $\tilde{c}_n^i$ such that $m(c_n^i,\tilde{c}_n^i)=\|(c_n^i+\tilde{c}_n^i)\cdot \omega+2\theta^*\|\leq6\delta_n^{1/2}$ and $\|c_n^i-\tilde{c}_n^i\|_1=s_n$. The center set  of the $(n+1)$-th stage blocks is defined as 
$$P_{n+1}=\big\{c_{n+1}^i=(c_n^i+\tilde{c}_n^i)/2:\ c_n^i\in Q_n \big\},$$ 
and 
$$Q_{n+1}=\big\{c_{n+1}^i\in P_{n+1}:\ \operatorname{dist}(\sigma(H_{B_{n+1}^i}(\theta^*)),E^*)<\delta_{n+1}:=e^{-l_{n+1}^{2/3}}\big\}. $$
An analog of Lemma \ref{514} shows that there exists $\mu_{n+1}=0\text{ or }1/2$ such that for every $c_{n+1}^i \in Q_{n+1}$, we have \begin{equation}\label{pren}
	\|\theta^*- c_{n+1}^i\cdot \omega+\mu_{n+1}\|\leq 3\delta_n^{1/2},
\end{equation}
 which implies that there exists a symmetric point $\theta_{n+1}^i$ satisfying  \begin{equation}\label{smn}
\theta_{n+1}^i:=-c_{n+1}^i\cdot \omega+\mu_{n+1} \ { \rm (mod\ 1)},\  |\theta_{n+1}^i-\theta^*|\leq3\delta_n^{1/2}.
\end{equation}
In this case, we must have $s_{n-1}\geq 10l_{n-1}^2$ since by the {\bf Center Theorem} (of stage $n-1$), a third $(n-1)$-singular block inside the $l_{n+1}(\sim l_n^4)$-size block $B_{n+1}^i$ is  excluded. Thus $c_{n-1}^i=c_n^i,\tilde{c}_{n-1}^i=\tilde{c}_n^i$ and moreover the set  $\Lambda=B_{n+1}^i\setminus(B_{n-1}^i\cup \tilde{B}_{n-1}^i)$ is $(n-1)$-good. Notice that by the Diophantine condition 
\begin{align}\label{1209n}
	\|2c_n^i\cdot \omega+2\theta^*\|&\geq\|c_n^i-\tilde{c}_n^i\|-\|(c_n^i+\tilde{c}_n^i)\cdot \omega+2\theta^*\|\\
	\nonumber&\geq \frac{\gamma}{s_n^\tau}-6\delta_n^{1/2}>3\delta_{n-1}^{1/2}.
\end{align}
So it is not the case of \textbf{(H6)} in {\bf Class B}. Thus, $c_n^i$ belongs to {\bf Class A} and \textbf{(H1)}--\textbf{(H5)} hold true. For $|\theta-\theta^*|=O(\delta_n^{1/2})$, since \begin{equation}\label{chabien}
	H_{\tilde{B}_n^i}(\theta)=H_{B_n^i}(-\theta-(c_n^i+\tilde{c}_n^i)\cdot \omega)=H_{B_n^i}(\theta+O(\delta_n^{1/2})),
\end{equation}
 there is also a unique eigenvalue $\tilde{E}_n^i(\theta)$ of $H_{\tilde{B}_n^i}(\theta)$ so that $|\tilde{E}_n^i(\theta)-E^*|=O(\delta_n^{1/2})$ and 
   the corresponding eigenfunction $\tilde{\psi}_n^i$  decays exponentially fast away from $\tilde{c}_n^i$. We are now in a similar setting as \textbf{Subcase B} of \textbf{Case 1} and the analogs   of the proposition hold true if we replace $m$ by $n$. We will list these propositions, however,  sketch the proofs that can be trivially  established from replacing $m$ by $n$.  We only concentrate on the nontrivial ones.  Now we  show  how to get back to  {\bf Class B}  of the  induction hypothesis from \textbf{Case 2}.
\begin{prop}\label{325n}Let $c_{n+1}^i\in Q_{n+1}.$ Then for $|\theta-\theta^*|<10\delta_{n}^{1/2}$, 
	\begin{itemize}
		\item [\textbf{(a)}] $H_{B_{n+1}^i}(\theta)$ has  exactly two eigenvalues $E_{n+1}^i(\theta)$ and $\mathcal{E}_{n+1}^i(\theta)$ in the interval of $|E-E^*|<50M_1\delta_{n}^{1/2}$. Moreover, any other $\hat{E}\in\sigma(H_{B_{n+1}^i}(\theta)) $ must obey $|\hat{E}-E^*|\geq\delta_{n-1}/6$.
		\item [\textbf{(b)}] The corresponding eigenfunction of $E_{n+1}^i$\ {\rm (resp. $\mathcal{E}_{n+1}^i$)}, $\psi_{n+1}$\ {\rm(resp. $\Psi_{n+1}$)} decays exponentially fast away from $c_n^i$ and $\tilde{c}_n^i$, 
		\begin{align}\label{1421n}
			\begin{split}
				|\psi_{n+1}(x)|\leq e^{-(\gamma_0/4)\|x-c_n^i\|_1}+e^{-(\gamma_0/4)\|x-\tilde{c}_n^i\|_1},\\
				|\Psi_{n+1}(x)|\leq e^{-(\gamma_0/4)\|x-c_n^i\|_1}+e^{-(\gamma_0/4)\|x-\tilde{c}_n^i\|_1}		
			\end{split}
		\end{align}
		for $\operatorname{dist}(x,\{c_n^i,\tilde{c}_n^i\})\geq l_n^{6/7}.$	
		\item [\textbf{(c)}] The two eigenfunctions can be expressed as 
		\begin{align}\label{1224n}
			\begin{split}
				\psi_{n+1}=A\psi_n+B\tilde{\psi}_n+O(\delta_n^{10}),\\
				\Psi_{n+1}=B\psi_n-A\tilde{\psi}_n+O(\delta_n^{10}),
			\end{split}	
		\end{align}
		where $A^2+B^2=1$.
		\item [\textbf{(d)}] $\|G_{B_{n+1}^i}^{\perp\perp}(E_{n+1}^i)\|\leq10\delta_{n-1}^{-1}$, where $G_{B_{n+1}^i}^{\perp\perp}$ denotes the Green's function for $B_{n+1}^i$ on the orthogonal complement of the space spanned by $\psi_{n+1}$ and $\Psi_{n+1}$. 
	\end{itemize}
\end{prop}
\begin{proof}
	By the exponential decay of $\psi_{n}$ and $\tilde{\psi}_n$, we have
	\begin{align*}
	\|(H_{B_{n+1}^i}(\theta^*)-E^*)\psi_n\|&\leq|E_n^i(\theta^*)-E^*|+\|\Gamma_{B_n^i}\psi_{n}\|\leq2\delta_n,\\
	\|(H_{B_{n+1}^i}(\theta^*)-E^*)\tilde{\psi}_n\|&\leq|\tilde{E}_n^i(\theta^*)-E^*|+\|\Gamma_{B_n^i}\tilde{\psi}_{n}\|\leq6\delta_{n}^{1/2}+2\delta_n.
	\end{align*}
 The two orthogonal trial wave functions  give two eigenvalues of $H_{B_{n+1}^i}(\theta^*)$  in $|E-E^*|<10\delta_n^{1/2}$ by Corollary \ref{trialcor}. Using $|V'|\leq M_1$,  we deduce $H_{B_{n+1}^i}(\theta)$  has  at least two eigenvalues in $|E-E^*|<50M_1\delta_{n}^{1/2}$, which proves the existence part of \textbf{(a)}. To prove  \textbf{(b)}, we  restrict the equation  $H_{B_{n+1}}(\theta)\psi_{n+1}=E_{n+1}(\theta)\psi_{n+1}$ to the $(n-1)$-good  set $A=B_{n+1}^i\setminus(B_{n-1}^i\cup \tilde{B}_{n-1}^i)$ to obtain 
	$$ \psi_{n+1}(x)=\sum_{z,z'}G_A(x,z)\Gamma_A\psi_{n+1}(z'),$$ which gives \eqref{1421n}. 
	Now we establish \textbf{(c)}. It suffices to show $\psi_{n+1}$ and $\Psi_{n}$ are close to  a linear combination of $\psi_n$ and $\tilde{\psi}_n$ inside $B_n^i\cup \tilde{B}_n^i$.
	We restrict the equation $H_{B_{n+1}^i}(\theta)\psi_{n+1}=E_{n+1}^i(\theta)\psi_{n+1}$ to $B_n^i$ to get 
	$$\left(H_{B_{n}^i}-E_{n+1}^i\right)\psi_{n+1}=\Gamma_{B_{n}^i}\psi_{n+1}.$$
	Combining \eqref{1421n} and the above equation, we get  $$\|P_n^\perp\psi_{n+1}\|=\|G_{B_n^i}^\perp(E_{n+1}^i)P_n^\perp\Gamma_{B_{n}^i}\psi_{n+1}\|= O(\delta_{n-1}^{-1}e^{-\frac{1}{4}\gamma_0l_n})\leq\frac{1}{2}\delta_n^{10},$$
	where $P_n^\perp$ is the projection onto the orthogonal complement of $\psi_n$ and $G_{B_n^i}^\perp(E_{n+1}^i)$ is the Green's function of  $B_n^i$ on $\operatorname{Range}P_n^\perp$ with the upper bound
	\begin{align}\label{839n}
		\|G_{B_n^i}^\perp(E_{n+1}^i)\|&\leq\operatorname{dist}(\sigma(H_{B_n^i}(\theta))- E_n^i(\theta),E_{n+1})^{-1}\\
		\nonumber&\leq(\frac{\delta_{n-1}}{5}-\frac{\delta_{n-1}}{6})^{-1}\leq\frac{30}{\delta_{n-1}}
	\end{align}by \textbf{(H1)} of stage $n$. 
	Therefore inside $B_n^i$, we have 
	$$P_n^\perp\psi_{n+1}=O(\delta_n^{10})$$
	and hence,
	$$\psi_{n+1}\chi_{B_n^i}=a\psi_n+O(\delta_n^{10}),$$
	where $a=\langle \psi_{n+1},\psi_n\rangle$.
	By the approximation \eqref{chabien}, we get a similar estimate in $\tilde{B}_n^i$
	$$\psi_{n+1}\chi_{\tilde{B}_n^i}=b\tilde{\psi}_n+O(\delta_n^{10})$$
	with $b=\langle \psi_{n+1},\tilde{\psi}_n\rangle$. By \eqref{1421n}, we have $\|\psi_{n+1}\chi_{\tilde{B}_{n+1}^i\setminus (B_n^i\cup \tilde{B}_n^i)}\|\leq \delta_n^{10}$. Thus, we can write 
	$$\psi_{n+1}=a\psi_n+b\tilde{\psi}_n+O(\delta_n^{10}).$$
	Taking norm gives $k:=a^2+b^2=1-O(\delta_n^{10})$. We set $A=a/k$ and $B=b/k$. Hence, $A^2+B^2=1$ and $|A-a|,|B-b|= O(\delta_n^{10})$, which gives the desired expression of $\psi_{n+1}$. A similar argument gives $\Psi_{n+1}=C\psi_m+D\tilde{\psi}_n+O(\delta_n^{10})$ with $C^2+D^2=1$. For convenience, we write $A=\cos\alpha, B=\sin\alpha, C=\sin\beta, D=-\cos\beta$. Using $\langle\psi_{n+1} ,\Psi_{n+1}\rangle=0$, we get  $|\sin(\beta-\alpha)|=O(\delta_n^{10})$. We can choose $\beta$ satisfying $|\beta-\alpha|=O(\delta_n^{10})$. Thus $|B-C|=|\sin\alpha-\sin\beta|=O(\delta_n^{10})$ and $|A+D|=|\cos\alpha-\cos\beta|=O(\delta_n^{10})$, giving the desired expression $\Psi_{n+1}=B\psi_n-A\tilde{\psi}_n+O(\delta_n^{10})$. Now assume that $\hat{E}\in\sigma(H_{B_{n+1}^i}(\theta))$ is a third eigenvalue in the interval of $|\hat{E}-E^*|<\delta_{n-1}/6$. The Green's function estimates and   \eqref{839n} still hold if we replace $E_{n+1}$ by $\hat{E}$. Thus, by a similar argument, the eigenfunction of $\hat{E}$ can be expressed as 
	$$\hat{\psi}=\hat{A}\psi_n+\hat{B}\tilde{\psi}_n+O(\delta_n^{10})$$ 
	with $\hat{A}^2+\hat{B}^2=1.$ By the orthogonality, we have $A\hat{A}+B\hat{B}=O(\delta_n^{10})$ and $B\hat{A}-A\hat{B}=O(\delta_n^{10})$. This is impossible since  $(A\hat{A}+B\hat{B})^2+(B\hat{A}-A\hat{B})^2=1$.
	So,  a third eigenvalue must obey $|\hat{E}-E^*|\geq\delta_{n-1}/6$. Finally,  \textbf{(d)} follows from \textbf{(a)} immediately.
\end{proof}
 We also have 
\begin{lem}\label{chan}
	For $|\theta-\theta^*|<10\delta_n^{1/2}$, we have 
	\begin{equation}\label{n+}
		|\frac{d}{d\theta}\left(E_n^i+\tilde{E}_n^i\right)(\theta)|\leq \delta_n^{1/3}.
	\end{equation}
\end{lem}
\begin{proof}
	From \eqref{chabien}, we obtain  $\tilde{E}_n^i(\theta)=E_n^i(-\theta+2\theta_{n+1}^i)$.
	Thus,
	\begin{align*}
		|\frac{d}{d\theta}\left(E_n^i+\tilde{E}_n^i\right)(\theta)|&=|\frac{d}{d\theta}E_n^i(\theta)-\frac{d}{d\theta}E_n^i(-\theta+2\theta_{n+1}^i)|\\
		&= |\frac{d^2}{d\theta^2}E_n^i(\xi)|\cdot |2\theta-2\theta_{n+1}^i|\\
		&\leq O(\delta_{n-1}^{-1}\delta_n^{1/2})
		\\
		&\leq \delta_n^{1/3},
	\end{align*}
	where on the third line we used the estimate 
	\begin{align*}
		|\frac{d^2}{d \theta^2} E_{n}^i(\theta)|&=|\left\langle\psi_{n}, V'' \psi_{n}\right\rangle-2\left\langle\psi_{n}, V' G_{B_{n}^i}^{\perp}(E_{n}^i) V' \psi_{n}\right\rangle|\\
		&\leq O(\|G_{B_{n}^i}^{\perp}(E_{n}^i)\|)\\
		&\leq O(\delta_{n-1}^{-1})
	\end{align*}
	for $|\theta-\theta^*|<10\delta_n^{1/2}<\delta_{n-1}/(10M_1)$ by \textbf{(H1)} of  stage $n$.
\end{proof}
\begin{lem}\label{xiajien}
	For $|\theta-\theta^*|<10\delta_n^{1/2}$, we have $|\frac{d}{d\theta}E_n^i(\theta)|\geq\delta_{n-1}^{2}$.
\end{lem}
\begin{proof}
	Assume  it is not true. By \textbf{(H4)} and recalling \eqref{1209n}, we have  
	\begin{align*}
		|\frac{d}{d\theta}E_n^i(\theta)|&\geq \min(\|\theta+c_n^i\cdot \omega\|,\|\theta+c_n^i\cdot \omega-\frac{1}{2}\|)\\
		&\geq \frac{1}{2}\|2\theta+2c_n^i\cdot \omega\|\\
		&>\delta_{n-1},
	\end{align*}
which leads to a contradiction.
\end{proof} 

\begin{prop}\label{gujin}
	Let $|\theta-\theta^*|<10\delta_n^{1/2}$. Then 
	\begin{itemize}
		\item [\textbf{(a)}] $E_{n+1}^i$ and $\mathcal{E}_{n+1}^i$ are $C^1$ functions and if $ E_{n+1}^i(\theta)\neq\mathcal{E}_{n+1}^i(\theta)$, then
		\begin{align}
			\frac{d}{d \theta} E_{n+1}^i & =(A^2-B^2) \frac{d}{d\theta} E_n^i+O(\delta_{n}^{1/3} ), \label{den}\\
			\frac{d}{d \theta} \mathcal{E}_{n+1}^i & =(B^2-A^2) \frac{d}{d\theta} E_n^i+O(\delta_{n}^{1/3} ).\nonumber
		\end{align}
		\item [\textbf{(b)}] If $ E_{n+1}^i(\theta)\neq\mathcal{E}_{n+1}^i(\theta)$, then $\frac{d^2}{d \theta^2} E_{n+1}^i(\theta)$ and 	$\frac{d^2}{d \theta^2} \mathcal{E}_{n+1}^i(\theta) $ exist. Moreover,
		\begin{align}
			\frac{d^2}{d \theta^2} E_{n+1}^i & =\frac{2\left\langle\psi_{n+1}^i, V' \Psi_{n+1}^i\right\rangle^2}{E_{n+1}^i-\mathcal{E}_{n+1}^i}+O(\delta_{n-1}^{-1} ),\label{dfn} \\
			\frac{d^2}{d \theta^2} \mathcal{E}_{n+1}^i & =\frac{2\left\langle\psi_{n+1}^i, V' \Psi_{n+1}^i\right\rangle^2}{\mathcal{E}_{n+1}^i-E_{n+1}^i}+O(\delta_{n-1}^{-1} ). \label{ddfn}
		\end{align}
		\item [\textbf{(c)}] At the point $E_{n+1}^i(\theta)\neq \mathcal{E}_{n+1}^i(\theta)$, if $|\frac{d}{d\theta}E_{n+1}^i(\theta)|\leq10\delta_n^{1/2}$, then $|\frac{d^2}{d\theta^2}E_{n+1}^i(\theta)|>\delta_n^{-1/3}>2$. Moreover, the  sign of $\frac{d^2}{d\theta^2}E_{n+1}^i(\theta)$ is the same as that of $ E_{n+1}^i(\theta)- \mathcal{E}_{n+1}^i(\theta)$. The analogous conclusion holds by exchanging  $E_{n+1}^i(\theta)$ and $\mathcal{E}_{n+1}^i(\theta)$.
	\end{itemize}
\end{prop}
\begin{proof} When $E_{n+1}^i$ is simple, by \eqref{1224n} and Lemma \ref{chan}, we have 
	\begin{align*}
		\frac{d}{d \theta} E_{n+1}^i & =\left\langle\psi_{n+1}^i, V' \psi_{n+1}^i\right\rangle=A^2 \frac{d}{d \theta} E_n^i+B^2 \frac{d}{d \theta} \tilde{E}_n^i+O(\delta_n^{10} ) \\
		& =(A^2-B^2) \frac{d}{d \theta} E_n^i+B^2(\frac{d}{d \theta} E_n^i+\frac{d}{d \theta} \tilde{E}_n^i)+O(\delta_n^{10} )\\
		& =(A^2-B^2) \frac{d}{d\theta} E_n^i+O(\delta_n^{1/3} ),
	\end{align*}
	where we used \eqref{n+} in the last estimate and complete the proof of  \textbf{(a)}.
	To prove  \textbf{(b)},  we use the formula 
	\begin{align*}
	\frac{d^2}{d \theta^2} E_{n+1}^i&=\left\langle\psi_{n+1}^i, V'' \psi_{n+1}^i\right\rangle+2 \frac{\left\langle\psi_{n+1}^i, V' \Psi_{n+1}^i\right\rangle^2}{E_{n+1}^i-\mathcal{E}_{n+1}^i}\\
	&\ \ -2\left\langle V' \psi_{n+1}^i,G^{\perp \perp}_{B_{n+1}^i}(E_{n+1}^i) V' \psi_{n+1}^i\right\rangle .
	\end{align*}
	The remainder term is bounded by $2\|G^{\perp \perp}_{B_{n+1}^i}(E_{n+1}^i)\|\cdot\|V' \psi_{n+1}^i\|^2$, where we can use the estimate $\|G^{\perp \perp}_{B_{n+1}^i}(E_{n+1}^i)\| \leq10\delta_{n-1}^{-1}$ in \textbf{(d)} of Proposition \ref{325n}. Now we turn to the proof of \textbf{(c)}. If $|\frac{d}{d\theta}E_{n+1}^i(\theta)|\leq10\delta_n^{1/2}$, then by \eqref{den},  we have 
	$$|A^2-B^2|\cdot  |\frac{d}{d\theta}E_n^i(\theta)|\leq10\delta_n^{1/2}+O(\delta_n^{1/3}) \leq\delta_n^{1/2},$$
	which implies $A^2\approx B^2\approx \frac{1}{2}$ by Lemma \ref{xiajien}. Thus,
	\begin{align}\label{2107}
\begin{split}
			|\langle\psi_{n+1}^i, V' \Psi_{n+1}^i\rangle|&=|AB \frac{d}{d \theta} E_n^i-AB \frac{d}{d \theta} \tilde{E}_n^i+O(\delta_n^{1 / 3} )| \\
		& \geq2AB |\frac{d}{d\theta} E_n^i|-O(\delta_n^{1 / 3} )\\
		&\geq\frac{1}{2}\delta_{n-1}^{2}.
\end{split}
	\end{align}
	By \textbf{(a)} of Proposition \ref{325n}, we have $|E_{n+1}^i-\mathcal{E}_{n+1}^i|\leq100M_1\delta_{n}^{1/2}$.  Using  \eqref{dfn}, we obtain  $|\frac{d}{d\theta}E_{n+1}^i(\theta)|\geq\frac{1}{4}\delta_{n-1}^{4}(100M_1\delta_n^{1/2})^{-1}-O(\delta_{n-1}^{-1})>\delta_n^{-1/3}$, whose the sign is determined by that of $ E_{n+1}^i(\theta)- \mathcal{E}_{n+1}^i(\theta)$.
\end{proof}
Since $H_{B_{n+1}^i}(\theta^*+h)=H_{B_{n+1}^j}(\theta^*)$, we deduce from \textbf{(a)} of Proposition \ref{325n}  and  Lemma \ref{qiang} that  $H_{B_{n+1}^j}(\theta^*)$ also has exactly  two eigenvalues $E_n^j, \mathcal{E}_n^j$ in the interval of $|E-E^*|\leq 50M_1\delta_n^{1/2}$ satisfying $\{E_{n}^j,\mathcal{E}_{n}^j\}=\{E_{n}^i(\theta^*+h),\mathcal{E}_{n}^i(\theta^*+h)\} .$\smallskip

The {\bf Center Theorem} of stage $n+1$ in {\bf Case 2} is 
\begin{thm}[]\label{C1n} For  any $c_{n+1}^i,c_{n+1}^j\in Q_{n+1}$ we have 
	\begin{align}\label{changn}
		m(c_{n+1}^i,c_{n+1}^j)&\leq\sqrt{2}\min(|E_{n+1}^i(\theta^*)-E_{n+1}^j(\theta^*)|^{1/2},|\mathcal{E}_{n+1}^i(\theta^*)-\mathcal{E}_{n+1}^j(\theta^*)|^{1/2},\nonumber\\
		&\ \ |E_{n+1}^i(\theta^*)-\mathcal{E}_{n+1}^j(\theta^*)|^{1/2},|\mathcal{E}_{n+1}^i(\theta^*)-E_{n+1}^j(\theta^*)|^{1/2})\\
		\nonumber&\leq2\delta_{n+1}^\frac{1}{2}.
	\end{align}
\end{thm}
\begin{proof}
Using \eqref{pren} gives us $m(c_{n+1}^i,c_{n+1}^j)\leq6\delta_n^{1/2}$, which implies that $\theta^*\pm m(c_{n+1}^i,c_{n+1}^j)$ belongs to the interval of $|\theta-\theta^*|<10\delta_n^{1/2}$,  where $E_{n+1}^i$ and $\mathcal{E}_{n+1}^i$ are well defined.  We also  recall \eqref{smn} that the symmetric point $\theta_{n+1}^i$ belongs to the interval of $|\theta-\theta^*|<10\delta_n^{1/2}$. So there will be two cases.\smallskip\\
	\textbf{Case I}.
	$E_{n+1}^i(\theta_{n+1}^i) \neq \mathcal{E}_{n+1}^i(\theta_{n+1}^i)$. 
	Without loss of generality,  we may assume $E_{n+1}^i(\theta_{n+1}^i)>\mathcal{E}_{n+1}^i(\theta_{n+1}^i)$. Notice that union of two eigenvalue curves is symmetric about $\theta_{n+1}^i$. Thus,  we must have
	$$
	\frac{d}{d \theta} E_{n+1}^i(\theta_{n+1}^i)=\frac{d}{d \theta} \mathcal{E}_{n+1}^i(\theta_{n+1}^i)=0 .
	$$ 
	By  \textbf{(b)} and \textbf{(c)} of Proposition \ref{gujin}, we see that $\theta_{n+1}^i$ is a local minimum point of $ E_{n+1}^i$ and a local maximum one of $ \mathcal{E}_{n+1}^i$. Moreover, $\frac{d}{d \theta} E_{n+1}^i$ is increasing and $\frac{d}{d \theta} \mathcal{E}_{n+1}^i$ is decreasing whenever $|\frac{d}{d \theta} E_{n+1}^i|\leq10\delta_{n}^{1/2}$. Thus, $E_{n+1}^i>\mathcal{E}_{n+1}^i$ continues to hold for all $|\theta-\theta^*|<10\delta_n^{1/2}$, which implies, in particular, $\frac{d^2}{d\theta^2}E_{n+1}^i(\theta)>2$ whenever $|\frac{d}{d\theta}E_{n+1}^i(\theta)|<10\delta_n^{1/2}$. Moreover, $\frac{d}{d\theta}E_{n+1}^i$ (resp. $\frac{d}{d\theta}\mathcal{E}_{n+1}^i$) cannot reenter the band $|\frac{d}{d\theta}E|<10\delta_n^{1/2}$ since it is increasing (resp. decreasing) there.   From  the preliminary bound   $m(c_{n+1}^i,c_{n+1}^j)\leq6\delta_n^{1/2}$, we deduce that     $E_{n+1}^i(\theta),\mathcal{E}_{n+1}^i(\theta)$ satisfy the condition of  Lemma \ref{C2} with $\theta_2=\theta^*+h,\theta_1=\theta^*,\delta=10\delta_n^{1/2},|h|\leq\delta.$ Thus, we get
	\begin{align*}
	|E_{n+1}^i(\theta^*+h)-E_{n+1}^i(\theta^*)|&\geq \frac{1}{2}\min(h^2,|2\theta^*+h-2\theta_{n+1}^i|^2)\\
	&=\frac{1}{2}h^2
	\end{align*}
	and the same estimate holds  true for $\mathcal{E}_{n+1}^i$,
	where $h=(c_{n+1}^j-c_{n+1}^i)\cdot \omega$ or $-((c_{n+1}^i+c_{n+1}^j)\cdot \omega+2\theta^*)$ {\rm (mod $1$)} satisfying $|h|=m(c_{n+1}^i,c_{n+1}^j)$.
	An easy inspection gives us 
	\begin{align*}
	&\ \ \ |\mathcal{E}_{n+1}^i(\theta^*+h)-E_{n+1}^i(\theta^*)|\\
	&\geq\min(|E_{n+1}^i(\theta^*+h)-E_{n+1}^i(\theta^*)|,|\mathcal{E}_{n+1}^i(\theta^*+h)-\mathcal{E}_{n+1}^i(\theta^*)|)\\
	&\geq\frac{1}{2}h^2,\\ 
	&\ \ \ |E_{n+1}^i(\theta^*+h)-\mathcal{E}_{n+1}^i(\theta^*)|\\
	&\geq\min(|E_{n+1}^i(\theta^*+h)-E_{n+1}^i(\theta^*)|,|\mathcal{E}_{n+1}^i(\theta^*+h)-\mathcal{E}_{n+1}^i(\theta^*)|)\\&\geq\frac{1}{2}h^2.
	\end{align*}
	Now \eqref{changn}  follows from  $\{E_{n+1}^j(\theta^*),\mathcal{E}_{n+1}^j(\theta^*)\}=\{E_{n+1}^i(\theta^*+h),\mathcal{E}_{n+1}^i(\theta^*+h)\}$ since $H_{B_{n+1}^i}(\theta^*+h)=H_{B_{n+1}^j}(\theta^*)$, and one of the eigenvalue differences must be bounded above by $2\delta_{n+1}$ by the definition of $Q_{n+1}$.\smallskip\\
	\textbf{Case II}. $E_{n+1}^i(\theta_{n+1}^i)= \mathcal{E}_{n+1}^i(\theta_{n+1}^i)$. 
	In this case,  we claim that $|\frac{d}{d \theta} E_{n+1}^i|\geq10\delta_n^{1/2}$ and $|\frac{d}{d \theta} \mathcal{E}_{n+1}^i|\geq10\delta_n^{1/2}$ hold for $|\theta-\theta^*|<10\delta_n^{1/2}.$ Moreover, they have opposite signs. First,  we show it is true for $\theta=\theta_{n+1}^i$. An analog of  Lemma \ref{daogroup} gives us 
	\begin{align*}
	&\ \ \ \{\frac{d}{d \theta} E_{n+1}^i(\theta_{n+1}^i),\frac{d}{d \theta} \mathcal{E}_{n+1}^i(\theta_{n+1}^i)\}\\
	&=\{\text{Eigenvalues of the }2\times2\text{ matrix } PH'_{B_{n+1}^i}(\theta_{n+1}^i)P\},
	\end{align*}
	 where $P$ is the projection onto the two dimensional eigenspace of $E_{n+1}^i(\theta_{n+1}^i)$.
	To calculate these eigenvalues,  we represent $PV'P:=PH'P$ in a special basis. Notice that $H_{B_{n+1}^i}(\theta_{n+1}^i)$ commutes with the reflect operator $(R\psi)(x):=\psi(2c_{n+1}^i-x)$. It follows that $\operatorname*{Range}P$ is a two dimensional  invariant subspace of $R$, which can be spanned by two eigenfunctions of $R$ since $R$ is  diagonalizable. All the eigenfunctions of $R$ are symmetric functions $\{\psi_s\}$ and antisymmetric functions $\{\psi_a\}$. We note that   $\operatorname*{Range}P$ cannot be spanned by only symmetric functions (resp. antisymmetric functions), otherwise $\psi_{n+1}$ and $\Psi_{n+1}$ are  symmetric  (resp. antisymmetric), contradicting  the expression \eqref{1224}. This allows us to express $PV'P$ in the basis $\{\psi_s,\psi_a\}$, which consists of one symmetric function and antisymmetric function
	$$P V' P=\left(\begin{array}{ccc}
		\left\langle\psi_s, V' \psi_s\right\rangle & \left\langle\psi_s, V' \psi_a\right\rangle \\
		\left\langle\psi_s, V' \psi_a\right\rangle & \left\langle\psi_a, V' \psi_a\right\rangle
	\end{array}\right)\ (\text {at } \theta=\theta_{n+1}^i) .
	$$
	Since $v$ is even and $1$-periodic , we deduce $(V'(\theta_{n+1}^i))(2c_1^i-x)=v'(\theta_{n+1}^i+(2c_1^i-x)\cdot \omega)=-v'(\theta_{n+1}^i+x\cdot \omega)=-(V'(\theta_{n+1}^i))(x)$, yielding $V'(\theta_{n+1}^i)$ is antisymmetric. Now by the symmetry (anti-symmetry) properties of $\psi_s, \psi_a$, and $V'(\theta_{n+1}^i)$, we have $\langle\psi_s, V' \psi_s\rangle=$ $\langle\psi_a, V' \psi_a\rangle=0$, which gives us
	$$
	P V' P=\left(\begin{array}{cc}
		0 & \left\langle\psi_s, V' \psi_a\right\rangle \\
		\left\langle\psi_s, V' \psi_a\right\rangle & 0
	\end{array}\right)
	$$
	and therefore
	$$
	\frac{d}{d\theta} E_{n+1}^i(\theta_{n+1}^i)=-\frac{d}{d\theta} \mathcal{E}_{n+1}^i(\theta_{n+1}^i)=\langle\psi_s, V' \psi_a\rangle.
	$$
	We choose $E_{n+1}^i$ to satisfy $\frac{d}{d\theta} E_{n+1}^i(\theta_{n+1}^i)\geq0$ and  show that it is not too small and then extend this for $|\theta-\theta^*| \leq 10 \delta_n^{1 / 2}$. Using the symmetry properties and the decay of the eigenfunctions, we have  $\psi_{n}= \pm R\tilde{\psi}_n+O(\delta_{n}^{10})$,  $\psi_s=1/\sqrt{2}(\psi_n+R\psi_n) +O(\delta_{n}^{10})$ and $\psi_a=1/\sqrt{2}(\psi_n-R\psi_n) +O(\delta_{n}^{10})$. So,
	$$\frac{d}{d\theta} E_{n+1}^i(\theta_{n+1}^i)=\langle\psi_n,V'\psi_{n}\rangle+O(\delta_n^{10})=\frac{d}{d\theta} E_{n}^i(\theta_{n+1}^i)+O(\delta_n^{10}).$$
	By Lemma \ref{xiajien}, we get
	$$
	\frac{d}{d\theta}E_{n+1}^i(\theta_{n+1}^i) \geq \frac{1}{2}\delta_{n-1}^2\geq10\delta_n^{1/2}.
	$$
	We now  show that this continues to hold for all $\theta$ in the interval $|\theta-\theta^*| \leq 10 \delta_n^{1 / 2}$. Since $E_{n+1}^i$ is increasing and $\mathcal{E}_{n+1}^i$ is decreasing, we deduce  $E_{n+1}^i>\mathcal{E}_{n+1}^i$ for $\theta>\theta_{n+1}^i$.  If $\frac{d}{d\theta}E_{n+1}^i(\theta)\leq 10\delta_n^{1/2} $ for some smallest $\theta>\theta_{n+1}^i$, by  \textbf{(c)} of Proposition \ref{gujin}, we have  $\frac{d^2}{d\theta^2}E_{n+1}^i(\theta)>0$. This is impossible. The same argument shows there is no $\theta<\theta_{n+1}^i$ such that $\frac{d}{d\theta}E_{n+1}^i(\theta)\leq 10\delta_n^{1/2} $, which proves our claim. In this case, we have  $E_{n+1}^i(\theta)=\mathcal{E}_{n+1}^i(2\theta_{n+1}^i-\theta)$ by the symmetry property of the eigenvalue curve. Thus, by the preliminary bound  $m(c_{n+1}^i,c_{n+1}^j)\leq6\delta_n^{1/2}$, we obtain 
	\begin{align*}
	|E_{n+1}^i(\theta^*+h)-E_{n+1}^i(\theta^*)|&\geq 10\delta_n^{1/2}|h|\geq h^2,\\
	 |\mathcal{E}_{n+1}^i(\theta^*+h)-\mathcal{E}_{n+1}^i(\theta^*)|&\geq 10\delta_n^{1/2}|h|\geq h^2,\\
	|E_{n+1}^i(\theta^*+h)-\mathcal{E}_{n+1}^i(\theta^*)|&=|E_{n+1}^i(\theta^*+h)-E_{n+1}^i(2\theta_{n+1}^i-\theta^*)|\\
	&\geq 10\delta_n^{1/2}|2\theta^*+h-2\theta_{n+1}^i|\geq h^2,\\
	|\mathcal{E}_{n+1}^i(\theta^*+h)-E_{n+1}^i(\theta^*)|&=|\mathcal{E}_{n+1}^i(\theta^*+h)-\mathcal{E}_{n+1}^i(2\theta_{n+1}^i-\theta^*)|\\
	&\geq 10\delta_n^{1/2}|2\theta^*+h-2\theta_{n+1}^i|\geq h^2,
	\end{align*}
	where $h=(c_{n+1}^j-c_{n+1}^i)\cdot \omega$ or $-((c_{n+1}^i+c_{n+1}^j)\cdot \omega+2\theta^*)$ {\rm (mod $1$)} satisfying $|h|=m(c_{n+1}^i,c_{n+1}^j)$. Now \eqref{changn} follows from  $\{E_{n+1}^j(\theta^*),\mathcal{E}_{n+1}^j(\theta^*)\}=\{E_{n+1}^i(\theta^*+h),\mathcal{E}_{n+1}^i(\theta^*+h)\}$ since $H_{B_{n+1}^i}(\theta^*+h)=H_{B_{n+1}^j}(\theta^*)$, and one of the eigenvalue differences must be bounded above by $2\delta_{n+1}$ by the definition of $Q_{n+1}$.
\end{proof}

\begin{thm}\label{daonb}
	For $|\theta-\theta^*|<10\delta_n^{1/2}$, we have 
	$$|\frac{d}{d \theta} E_{n+1}^i(\theta)|\geq\min(\delta_n^2,|\theta-\theta_{n+1}^i|).$$
\end{thm}
\begin{proof}We consider two cases. \smallskip\\
	\textbf{Case I}. $E_{n+1}^i(\theta_{n+1}^i)> \mathcal{E}_{n+1}^i(\theta_{n+1}^i)$. It immediately follows from Lemma \ref{C2} and \textbf{(c)} of Proposition \ref{gujin}.\\
	\textbf{Case II}. $E_{n+1}^i(\theta_{n+1}^i)= \mathcal{E}_{n+1}^i(\theta_{n+1}^i)$. In this case,  we have $|\frac{d}{d \theta} E_{n+1}^i(\theta)|\geq10\delta_n^{1/2}\geq \delta_n^2.$
\end{proof}
\begin{thm}
	If $c_{n+1}^i \in Q_{n+1}$, then
	$$
	|E_{n+1}^i(\theta)-\mathcal{E}_{n+1}^i(\theta)| \geq \delta_n^2 |\theta-\theta_{n+1}^i|
	$$
	for all $\theta$ in the interval  of $|\theta-\theta^*| \leq 10 \delta_n^{1 / 2}$.
\end{thm}
\begin{proof}
	We  consider two cases.\smallskip\\
	\textbf{Case I.} $E_{n+1}^i(\theta_{n+1}^i)>\mathcal{E}_{n+1}^i(\theta_{n+1}^i)$. Then 
	$$\frac{d}{d\theta}E_{n+1}^i(\theta_{n+1}^i)=\frac{d}{d\theta}\mathcal{E}_{n+1}^i(\theta_{n+1}^i)=0$$
	and by \eqref{2107},
	$$
	|\langle\psi_{n+1}^i, V' \Psi_{n+1}^i\rangle(\theta_{n+1}^i)| \geq \frac{1}{2}\delta_{n-1}^{2}.
	$$
	Therefore, there must be a largest  interval $\theta_{n+1}^i \leq \theta \leq \theta_d$, where $|\langle\psi_{n+1}^i, V' \Psi_{n+1}^i\rangle(\theta)| \geq \frac{1}{2} \delta_{n-1}^{2} $. If $\theta$ is in this interval, then
	\begin{align*}
		(E_{n+1}^i-\mathcal{E}_{n+1}^i)(\theta)&= (E_{n+1}^i-\mathcal{E}_{n+1}^i)(\theta_{n+1}^i)\\
		&\ \  +\frac{d}{d \theta}(E_{n+1}^i-\mathcal{E}_{n+1}^i)(\theta_{n+1}^i)\cdot (\theta-\theta_{n+1}^i) \\
		&\ \ +\frac{1}{2} \frac{d^2}{d \theta^2}(E_{n+1}^i-\mathcal{E}_{n+1}^i)(\xi)\cdot (\theta-\theta_{n+1}^i)^2\\
		&\geq\frac{1}{2} \frac{d^2}{d \theta^2}(E_{n+1}^i-\mathcal{E}_{n+1}^i)(\xi)\cdot (\theta-\theta_{n+1}^i)^2.
	\end{align*}
	By \eqref{dfn} and \eqref{ddfn}, we have
	\begin{align*}
		\frac{d^2}{d \theta^2}(E_{n+1}^i-\mathcal{E}_{n+1}^i)(\xi) & =\frac{4\langle\psi_{n+1}^i, V' \Psi_{n+1}^i\rangle^2(\xi)}{(E_{n+1}^i-\mathcal{E}_{n+1}^i)(\xi)}+O(\delta_{n-1}^{-1}) \\
		& \geq\frac{\delta_{n-1}^4}{2(E_{n+1}^i-\mathcal{E}_{n+1}^i)(\theta)}\\
		&\geq\frac{2\delta_{n}^4}{(E_{n+1}^i-\mathcal{E}_{n+1}^i)(\theta)},
	\end{align*}
which implies
	$$
	(E_{n+1}^i-\mathcal{E}_{n+1}^i)(\theta) \geq  \frac{\delta_{n}^4 }{(E_{n+1}^i-\mathcal{E}_{n+1}^i)(\theta)}(\theta-\theta_{n+1}^i)^2
	$$
	and proves the theorem.
	We now consider the case when $\theta \geq \theta_d$. By the argument in the proof of  Theorem \ref{C1n} (\textbf{Case I}), we have 
	$$
	\frac{d}{d \theta} E_{n+1}^i \geq 10\delta_n^{1/2} \text{ and }  \frac{d}{d \theta} \mathcal{E}_{n+1}^i \leq-10\delta_n^{1/2}, 
	$$
	for $\theta\geq\theta_d$, which gives us 
	\begin{align*}
		(E_{n+1}^i-\mathcal{E}_{n+1}^i)(\theta) & =(E_{n+1}^i-\mathcal{E}_{n+1}^i)(\theta_d)+\frac{d}{d \theta}(E_{n+1}^i-\mathcal{E}_{n+1}^i)(\xi)\cdot (\theta-\theta_d) \\
		& \geq(E_{n+1}^i-\mathcal{E}_{n+1}^i)(\theta_d)+20 \delta_n^{1/2} (\theta-\theta_d)\\
		&\geq\delta_{n}^{2} (\theta_d-\theta_{n+1}^i)+20 \delta_n^{1/2} (\theta-\theta_d)\\
		&\geq\delta_n^2 (\theta-\theta_{n+1}^i).
	\end{align*}
	\textbf{Case II}. $E_{n+1}^i(\theta_{n+1}^i)= \mathcal{E}_{n+1}^i(\theta_{n+1}^i)$.
	In this case,   we have $\frac{d}{d \theta} E_{n+1}^i\geq10\delta_n^{1/2}$ and $\frac{d}{d \theta} \mathcal{E}_{n+1}^i\leq-10\delta_n^{1/2}$. Thus 
	\begin{align*}
		|(E_{n+1}^i-\mathcal{E}_{n+1}^i)(\theta)| & =|(E_{n+1}^i-\mathcal{E}_{n+1}^i)(\theta_{n+1}^i) +\frac{d}{d \theta}(E_{n+1}^i-\mathcal{E}_{n+1}^i)(\xi)\cdot (\theta-\theta_{n+1}^i)| \\
		& \geq 20 \delta_n^{1/2} |\theta-\theta_{n+1}^i|\geq \delta_n^2|\theta-\theta_{n+1}^i|.
	\end{align*}
\end{proof}

Finally, we estimate  Green's functions on $(n+1)$-good sets.  

\begin{thm}\label{ng}If $\Lambda$ is $(n+1)$-good, then for all $|\theta-\theta^*|<\delta_{n+1}/(10M_1), |E-E^*|<\delta_{n+1}/5$, 
\begin{align*}
	\|G_\Lambda(\theta; E)\|&\leq10\delta_{n+1}^{-1},\\
	|G_\Lambda(\theta;E)(x,y)|&<e^{-\gamma_{n+1}\|x-y\|_1}\ {\rm for}\ \|x-y\|_1\geq l_{n+1}^\frac{5}{6},
	\end{align*}
where $\gamma_{n+1}=(1-O(l_{n+1}^{-\frac{1}{30}}))\gamma_{n}.$
\end{thm}
\begin{proof} The proof is similar to  that of Theorem \ref{1g}, which can be established via three key steps. 

First,  we consider the case when $\Lambda=B_{n+1}^i$ is a $(n+1)$-regular block. By the definition of $(n+1)$-regular, we have 
	\begin{equation*}
		\|G_{B_{n+1}^i}(\theta^*;E^*)\|\leq\delta_{n+1}^{-1}.
	\end{equation*}
	So  by the Neumann series argument, for $|\theta-\theta^*|<\delta_{n+1}/(10M_1)$ and  $|E-E^*|<\frac{2}{5}\delta_{n+1}$, 
	\begin{equation*}\label{L2n+1}
		\|G_{B_{n+1}^i}(\theta;E)\|\leq2\delta_{n+1}^{-1}.
	\end{equation*}
For convenience, we omit the dependence of Green functions on  $\theta$ and $ E$. Let $x,y\in B_{n+1}^i$ satisfy $\|x-y\|_1\geq l_{n+1}^\frac{4}{5}$. Since $G_{B_{n+1}^i}$ is self-adjoint,  we may  assume $\|x-c_{n+1}^i\|_1\geq l_{n+1}^\frac{3}{4}$. Let $I_{n+1}^i$ be a $l_{n+1}^\frac{2}{3}$-size block  centered at $c_{n+1}^i$ such that $A=B_{n+1}^i\setminus I_{n+1}^i$ is $n$-good.  Hence by induction hypothesis,  we have 
\begin{align*}\label{L2n} 
	\|G_A\|&\leq10\delta_n^{-1},\\
|G_A(x,y)|&\leq e^{-\gamma_n\|x-y\|_1}\ \text{for $\|x-y\|_1\geq l_n^\frac{5}{6}$.}
\end{align*}
Using the resolvent identity, we obtain 
	\begin{align*}
		|G_{B_{n+1}^i}(x,y)|&=|G_{A}(x,y)\chi_A(y)+\sum_{z,z'}G_{A}(x,z)\Gamma_{z,z'}G_{B_{n+1}^i}(z',y)|\\
		&\leq e^{-\gamma_n\|x-y\|_1}+C(d)\sup_{z,z'}e^{-\gamma_n\|x-z\|_1}|G_{B_{n+1}^i}(z',y)|\\
		&\leq e^{-\gamma_n\|x-y\|_1}+C(d)\sup_{z,z'}e^{-\gamma_n\|x-z\|_1}e^{-\gamma_n(\|z'-y\|_1-l_{n+1}^\frac{3}{4})}\delta_{n+1}^{-1}\\
		&\leq e^{-\gamma'_n\|x-y\|_1}
	\end{align*}
	with  $\gamma'_{n}=(1-O(l_{n+1}^{-\frac{1}{30}}))\gamma_{n}$, where we have used if $\|z'-y\|\leq l_{n+1}^\frac{3}{4}$, 
	\begin{equation*}
		|G_{B_{n+1}^i}(z',y)|\leq 	\|G_{B_{n+1}^i}\|\leq2\delta_{n+1}^{-1}\leq 2e^{-\gamma_n(\|z'-y\|_1-l_{n+1}^\frac{3}{4})}\delta_{n+1}^{-1}, 
	\end{equation*}
and if $\|z'-y\|\geq l_{n+1}^\frac{3}{4}$, 
	\begin{align*}
	|G_{B_{n+1}^i}(z',y)|=	|G_{B_{n+1}^i}(y,z')|&\leq 	\sum_{w,w'}|G_{A}(y,w)\Gamma_{w,w'}G_{B_{n+1}^i}(w',z')|\\
	&\leq C(d)e^{-\gamma_n\|y-w\|_1}\|G_{B_{n+1}^i}\| \\
	&\leq C(d)	e^{-\gamma_n(\|z'-y\|_1-l_{n+1}^\frac{3}{4})}\delta_{n+1}^{-1},
\end{align*}
 and  $\delta_{n+1}^{-1}=e^{l_{n+1}^\frac{2}{3}}\ll e^{\gamma_{n}\|x-y\|_1}$ to bound the second term. 
 
 Second,  we establish the upper bound on the norm of Green's functions restricted to general   $(n+1)$-good set.  Now  assume $\Lambda$ is an arbitrary $(n+1)$-good set.  Thus, all the $(n+1)$- stage blocks $B_{n+1}^i$ inside $\Lambda$ are $(n+1)$-regular. We must show that $G_\Lambda$ exists. By the Schur's test Lemma, it suffices to show 
	\begin{equation}\label{Schurn+1}
		\sup_x\sum_{y}|G_\Lambda(\theta;E+i0)(x,y)|<C<\infty.
	\end{equation}
	Denote $P'_{n+1}=\{c_{n+1}^i\in P_{n+1}:\ B_{n+1}^i\subset\Lambda\} \text{ and } \Lambda'=\Lambda\setminus(\cup_{c_{n+1}^i\in P'_{n+1}} I_{n+1}^i)$. Since $\Lambda$ is $(n+1)$-good, one can check that $\Lambda'$ is $n$-good.  For $x\in \Lambda\setminus(\cup_{c_{n+1}^i\in P'_{n+1}} 2I_{n+1}^i)$ ($2I_{n+1}^i$ is a $2l_{n+1}^\frac{2}{3}$-size block  centered at $c_{n+1}^i$), we have 
	\begin{align*}
		\sum_y|G_\Lambda(x,y)|&\leq \sum_y|G_{\Lambda'}(x,y)|+\sum_{z,z',y}|G_{\Lambda'}(x,z)\Gamma_{z,z'}G_{\Lambda}(z',y)|\\
		&\leq	C(d)\delta_n^{-2}+	C(d)e^{-l_{n+1}^\frac{2}{3}}\sup_{z'}\sum_y|G_{\Lambda}(z',y)|.
	\end{align*}
	For $x\in 2I_{n+1}^i$, we have 
	\begin{align*}
		\sum_y|G_\Lambda(x,y)|&\leq \sum_y|G_{B_{n+1}^i}(x,y)|+\sum_{z,z',y}|G_{B_{n+1}^i}(x,z)\Gamma_{z,z'}G_{\Lambda}(z',y)|\\
		&\leq \delta_{n+1}^{-2}+C(d)e^{-\frac{1}{2}l_{n+1}}\sup_{z'}\sum_y|G_{\Lambda}(z',y)|.
	\end{align*}
	By taking supremum in $x$, we get $$\sup_x\sum_y|G_\Lambda(x,y)|\leq \delta_{n+1}^{-2}+\frac{1}{2}\sup_x\sum_y|G_\Lambda(x,y)|,$$
	which gives \eqref{Schurn+1}.
	Thus it follows that for $|\theta-\theta^*|<\delta_{n+1}/(10M_1)$ and $|E-E^*|<\frac{2}{5}\delta_{n+1}$, $G_\Lambda(\theta;E)$ exists, from which we get  $\operatorname{dist}(\sigma(H_\Lambda(\theta)),E^*)\geq \frac{2}{5}\delta_{n+1}$.  Hence $\operatorname{dist}(\sigma(H_\Lambda(\theta)),E)\geq \frac{1}{5}\delta_{n+1}$ for $|E-E^*|<\frac{1}{5}\delta_{n+1}$, giving  the desired  bound 
	$$\|G_\Lambda(\theta;E)\|=\frac{1}{\operatorname{dist}(\sigma(H_\Lambda(\theta)),E)}\leq10\delta_{n+1}^{-1}.$$

	Finally,  we use the above upper bound on norms of Green's functions  and iteration of the resolvent identity  to prove the off-diagonal decay of Green's function. 
	Let $x,y\in \Lambda$ such that $\|x-y\|\geq l_{n+1}^\frac{5}{6}$. We define 
	\[B_x=\left\{\begin{aligned}
		&\Lambda_{\l_1^\frac{1}{2}}(x)\cap\Lambda  \quad \text{if }   x\in \Lambda\setminus\cup_{c_1^i\in P'_1} 2I_1^i,  \\
		&B_m^i\ \quad \text{if }  x \in2I_m^i,\text{ $m\leq n+1$ is the first stage such that $c_m^i\notin Q_m$}. \ 
	\end{aligned}\right. \]
	The set $B_x$ has the following two properties: \textbf{(1).} $B_x$ is $m$-good for some $0\leq m\leq n+1$; \textbf{(2).} The $x$ is close to the center of $B_x$ and away from its relative boundary with $\Lambda$. We can iterate the resolvent identity  to obtain 
	\begin{align}
		|G_\Lambda(x,y)|&\leq\prod_{s=0}^{L-1} (C(d) l_{m_s}^d e^{-\gamma_{m_s-1}'\|x_{s}-x_{s+1}\|_1})|G_\Lambda(x_L,y)|\nonumber \\
		&\leq e^{-\gamma_n''\|x-x_L\|_1}|G_\Lambda(x_L,y)|, \label{728n}
	\end{align}
	where $x_0:=x$, $B_{x_s}$ is a  $0$-good set  or  a regular  block of stage $m_s$ and   $x_{s+1}\in \partial B_{x_{s}}.$ Thus $\|x_{s}-x_{s+1}\|_1\geq \frac{1}{2}l_{m_s}$. We stop the iteration until  $y\in B_{x_L}$. 
	Using the resolvent identity again, we get 
	\begin{align}
		|G_\Lambda(x_L,y)|&\leq|G_{B_{x_L}}(x_L,y)|+\sum_{z,z'}|G_{B_{x_L}}(x_L,z)\Gamma_{z,z'}G_{\Lambda}(z',y)|\nonumber\\
		&\leq C(d) e^{-\gamma_n'(\|x_L-y\|_1-l_{n+1}^\frac{4}{5})}\delta_{n+1}^{-1},\label{728.n}
	\end{align}
	where we have used the exponential  off-diagonal decay of $G_{B_{x_L}}$ and the  estimate  $\|G_{\Lambda}\|\leq10\delta_{n+1}^{-1}$. So combining \eqref{728n} and \eqref{728.n} gives the desired off-diagonal estimate
	$$|G_\Lambda(x,y)|\leq e^{-\gamma_{n+1}\|x-y\|_1}$$
	with $\gamma_{n+1}=(1-O(l_{n+1}^{-\frac{1}{30}}))\gamma_n$. 
\end{proof}
\section{Arithmetic  version of  Anderson localization}
In this section, we will finish the proof of Theorem \ref{AL} by using Green's function estimates. 
\begin{proof}[Proof of Theorem \ref{AL}]\label{sec3}
Let $\varepsilon_0$ be small enough such that Theorem \ref{mainthm} holds true.  Fix $\theta^*\not\in\Theta$. Let $E^*$ be a generalized eigenvalue of $H(\theta^*)$ and $\psi\neq0$ be the corresponding generalized eigenfunction satisfying $|\psi(x)|\leq(1+\|x\|_1)^d.$ 	 From Schnol's theorem, it suffices to show $\psi$  decays exponentially. For this purpose, note first  there exists (since $\theta^*\notin \Theta$) some $n_1\geq1$ such that
	\begin{equation}\label{xiajie.}
		\left\| 2 \theta^*+x\cdot \omega \right\|>\|x\|_1^{-d-2} 
	\end{equation}
for  $x\in \Z^d$ satisfying $\|x\|_1 \geq l_{n_1}.$	
We claim that there exists some $n_2\geq1$ such that for all $n \geq n_2$, 
\begin{equation}\label{2148}
	\Lambda_{100l_{n}}\bigcap\left( \bigcup_{c_n^i\in Q_n} B_n^i\right)\neq\emptyset.
\end{equation}
Otherwise,  there exists a subsequence  $n_r\to +\infty$ such that 
\begin{equation}\label{2114}
\Lambda_{100l_{n_r}}\bigcap\left( \bigcup_{c_{n_r}^i\in Q_{n_r}} B_{n_r}^i\right)=\emptyset.
\end{equation}
By the result of Appendix \ref{AC}, there exists $\Lambda\subset \Z^d $ such that 
\begin{equation}\label{2120}
\Lambda_{50l_{n_r}}\subset\Lambda\subset \Lambda_{100l_{n_r}}\text{ and }	B_m^i\cap \Lambda\neq\emptyset\Rightarrow B_m^i\subset \Lambda\ \text{ for $1\leq m\leq n_r$} .
\end{equation}
Let  $G_\Lambda=G_\Lambda(\theta^*;E^*)=(H_\Lambda(\theta^*)-E^*)^{-1}$.  From \eqref{2114} and \eqref{2120}, we deduce that $\Lambda$ is $n_r$-good. Thus if $\|x\|_1\leq l_{n_r} $, 
we have $$|\psi(x)|\leq\sum_{z,z'}|G_\Lambda(x,z)\Gamma_{z,z'} \psi(z')|\leq C(d)l_{n_r}^{2d}e^{-\frac{1}{2}\gamma_0l_{n_r}},$$
where we use  $\operatorname{dist}(x,\partial\Lambda)>l_{n_r}$ and the exponential off-diagonal decay of  $G_\Lambda$.
Taking $r\to \infty$ yields $\psi=0$, which contradicts the assumption $\psi\neq0$. Hence we prove the claim.
Recalling again Appendix \ref{AC}, there exists $X_n$  such that,
  $$\Lambda_{4l_{n+2}}\setminus\Lambda_{l_{n+1}}\subset X_n\subset \Lambda_{4l_{n+2}+50l_n}\setminus\Lambda_{l_{n+1}-50l_n}$$
  and 
$$B_m^i\cap A_n\neq\emptyset\Rightarrow B_m^i\subset A_n\ \text{ for $1\leq m\leq n$} .$$
Denote  $Y_n=\Lambda_{3l_{n+2}}\setminus\Lambda_{2l_{n+1}}$.  Then  for $\|x\|>\max(2l_{n_1+1},2l_{n_2+1})$, there exists $n\geq\max(n_1,n_2)$ such that $x\in Y_n$. Recall that \eqref{2148} holds for this $n$, i.e.,  $B_n^i\cap\Lambda_{100l_{n}}\neq\emptyset$ for some $c_n^i\in Q_n$.  So if there exists some $B_n^j\ (c_n^j\in Q_n)$ such that $B_n^j\subset A_n$,  then the {\bf Center Theorem}  shows 
\begin{equation}\label{centeral}
	m(c_n^i,c_n^j):=\min(\|(c_n^i-c_n^j)\cdot \omega\|,\|2\theta^*+(c_n^i+c_n^j)\cdot \omega\|)\leq2\delta_n^{1/2}.
\end{equation}
We will  prove that \eqref{centeral} contradicts  \eqref{xiajie.}. By the Diophantine  condition of $\omega$, we have 
$$ \|(c_n^i-c_n^j)\cdot \omega\|\geq \frac{\gamma}{\|c_n^i-c_n^j\|_1^\tau}\geq\frac{\gamma}{(5l_{n+2})^\tau}>2\delta_n^{1/2}.$$
Thus, if \eqref{centeral} holds, we must have
 \begin{equation}\label{mao}
	\|2\theta^*+(c_n^i+c_n^j)\cdot \omega\|\leq2\delta_n^{1/2}.
\end{equation}
We note that $\|c_n^i+c_n^j\|_1\geq\|c_n^j\|_1-\|c_n^i\|_1\geq l_{n+1}-200l_n>l_n$. Thus \eqref{xiajie.} gives 
 $$	\left\| 2 \theta^*+(c_n^i+c_n^j)\cdot \omega \right\|\geq\|c_n^i+c_n^j\|_1^{-d-2}\geq (5l_{n+1})^{-d-2},$$
 which contradicts  \eqref{mao}. So,  there is no singular block $B_n^j$ contained in  $ X_n$, namely,  $X_n$ is $n$-good and the Green's function estimates hold true.
 Recalling $x\in Y_n$, one has  $\operatorname{dist}(x,X_n)\geq \|x\|_1/5 \geq l_n$. Thus, we obtain 
 \begin{align}\label{shuai}
 	\begin{split}
	|\psi(x)|&\leq\sum_{z,z'}|G_{X_n}(x,z)\Gamma_{z,z'} \psi(z')|\\
&\leq C(d)l_{n+2}^{2d}e^{-\frac{1}{10}\gamma_0\|x\|_1}.\\
&\leq e^{-\frac{1}{20}\gamma_0\|x\|_1},
 	\end{split}
 \end{align}
 which proves  the exponential decay of $|\psi(x)|$ for  $\|x\|_1>\max(2l_{n_1+1},2l_{n_2+1})$.
 \end{proof}

\section{Dynamical localization}\label{sec4}
In this section, we will prove Theorem \ref{t1} about the dynamical localization. 
\begin{proof}[Proof of Theorem \ref{t1}]
Let $\varepsilon_0$ be small enough such that Theorem \ref{mainthm} holds true. 	Since Anderson localization  holds  for $\theta\in \Theta_A$ by  Theorem \ref{AL}, let $\{\varphi_\alpha,E_\alpha\}_{\alpha\in\N}$ denote a complete set of eigenstates and corresponding eigenvalues of $H(\theta)$. For simplicity, we omit the dependence of $H(\theta)$ on $\theta$. Then 
	$${\bm e}_0=\sum_\alpha\varphi_\alpha(0)\varphi_\alpha$$
	and hence 
	$$e^{itH}{\bm e}_0=\sum_\alpha e^{itE_\alpha}\varphi_\alpha(0)\varphi_\alpha.$$
Thus, it is sufficient to estimate
	\begin{equation}\label{3}
		\sum_\alpha\left(\sum_x (1+\|x\|_1)^q |\varphi_\alpha(x)|\right) |\varphi_\alpha(0)|.
	\end{equation}
	Let $I_0=\emptyset$ and $I_j=\{\alpha:\ |\varphi_\alpha(0)|>e^{-\gamma_0l_j}\}\  (j\geq1).$
	Then \begin{equation}\label{kuai}
		\eqref{3}=\sum_{j=1}^{+\infty}\sum_{\alpha\in I_j\setminus I_{j-1}}\left(\sum_x (1+\|x\|_1)^q |\varphi_\alpha(x)|\right) |\varphi_\alpha(0)|.
	\end{equation}
We claim that 	for $\alpha\in I_j$ and $n\geq j$,  
	\begin{equation}\label{jiao}
		\Lambda_{100l_n}\cap\left(\bigcup_{c_n^i\in Q_n}B_n^i\right)\neq \emptyset.
	\end{equation}
Otherwise, there exists some $n$-good set $\Lambda$ such that $\Lambda_{50l_n}\subset\Lambda\subset\Lambda_{100l_n}$. Then we get a contradiction of 
 \begin{align*}
 	|\varphi_\alpha(0)|&\leq \sum_{(z,z')\in \partial \Lambda} | G_{\Lambda}(0, z) \varphi_\alpha(z')|<e^{-\gamma_0l_n}\leq e^{-\gamma_0l_j}.
 \end{align*}
Assume \begin{equation}\label{AAsu}
	\delta_m^{1/4}<A\leq\delta_{m-1}^{1/4}\ (\delta_{-1}:=+\infty).
\end{equation} Then by  \eqref{D1} and $\omega\in{\rm DC}_{\tau,\gamma}$,  we have  for $n\geq m$,    $x\in \Lambda_{100l_n}$ and $ x'\in \Lambda_{4l_{n+2}+50l_n}\setminus\Lambda_{l_{n+1}-50l_n}$,
\begin{align}\label{yuansu}
	\nonumber m(x,x')&=\min(\|(x-x')\cdot \omega\|,\|(x+x')\cdot \omega+2\theta\|)\\
	&\geq\min\left(\frac{\gamma}{(5l_{n+2})^\tau},\frac{A}{(5l_{n+2})^{d+1}}\right)>2\delta_n^{1/2}.
\end{align}  If $\alpha \in I_j$, then \eqref{jiao} holds for $n\geq j$.  Thus by \eqref{yuansu} and the {\bf Center Theorem}, for $n\geq\max(m,j)$, there is no singular block of  the $n$-th generation  inside  $\Lambda_{4l_{n+2}}\setminus\Lambda_{l_{n+1}}$, which  proves   
$	|\varphi_\alpha(x)|\leq e^{-\frac{1}{20}\gamma_0\|x\|_1}$ for $\|x\|_1\geq\max(2l_{m+1},2l_{j+1})$ (the proof is the same as that of \eqref{shuai}).
 From the Hilbert-Schmidt argument, we have  \begin{align*}
	C(d)l_{\max(m,j)+1}^d&\geq\sum_{\|x\|_1\leq2l_{\max(m,j)+1}}\sum_\alpha|\varphi_\alpha(x)|^2\\
	&\geq\sum_{\alpha\in I_j}\sum_{\|x\|_1\leq2l_{\max(j,m)+1}}|\varphi_\alpha(x)|^2\\
	&=\#I_j\sum_{\alpha\in I_j}\sum_{\|x\|_1>2N_{\max(j,m)+1}}|\varphi_\alpha(x)|^2\\
	&\geq\frac{1}{2}\#I_j.
\end{align*}
Thus   $\#I_j\leq C(d)l_{\max(j,m)+1}^{d}$. 

To estimate \eqref{kuai}, using 	$|\varphi_\alpha(x)|\leq e^{-\frac{1}{20}\gamma_0\|x\|_1}$ for $\alpha \in I_m$ and  $\|x\|_1\geq2l_{m+1}$, we get 
\begin{align}\label{jizhong}
\nonumber&\ \ \ \sum_{j=1}^{m}\sum_{\alpha\in I_j\setminus I_{j-1}}\left(\sum_x(1+\|x\|_1)^q |\varphi_\alpha(x)|\right) |\varphi_\alpha(0)|\\
\nonumber&\leq\sum_{\alpha\in I_m}\left(\sum_x(1+\|x\|_1)^q |\varphi_\alpha(x)|\right)\nonumber\\
	&\leq\#I_m\sup_{\alpha\in I_m}\left(\sum_{\|x\|_1\leq2l_{m+1}}+\sum_{\|x\|_1>2l_{m+1}}\right)(1+\|x\|_1)^q |\varphi_\alpha(x)|\nonumber\\
	&\leq C({q,d})l_{m+1}^{q+2d	}.
\end{align}
Using 	$|\varphi_\alpha(x)|\leq e^{-\frac{1}{20}\gamma_0\|x\|_1}$ for $ j\geq m,\alpha \in I_j$ and $\|x\|_1\geq2l_{j+1}$, we get 
\begin{align*}
&\ \ \ \sum_{\alpha\in I_j\setminus I_{j-1}}\left(\sum_x (1+\|x\|_1)^q |\varphi_\alpha(x)|\right) |\varphi_\alpha(0)|\\
	&\leq\#I_j\sup_{\alpha\in I_j}\left(\sum_{\|x\|_1\leq2l_{j+1}}+\sum_{\|x\|_1>2l_{j+1}}\right)e^{-\gamma_0l_{j-1}}\\
	&\leq C_{q,d}l_{j+1}^{q+2d	}e^{-\gamma_0l_{j-1}},
\end{align*}
where $l_0:=0$. Summing  up $j$ for $j\geq m+1$ gives  
\begin{align}\label{zai}
		\nonumber &\ \ \ \sum_{j=m+1}^{\infty}\sum_{\alpha\in I_j\setminus I_{j-1}}\left(\sum_x (1+\|x\|_1)^q |\varphi_\alpha(x)|\right) |\varphi_\alpha(0)|\\
		&\leq\left\{\begin{aligned}
	& C({q,d})e^{-\frac{\gamma_0}{2}l_m} &\text{if $m\geq1$,}\\
	&C({q,d})l_2^{q+2d} &\text{if $m=0$.}
\end{aligned}\right. 
\end{align}
From \eqref{jizhong} and \eqref{zai}, we obtain 
\begin{align*}
	\eqref{kuai}&\leq C({q,d})\max(l_{m+1}^{q+2d},l_2^{q+2d})\\
	&\leq C({q,d})\max(|\log \max(A,1)|^{12(q+2d)},|\log\varepsilon_0|^{12(q+2d)}),
\end{align*}
where we use  \eqref{AAsu}  (i.e.,  $A\leq\delta_{m-1}^{1/4}$), which implies  $|\log\delta_{m-1}|\leq4|\log A|$ and $l_{n+2}\leq l_n^8$. 

Hence, we finish the proof of the dynamical localization. It remains to prove the strong dynamical localization. For this, recalling \eqref{D1}, then  taking integration leads  to 
		\begin{align*}
			&\ \ \ \left(\int_{\Theta_{\delta_0}}	+\sum_{n=1}^{+\infty}	\int_{\Theta_{\delta_n}\setminus\Theta_{\delta_{n-1}}}\right)\sup_{t\in \mathbb{R}}\sum_{x\in \mathbb{Z}^d}(1+\|x\|_1)^q|\langle e^{itH(\theta)}{\bm e}_0, {\bm e}_x\rangle| d\theta\\
			&\leq C({q,d})\left(|\log\varepsilon_0|^{12(q+2d)}+\sum_{n=1}^{+\infty}|\log\delta_n|^{12(q+2d)}\delta_{n-1}\right)\\
			&<+\infty,
		\end{align*}
	which concludes proof. 
\end{proof}
\section{H\"older continuity of the IDS}\label{sec5}

In this section, we prove Theorem \ref{thm2}.

\begin{proof}[Proof of Theorem \ref{thm2}]
	Let $\varepsilon_0$ be small enough such that Theorem \ref{mainthm} holds true. Fix $\theta^* \in \T ,E^*\in \R$ and $\eta>0$. We are going to estimate  the number of eigenvalues of  $H_\Lambda(\theta^*)$ belonging to $[E^*-\eta,E^*+\eta].$
 For this purpose, we first introduce a useful lemma which  connects the  $L^2$ bound of Green's function  with the numbers of eigenvalues of  the self-adjoint operator inside a certain interval $[E^*-\eta,E^*+\eta].$
 \begin{lem}\label{IDSL}
 	Let $H$ be a  self-adjoint operator on $\Z^{d}$  and $\Lambda\subset \Z^d$  be a finite set. Assume there exists some  $\Lambda'\subset\Z^d$ such that $\#(\Lambda\setminus\Lambda')+\#(\Lambda'\setminus\Lambda)\leq M$ and $\|G_{\Lambda'}(E^*)\|\leq(2\eta)^{-1}$,  where $G_{\Lambda'}(E^*)=(H_{\Lambda'}-E^*)^{-1}$. Then the number of eigenvalues of $H_\Lambda$  inside $[E^*-\eta,E^*+\eta]$ is at most $3M$. 
  \end{lem}
 	\begin{proof}
 	Denote $T=H-E^*$.	Let $\{\xi_l\}_{l=1}^L$ be the orthonormal eigenfunctions of $H_\Lambda$ with corresponding eigenvalues belonging to $[E^*-\eta,E^*+\eta]$. Then for every $\xi\in \{\xi_l\}_{l=1}^L$, we have $\|T_{\Lambda}\xi\|\leq\eta$  and then 
 	\begin{align*}
 		\eta\geq\|R_{\Lambda\cap\Lambda'}TR_\Lambda\xi\|&=\|(R_{\Lambda'}-R_{\Lambda'\setminus\Lambda})TR_\Lambda\xi\|\\
 		&=\|R_{\Lambda'}TR_{\Lambda'}\xi+R_{\Lambda'}TR_{\Lambda\setminus\Lambda'}\xi-R_{\Lambda'\setminus\Lambda}TR_\Lambda\xi\|.
 	\end{align*}
 Using $\|G_{\Lambda'}(E^*)\|\leq(2\eta)^{-1}$, we obtain
\begin{equation}\label{752}
	 \|R_{\Lambda'}\xi+G_{\Lambda'}(E^*)(R_{\Lambda'}TR_{\Lambda\setminus\Lambda'}\xi-R_{\Lambda'\setminus\Lambda}TR_\Lambda)\xi\|\leq1/2.
\end{equation}
 	Denote $\mathcal{H}=\operatorname{Range}G_{\Lambda'}(E^*)(R_{\Lambda'}TR_{\Lambda\setminus\Lambda'}\xi-R_{\Lambda'\setminus\Lambda}TR_\Lambda)$. Thus \begin{equation}\label{809}
 		\operatorname{dim}\mathcal{H}\leq \operatorname{Rank}R_{\Lambda\setminus\Lambda'}+\operatorname{Rank}R_{\Lambda'\setminus\Lambda}\leq M.
 	\end{equation}
 	From \eqref{752}, we deduce 
\begin{equation}\label{808}
	 \|R_{\Lambda'}\xi\|^2-\|P_\mathcal{H}R_{\Lambda'}\xi\|^2=\|P_\mathcal{H}^\perp R_{\Lambda'}\xi\|^2\leq1/4.
\end{equation}
 	Hence, from \eqref{809} and \eqref{808}, we get 
 	\begin{align*}
 		L=\sum_{l=1}^L\|\xi_l\|^2	&=\sum_{l=1}^L\|R_{\Lambda'}\xi_l\|^2+\sum_{l=1}^L\|R_{\Lambda\setminus\Lambda'}\xi_l\|^2\\
 		&\leq L/4+\sum_{l=1}^L\|P_\mathcal{H}R_{\Lambda'}\xi_l\|^2+\sum_{l=1}^L\|R_{\Lambda\setminus\Lambda'}\xi_l\|^2\\
 		&\leq L/4+\operatorname{dim}\mathcal{H}+\#(\Lambda\setminus\Lambda')\\
 		&\leq L/4+2M,
 	\end{align*}
which concludes the proof.
	  	\end{proof}
	Now, let  $N$ be  sufficiently large depending on $\eta$.	For $\eta>\delta_0^{20}=\varepsilon_0$, define $$Q_\eta=\{x\in \Lambda_N:\ |v(\theta^*+x\cdot \omega)-E^*|\leq(2d+2) \eta\}.$$ 
Thus for any $x,x'\in Q_\eta$, we have  since  \eqref{wh}
$$\min(\|(x-x')\cdot \omega\|, \|2\theta^*+(x+x')\cdot \omega)\|) \leq C(d)\eta^{1/2}.$$
From the uniform distribution of $\{x\cdot \omega\}_{x\in \Z^d}$, we deduce that  $\#Q_\eta\leq C(d)\eta^{1/2}\#\Lambda_N$. Denote $\Lambda'=\Lambda_N\setminus Q_\eta$. Then 
 $$\|G_{\Lambda'}(E^*;\theta^*)\|\leq((2d+2)\eta-2d\varepsilon)^{-1}\leq(2\eta)^{-1}.$$
By Lemma \ref{IDSL}, $H_{\Lambda_N}(\theta^*)$ has at most $C(d)\eta^{1/2}\#\Lambda_N$ eigenvalues in $[E^*-\eta,E^*+\eta]$. Thus  $$\mathcal{N}_{\Lambda_N}(E^*+\eta;\theta^*)-\mathcal{N}_{\Lambda_N}(E^*-\eta;\theta^*)\leq C(d)\eta^{1/2}.$$

Next,  we consider the case when  $\delta_s^{20}\geq\eta\geq\delta_{s+1}^{20}$ for some $s\geq0$. By result of Appendix \ref{AC}, we can find  $\tilde{\Lambda}$ such that $\Lambda_N\subset\tilde{\Lambda}\subset\Lambda_{N+50l_{n+1}}$ and 
$$B_m^i\cap \tilde{\Lambda}\neq\emptyset\Rightarrow B_m^i\subset \tilde{\Lambda}\ \text{ for $1\leq m\leq n+1$} .$$
Define 
 $$P'_{n+1}=\{c_{n+1}^i\in P_{n+1}:\ B_{n+1}^i\subset\tilde{\Lambda}\}$$ and 
$$Q_\eta=\{k_{n+1}^i\in P_{n+1}':\   \operatorname{dist}(\sigma(H_{B_{n+1}^i}(\theta^*)),E^*) <20\eta\}.$$ 
Replacing $\delta_{n+1}$ with $\eta$ in the proof  of {\bf Center Theorem} from stage $n$ to stage $n+1$ (where we only use the relation $|\log\delta_{n+1}|\geq20|\log\delta_{n}|$), we get 
for any $x,x'\in Q_\eta$,
 $$m(x,x'):=\min(\|(x-x')\cdot \omega)\|, \|2\theta^*+(x+x')\cdot \omega)\|) \leq 2(20\eta)^{1/2}<20\eta^{1/2}.$$
Let $\Lambda'=\tilde{\Lambda}\setminus(\bigcup_{k_{n+1}^i\in Q_\eta}B_{n+1}^i)$. Replacing  $\delta_{n+1}$ with $\eta$ and similar to the proof of Theorem \ref{ng} (since we only use the relations $\delta_{n+1}<\delta_n/10$ and $|\log\delta_{n+1}|\lesssim l_{n+1}^{2/3}$ in the proof), we obtain
  \begin{equation}\label{1241}
		\| G_{\Lambda'}(\theta^*; E^*)\|\leq10(20\eta)^{-1}=(2\eta)^{-1}.
\end{equation}
 Notice that \begin{align}\label{1242}
		\nonumber&\ \ \ \#(\Lambda_N\setminus\Lambda')+\#(\Lambda'\setminus\Lambda_N)\\
	\nonumber&\leq\#(\tilde{\Lambda}\setminus\Lambda')+\#(\tilde{\Lambda}\setminus\Lambda_N)\\
	&\leq C(d)(l_{n+1}^d\eta^{1/2}\#\Lambda_N+l_{n+1}N^{d-1}).
\end{align}
Combining \eqref{1241}, \eqref{1242} and Lemma \ref{IDSL} gives  
\begin{align*}
\mathcal{N}_{\Lambda_N}(E^*+\eta;\theta^*)-\mathcal{N}_{\Lambda_N}(E^*-\eta;\theta^*)&\leq C(d)(l_{n+1}^d\eta^{1/2}+l_{n+1}/N)\\ 
&\leq C(d)\eta^{\frac{1}{2}}|\log\eta|^{8d}
\end{align*}
provided $N\gg1$,  where we use  $l_{n+1}\leq l_n^4\leq|\log \delta_n|^8\leq|\log\eta|^8$.

Finally, combining the above two cases leads to the desired proof. 
\end{proof}

\appendix{}
\section{}\label{AD}
\begin{lem}[Trial wave function] Let $H$ be a self-adjoint operator on a finite dimensional   Hilbert space  and $E^*\in \R$. If there exist $m$ orthonormal functions $\psi_k\ (1\leq k\leq m )$ such that  $\|(H-E^*)\psi_k\|\leq\delta$ for some $\delta>0$ and all $1\leq k\leq m$, then $H$ has $m$ eigenvalues $ E_{k}\ (1\leq k\leq m )$ counted in multiplicities  satisfying $\sum_{k=1}^{m}(E_k-E^*)^2\leq m\delta^2$. These $\psi_k$ are  called  trial functions.
\end{lem}
\begin{proof}
	Without loss of  generality, we may assume $E^*=0$. It suffices to show that the first $m$ eigenvalues of the  positive semidefinite operator   $H^2$,   $0\leq\lambda_1\leq\cdots\leq\lambda_m$ satisfy
$$	\sum_{k=1}^m \lambda_k\leq m\delta^2.$$
	 Denote by $P$ the orthogonal projection on the space spanned by $\psi_k\ (1\leq k\leq m)$. Thus, the restricted operator $PH^2P$ has $m$  eigenvalues $0\leq\mu_1\leq\cdots\leq\mu_m$ satisfying $\lambda_k\leq\mu_k$ by the min-max principle. Thus, we obtain \begin{align*}
	 		\sum_{k=1}^m \lambda_k\leq\sum_{k=1}^m \mu_k&=\operatorname{Trace}(PH^2P)\\
	 		&\leq \sum_{k=1}^{m}\langle\psi_k,PH^2P\psi_k\rangle\\
	 		&=\sum_{k=1}^m\|H\psi_k\|^2\\
	 		&\leq m\delta^2,
	 \end{align*} which  finishes  the proof.
 \end{proof}
This lemma immediately gives us 	   \begin{cor}\label{trialcor}
	   	If there  exists a trial function such that $\|\psi\|=1$ and   $\|(H-E^*)\psi\|\leq\delta$, then $H$ has at least one eigenvalue in $|E-E^*|\leq\delta.$ If there exist two orthogonal trial functions such that $\|\psi_1\|=\|\psi_2\|=1$, $\|(H-E^*)\psi_1\|\leq\delta$ and $
	   	\|(H-E^*)\psi_2\|\leq\delta$, then $H$ has at least two eigenvalues in $|E-E^*|\leq\sqrt{2}\delta$.
	   \end{cor}

\section{}\label{AA}
\begin{lem}[Morse]\label{C2}
	Let $E(\theta)$ be a $C^2$  function defined on $[a,b]$. Suppose that there is a point $\theta_s$ in the interval such that
	$E\left(\theta_s- \theta\right)=E\left(\theta_s+ \theta\right)$ for all $\theta$. We also assume that there exists $\delta>0$ such that,    $|E'(\theta)| \leq \delta$ implies $|E''(\theta)| \geq 2$ with  a unique sign for these $\theta$. Then
	\begin{align*}
	\left|E\left(\theta_2\right)-E\left(\theta_1\right)\right|& \geq \frac{1}{2}M^2(\theta_1,\theta_2)\\
	&:= \frac{1}{2}\min (\left|\theta_2-\theta_1\right|^2, \left|\theta_2+\theta_1-2 \theta_s\right|^2)
	\end{align*}
	provided $M(\theta_1,\theta_2)\leq\delta$.
	Moreover,
	$$
	|E'(\theta)| \geq \min 
	(\delta,|\theta-\theta_s|) .
	$$\end{lem}
\begin{proof}
	The proof is similar to that in \cite{Sur90} (cf.  Appendix A). Without loss of  generality, we may consider the case $\left|E'\right| \leq \delta$ implies $E^{\prime \prime} \geq 2$. By the symmetry,  we must have $E'\left(\theta_s\right)=0$; therefore $E^{\prime \prime}\left(\theta_s\right) \geq 2$. Let $\theta_d$ be the largest number satisfying 
	$$
	E^{\prime \prime}(\theta) \geq 2 \text { for } \theta_s \leq \theta \leq \theta_d .
	$$
	This implies that $E(\theta)$ is an increasing function to the right of the symmetry point. By the definition of $\theta_d$, we have $E^{\prime \prime}\left(\theta_d+\Delta \theta_n\right)<2$ for a sequence $\Delta \theta_n\to 0^+$.  Therefore $E'\left(\theta_d+\Delta \theta_n\right) >\delta$. This inequality must hold for every $\theta>\theta_d$.  Otherwise,  we would have a point $\theta>\theta_d$,  where $E'(\theta)=\delta,E'(\theta-\Delta\theta)>\delta$ for small $\Delta\theta>0$,  but   $E^{\prime \prime}(\theta) \geq 2>0$ by  $|E'(\theta)|\leq\delta$ and the assumption of $E$. This is impossible. Therefore
	$$
	E'(\theta) \geq \delta \text { for } \theta>\theta_d .
	$$
	So we have the following cases:\smallskip\\
	{\it Case} 1. $\theta_s\leq\theta_1<\theta_2\leq\theta_d$.  
	\begin{align*}
		E(\theta_2)-E(\theta_1)= E'(\theta_1) (\theta_2-\theta_1)+\frac{1}{2}E''(\xi)(\theta_2-\theta_1)^2\geq(\theta_2-\theta_1)^2.
	\end{align*}
	{\it Case }2. $\theta_d\leq\theta_1<\theta_2$.
	\begin{align*}
		E(\theta_2)-E(\theta_1)= E'(\xi) (\theta_2-\theta_1)\geq \delta(\theta_2-\theta_1)\geq (\theta_2-\theta_1)^2.
	\end{align*}
	{\it Case} 3. $\theta_s\leq\theta_1\leq\theta_d\leq\theta_2$. 
	\begin{align*}
		E(\theta_2)-E(\theta_1)&= E(\theta_2)-E(\theta_d)+E(\theta_d)-E(\theta_1) \\
		&\geq (\theta_2-\theta_d)^2+(\theta_d-\theta_1)^2\\
		&\geq\frac{1}{2}(\theta_2-\theta_1)^2.
	\end{align*}
	{\it Case} 4. $\theta_1\leq\theta_s\leq\theta_2$. Then we have $2\theta_s-\theta_1\geq\theta_s$. By Case $1$-$3$, we get  
	\begin{align*}
		|E(\theta_2)-E(\theta_1)|=|E(\theta_2)-E(2\theta_s-\theta_1)|\geq\frac{1}{2}(\theta_1+\theta_2-2\theta_s)^2
		.\end{align*}
	To prove the second inequality, we  consider the  two cases. \\
	{\it Case} 1. $\theta_s \leq \theta \leq \theta_d$.
	$$
	E'(\theta)=E'(\theta_s)+E''(\xi)(\theta-\theta_s) \geq \theta-\theta_s .
	$$
	{\it Case} 2. $\theta_d \leq \theta$. In this case,  we have $E'(\theta) \geq \delta$. \\
	For $\theta\leq\theta_s$, we use the symmetry property  of $E$ about $\theta_s$. 
	
	Hence we finish the proof. 
\end{proof} 

\section{}\label{AB}
\begin{thm}\label{daoshu}
	Let  $H(\theta)$ be a family of  finite dimensional self-adjoint operators with $C^2$ parametrization. Assume that $E(\theta^*)$ is a simple eigenvalue of $H(\theta^*)$ and $\psi(\theta^*)$ is its corresponding   eigenfunction. Then by Lemma \ref{neqs}, $E(\theta),\psi(\theta)$ can be $C^2$ parameterized in a neighborhood of $\theta^*$. Moreover, for $\theta$ belonging to this  neighborhood, we have  
	\begin{itemize}
		\item[\textbf{(1)}.] $\frac{d}{d\theta}E=\langle\psi,H'\psi\rangle$.
		\item[\textbf{(2)}.] $\frac{d^2}{d\theta^2}E=\langle\psi,H''\psi\rangle-2\langle H'\psi,G^\perp(E)H'\psi\rangle$, where $G^\perp(E)$ denotes the Green's function on the orthogonal complement of $\psi$.
		\item[\textbf{(3)}.] Let  $\mathcal{E}\neq E$ be another simple eigenvalue and $\Psi$ its   eigenfunction. Then  we have 
		$$\langle H'\psi,G^\perp(E)H'\psi\rangle=-\frac{\langle\Psi,H'\psi\rangle^2}{E-\mathcal{E}}+\langle H'\psi,G^{\perp\perp}(E)H'\psi\rangle,$$
	\end{itemize} 
where $G^{\perp\perp}(E)$ denotes the Green's function on the orthogonal complement of $\psi$ and $\Psi$.
\end{thm}
\begin{proof}
	Notice that $\langle\psi,\psi\rangle\equiv1$. So we have $\langle\psi,\psi'\rangle=0$.
	Taking derivatives  on the equation  $E= \langle\psi,H\psi\rangle$ yields 
	  \begin{equation}\label{yijie}
		\frac{d}{d\theta}E=\langle\psi,H'\psi\rangle+2\langle\psi',H\psi\rangle=\langle\psi,H'\psi\rangle,
	\end{equation} where we have used  $H\psi=E\psi$ and $\langle\psi,\psi'\rangle=0$.  This proves {\bf (1)}. Now we try to prove \textbf{(2)}.
Taking derivatives again on  \eqref{yijie} gives 
$\frac{d^2}{d\theta^2}E=\langle\psi,H''\psi\rangle+2\langle \psi',H'\psi\rangle$. Thus, it suffices to show 
$$\langle \psi',H'\psi\rangle=-\langle H'\psi,G^\perp(E)H'\psi\rangle.$$
Since $\psi'$ is orthogonal to $\psi$, we have 
\begin{align*}
	\langle \psi',H'\psi\rangle =	\langle P^\perp \psi',H'\psi\rangle=\langle G^\perp(E)(H-E)\psi',H'\psi\rangle=-\langle H'\psi,G^\perp(E)H'\psi\rangle,
\end{align*}
where we have used $G^\perp(E)\psi=0 $ and  $(H'-E')\psi=-(H-E)\psi'$ since $(H-E)\psi\equiv0$.
 Finally, the item \textbf{(3)} follows from $$G^\perp(E)=((H-E)^\perp)^{-1}=\sum_{E'\neq E}\frac{1}{E'-E}P_{E'}$$
and 
$$G^{\perp \perp}(E)=((H-E)^{\perp \perp})^{-1}=\sum_{E'\neq E,\mathcal{E}}\frac{1}{E'-E}P_{E'}$$
immediately, where $P_{E'}$ denotes the orthogonal projection on the eigenspace of $E'\in\R$.
\end{proof}
\section{}\label{AC}
\begin{thm}\label{1}
If  $s_{n}=\inf\{\|c_{n}^i-c_{n}^j\|_1:\ {c_{n}^i\neq c_{n}^j\in Q_{n} }\}\geq 10l_n^2$. Then we can associate every $c_{n+1}^i\in P_{n+1}=Q_{n}$ a block $B_{n+1}^i$ such that
\begin{itemize}
	\item[(1).] $\Lambda_{l_{n}^2}(c_{n+1}^i)	\subset B_{n+1}^i\subset	\Lambda_{l_{n}^2+50l_{n}}(c_{n+1}^i)$.
	\item[(2).] If $B_m^{j}\cap B_{n+1}^i\neq \emptyset\ (1\leq m\leq n)$, then $B_m^{j}\subset B_{n+1}^i$.
	\item[(3).] $B_{n+1}^i$ is symmetric about $c_{n+1}^i$ (i.e., $k\in B_{n+1}^i\Rightarrow 2c_{n+1}^i-k\in B_{n+1}^i$).
	\item[(4).] The set $B_{n+1}^i-c_{n+1}^i$ is independent of $i$, i.e.,  $B_{n+1}^j=B_{n+1}^i+(c_{n+1}^j-c_{n+1}^i)$.
\end{itemize} 
\end{thm}
\begin{thm}
	If $s_{n}<10l_n^2$. Then we can associate every $c_{n+1}^i\in P_{n+1}=\{c_{n+1}^i=(c_n^i+\tilde{c}_n^i)/2:\ c_n^i\in Q_n \}$ a block $B_{n+1}^i$ such that 
	\begin{itemize}
		\item[(1).] $\Lambda_{l_{n}^4}(c_{n+1}^i)	\subset B_{n+1}^i\subset	\Lambda_{l_{n}^4+50l_{n}}(c_{n+1}^i)$.
		\item[(2).] If $B_m^{j}\cap B_{n+1}^i\neq \emptyset \ (1\leq m\leq n)$, then $B_m^{j}\subset B_{n+1}^i$.
		\item[(3).] $B_{n+1}^i$ is symmetric about $c_{n+1}^i$ (i.e., $k\in B_{n+1}^i\Rightarrow 2c_{n+1}^i-k\in B_{n+1}^i$).
		\item[(4).] The set $B_{n+1}^i-c_{n+1}^i$ is independent of $i$, i.e.,   $B_{n+1}^j=B_{n+1}^i+(c_{n+1}^j-c_{n+1}^i)$.
	\end{itemize} 
\end{thm}
\begin{thm}
	For an arbitrary finite size set $\Lambda\subset \Z^d$ , there exists a  set $\tilde{\Lambda}$ such that  
		\begin{itemize}
		\item[(1).]$\Lambda\subset \tilde{\Lambda}\subset \Lambda^*$, where $\Lambda^*=\{k\in \Z^d:\ \operatorname{dist}(k,\Lambda)\leq50l_n\}.$
		\item[(2).]If $B_m^{j}\cap\tilde{\Lambda}\neq \emptyset\ (1\leq m\leq n)$, then $B_m^{j}\subset \tilde{\Lambda}$.
	\end{itemize}  
\end{thm}
We only give the proof of Theorem \ref{1}, since those of the other two theorems are similar.
\begin{proof} [Proof of Theorem \ref{1}]
In this proof, for a set $A$,  we denote $\Lambda_{L}(A)=\{k\in \Z^d:\ \operatorname{dist}(k,A)\leq L\}.$
	Before proving this theorem,  we prove a lemma concerning the set $P_r\ (1\leq r\leq n+1)$.
	\begin{lem}\label{2}
		For $c_r^i,c_r^j\in P_r$, we have 	$m(c_r^i,c_r^j):=\min(\|(c_r^i-c_r^j)\cdot \omega\|,\|2\theta^*+(c_r^i+c_r^j)\cdot \omega\|)\leq6\delta_{r-1}^{1/2}.$
	\end{lem}
\begin{proof}[Proof of Lemma \ref{2} ]
	We consider two cases. \smallskip\\
	\textbf{Case 1}.  $s_{r-1}\geq 10l_{r-1}^2$. Then $P_r=Q_{r-1}$ and the proof is completed by  the {\bf  Center Theorem}. \smallskip\\
	\textbf{Case 2}.  $s_{r-1}> 10l_{r-1}^2$. As in the proof of Lemma \ref{514}, one can show that there exists $\mu=0 \text{ or 1/2}$, such that $\|\theta^*+c_r^i\cdot \omega+\mu\|\leq3\delta_{r-1}^{1/2}$ and  $\|\theta^*+c_r^j\cdot \omega+\mu\|\leq3\delta_{r-1}^{1/2}$, which proves this lemma.
\end{proof}
Now	fix $k_0\in P_{n+1}$. We start with  $J_{0,0}=\Lambda_{l_n^2}(k_0)$.
	Denote 
	$$H_r=( k_0-P_{n+1}+P_{n-r})\cup(k_0+P_{n+1}-P_{n-r}),\  0\leq r\leq n-1.$$
	Define inductively $$J_{r,0} \subsetneqq J_{r,1}\subsetneqq \cdots \subsetneqq J_{r,t_r}:=J_{r+1,0}, $$
	where
	$$J_{r,t+1}=J_{r,t}\bigcup \left(  \bigcup_{\{ h\in H_{r}:\ \Lambda_{2l_{n-r}}(h)\cap J_{r,t}\neq\emptyset\} }\Lambda_{2l_{n-r}}(h)\right)  $$
	and $t_r$ is the  largest integer satisfying the $\subsetneqq$ relationship (the following argument shows that  $t_r<10$).
	Thus by definition, we have 
	\begin{equation}\label{??}
		h\in H_{r}, \ \Lambda_{2l_{n-r}}(h)\cap J_{r+1,0}\neq\emptyset\Rightarrow \Lambda_{2l_{n-r}}(h)\subset J_{r+1,0}.
	\end{equation}
	For $\tilde{k} \in k_0- P_{n+1}$, by Lemma \ref{2}, we have  
\begin{equation}\label{..}
	  \min\left(\|\tilde{k}\cdot  \omega\|,\|\tilde{k} \cdot \omega-2k_0\cdot \omega -2\theta^*\|\right)< 6\delta_n^{1/2}.
\end{equation}
 Choosing  a point $p\in P_{n-r}$, for convenience, we denote $\theta'=2k_0\cdot \omega+2\theta^*,\theta''=-p\cdot\omega-2\theta^*$.  From \eqref{..} and  Lemma \ref{2}, we deduce that  for any $h\in   k_0-P_{n+1}+P_{n-r}$,
  \begin{align}\label{chou}
& \min(\|(h-p)\cdot \omega\|,\|h\cdot\omega-\theta'\|,\|(h-p)\cdot \omega-\theta'\|,\|h\cdot\omega-\theta'-\theta''\|)\\
 \nonumber &\leq6\delta_{n-r-1}^{1/2}+6\delta_n^{1/2}.
  \end{align}
So \eqref{chou} says that the set $\{h\cdot \omega :\ h\in k_0-P_{n+1}+P_{n-r}\}$ must be close to one of the four fixing phases,  namely,  $\theta_i\ (i=1,2,3,4)$.  
Notice that $ k_0+P_{n+1}-P_{n-r}=2k_0-(k_0-P_{n+1}+P_{n-r})$ is symmetric to  $k_0-P_{n+1}+P_{n-r}$  about $k_0$. Thus the set $\{h\cdot \omega :\ h\in  k_0+P_{n+1}-P_{n-r}\}$ must be close to one of $\theta_{4+i}:=2k_0\cdot \omega-\theta_i\ (i=1,2,3,4)$. By the pigeonhole principle, any ten distinct elements of $H_r$ must contain two elements $h,\tilde{h}$ of them  such that $\|h\cdot \omega-\theta_i\|\leq 7\delta_{n-r-1}^{1/2}$ and $\|\tilde{h}\cdot \omega-\theta_i\|\leq 7\delta_{n-r-1}^{1/2}$ for some $1\leq i\leq 8$. Hence \begin{equation}\label{ky}
	\|(h-\tilde{h})\cdot \omega\|\leq14\delta_{n-r-1}^{1/2}.
\end{equation}
	We claim that $t_r<10$. Otherwise, there exist  distinct $h_t\in H_r\ (1\leq t\leq10)$ such that
	$$\Lambda_{2l_{n-r}}(h_{1})\cap J_{r,0}\neq\emptyset,\  \Lambda_{2l_{n-r}}(h_t)\cap\Lambda_{2l_{n-r}}(h_{t+1})\neq\emptyset. $$
	In particular, $\|h_t-h_{t+1}\|\leq4l_n$. Thus
	$\|h_t -h_{t'}\|_1\leq40l_{n-r}\text{ for all } (1\leq t,t'\leq10). $
On the other hand, by \eqref{ky}, there exist  $h_t\neq h_{t'}$ such that $\|(h_t-h_{t'})\cdot \omega\|\leq14\delta_{n-r-1}^{1/2}.$ The Diophantine condition gives $ \|h_t -h_{t'}\|_1>40l_{n-r}$. Hence we get a contradiction and prove the claim. Thus we have 
	\begin{equation}\label{?}
		J_{r+1,0}=J_{r,t_r}\subset \Lambda_{40l_{n-r}}(J_{r,0}).
	\end{equation}	
	Since $$\sum_{r=0}^{n-1}40l_{n-r}<50l_n,$$
	we find $J_{n,0}$  to satisfy
	$$\Lambda _{l_n^2}(k_0)=J_{0,0}\subset J_{n,0}\subset \Lambda_{50l_n}(J_{0,0})\subset\Lambda _{l_{n}^2+50l_n}(k_0).$$
	Next, for any $c_{n+1}^i\in P_{n+1},$ we   define
	\begin{equation}\label{wuhu}
		B_{n+1}^i=J_{n,0}+(c_{n+1}^i-k_0).
	\end{equation}
	Assume that for some $c_{n+1}^i\in P_{n+1}$ and $c_m^j\in P_m\ (1\leq m\leq n)$, 
$B_{n+1}^i\cap B_m^{j} \neq \emptyset$.
	Then \begin{equation}\label{cb}
		\left( B_{n+1}^i+(k_0-c_{n+1}^i)\right)  \cap \left( B_m^{j}+(k_0-c_{n+1}^i) \right) \neq \emptyset .
	\end{equation}
Since $	B_{n+1}^i+(k_0-c_{n+1}^i)=J_{n,0}, B_m^{j}+(k_0-c_{n+1}^i)\subset \Lambda_{l_m+50l_{m-1}}(h)\subset \Lambda_{1.5l_m}(h)$
where $h =k_0-c_{n+1}^i+c_m^i\in H_{n-m}$.
So \eqref{cb} can be restated as 
	$$J_{n,0}\cap\Lambda_{1.5l_m}(h)\neq\emptyset.$$
	Recalling \eqref{?}, we have
	$$J_{n,0}\subset \Lambda_{50l_{m-1}}(J_{n-m+1,0}).$$
	Thus $$ \Lambda_{50l_{m-1}}(J_{n-m+1,0})\cap \Lambda_{1.5l_m}(h)\neq\emptyset.   $$
	From $50l_{m-1}< 0.5l_m$, it follows that
	$$J_{n-m+1,0}\cap \Lambda_{2l_m}(h)\neq\emptyset.$$
	Recalling \eqref{??}, we deduce 
	$$\Lambda_{2l_m}(h)\subset J_{n-m+1,0}\subset J_{n,0}.$$
	Hence
	$$B_m^j \subset\Lambda_{2l_m}(c_m^j)=\Lambda_{2l_m}(h)+(c_m^j-h)\subset J_{n,0}+(c_m^j-h)=B_{n+1}^i.$$
	We will show $B_{n+1}^i-c_{n+1}^i$ is independent of $c_{n+1}^i \in P_{n+1}$. For this, recalling \eqref{wuhu}, we deduce  
	$$B_{n+1}^i-c_{n+1}^i=J_{n,0}-k_0$$ is independent of $c_{n+1}^i$.
	Finally, we prove the symmetry property of $B_{n+1}^i$. The definition of $H_r$ implies that it is symmetric about $k_0$,  which implies all   $J_{r,t}$ are symmetric  about $k_0$ as well.  In particular, $J_{n,0}$ is symmetrical about $k_0$. Using \eqref{wuhu} shows that $B_{n+1}^i$ is symmetric  about $c_{n+1}^i$. 
\end{proof}

\section*{Acknowledgments}
 Y. Shi was  partially supported by NSFC  (12271380).  Z. Zhang was  partially supported by  NSFC  (12171010, 12288101). The authors are very grateful to the handling editor and the anonymous referees  for their helpful suggestions. 
\section*{Data Availability}
The manuscript has no associated data.
\section*{Declarations}
{\bf Conflicts of interest} \ The authors  state  that there is no conflict of interest.

\bibliographystyle{alpha}

\end{document}